\documentclass{amsart}
\usepackage{mathrsfs}
\usepackage{graphicx}
\usepackage{amsmath}
\usepackage{amscd}
\usepackage{cite}
\usepackage{hyperref}
\usepackage{epsfig}
\usepackage{subfigure}
\usepackage{caption}
\usepackage{tikz}
  \usepackage{float}
  \usepackage{amssymb}
\usepackage{amsfonts}
\usepackage[all,cmtip]{xy}
\usepackage{appendix}

\newtheorem{Example}{Example}[section]
\newtheorem{Definition}{Definition}[section]
\newtheorem{Theorem}{Theorem}[section]
\newtheorem{Theorem/Definition}{Theorem/Definition}[section]

\newtheorem{Lemma}{Lemma}[section]
\newtheorem{Corollary}{Corollary}[section]

\newcommand{\pd}{\partial}
\newcommand{\bC}{{\mathbb C}}

\newcommand{\bQ}{{\mathbb Q}}

\newcommand{\cA}{{\mathcal A}}

\newcommand{\cD}{{\mathcal D}}
\newcommand{\cE}{{\mathcal E}}
\newcommand{\cF}{{\mathcal F}}
\newcommand{\cG}{{\mathcal G}}

\newcommand{\cM}{{\mathcal M}}

 \newcommand{\cS}{{\mathcal S}}

\newcommand{\half}{\frac{1}{2}}
\newcommand{\cV}{{\mathcal V}}
\newcommand{\cW}{{\mathcal W}}

\newcommand{\Mbar}{\overline{\cM}}

\newcommand{\wF}{{\widehat F}}
\newcommand{\wcF}{{\widehat{\mathcal{F}}}}
\newcommand{\tF}{{\widetilde F}}
\newcommand{\tpd}{{\tilde{\pd}}}
\newcommand{\tD}{{\widetilde{D}}}

\newcommand{\be}{\begin{equation}}
\newcommand{\ee}{\end{equation}}
\newcommand{\bea}{\begin{eqnarray}}
\newcommand{\ben}{\begin{eqnarray*}}
\newcommand{\een}{\end{eqnarray*}}
\newcommand{\eea}{\end{eqnarray}}

\DeclareMathOperator{\Aut}{Aut}
\DeclareMathOperator{\Id}{id}

\DeclareMathOperator{\val}{val}

\usepackage{color}

\definecolor{yellow}{rgb}{1,1,0}
\definecolor{orange}{rgb}{1,.7,0}
\definecolor{red}{rgb}{1,0,0}
\definecolor{green}{rgb}{0,1,1}
\definecolor{white}{rgb}{1,1,1}

\definecolor{A}{rgb}{.75,1,.75}

\theoremstyle{remark}
\newtheorem{Remark}{Remark}[section]

\begin{document}

\newtheorem{myDef}{Definition}
\newtheorem{thm}{Theorem}
\newtheorem{eqn}{equation}

\title[Fourier-Like Transforms of Stable Graphs and HAE]
{Fourier-Like Transforms of Stable Graphs and Holomorphic Anomaly Equations}

\author{Zhiyuan Wang}
\address{Department of Mathematical Sciences\\
Tsinghua University\\Beijing, 100084, China}
\email{zhiyuan-14@mails.tsinghua.edu.cn}

\author{Jian Zhou}
\address{Department of Mathematical Sciences\\
Tsinghua University\\Beijing, 100084, China}
\email{jianzhou@mail.tsinghua.edu.cn}

\begin{abstract}
In this paper we develop a theory of Fourier-like transforms on the space of stable graphs.
In particular,
we introduce a duality theory of stable graphs.
As an application,
we derive the holomorphic anomaly equations for general propagators
in the work of Aganagic, Bouchard and Klemm.
\end{abstract}

\maketitle

\tableofcontents

\section{Introduction}

Inspired by the physics literatures \cite{abk, bcov1, bcov2, emo, ey, gkmw}
on holomorphic anomaly equations,
the authors developed a formalism called {\em the abstract quantum field theory
and its realizations} in a previous work \cite{wz}.
In that work, our point of view was different with the physicists'.
In the physics literatures, physicists solved the holomorphic
anomaly equation recursively and reformulated the solutions
in terms of Feynman sums over stable graphs and Feynman rules.
We took a different strategy that reversed the steps.
Step 1,
we study the summations over stable graphs from an abstract point of view,
ignoring the Feynman rules and focusing on the combinatorial aspects
of the diagrammatics of stable graphs.
This part of the theory is called the {\em abstract quantum field theory}
and we derive natural quadratic recursion relations in this setting.
Step 2, we consider the following way to obtain Feynman rules
that assign the contributions of a stable graph.
The input is a sequence of holomorphic functions $F_{g,n}(t)$ and a formal variable $\kappa$.
For the vertices we take $F_{g,n}(t)$,
and for the propagator we take $\kappa$.
The Feynman rules enable us to construct a new sequence of functions $\wF_{g,n}(t,\kappa)$
as summations over stable graphs.
We call this a {\em realization of the abstract QFT}.
Step 3,
we proved that there are natural quadratic recursions for $\pd_\kappa\wF_{g,n}$
when $F_{g,n}(t) = \pd_t^n F_{g,0}(t)$ and
$\kappa$ is taken to be a function in $t$ that satisfies suitable conditions.
This aspect of the theory will be called the {\em realization of the quadratic recursion relations}.

One advantage of our formalism is that since we have taken a general treatment in Step 1,
we have the freedom of many different choices in Step 2 and Step 3.
This makes it possible to apply our formalism to various problems that
involve summation over stable graphs.
We presented some examples in \cite{wz},
including a construction of the nonholomorphic free energy and
a derivation of the holomorphic anomaly equation as an example.
We also indicated some relationship with Eynard-Orantin topological recursion \cite{eo}
and quantum spectral curves \cite{gs}.
Another application was given in \cite{wz2}.
In that work we took the Harer-Zagier formula \cite{hz, pe} as input data,
and solved the problems of computations of
the orbifold Euler characteristics of $\Mbar_{g,n}$.

In particular,
using Step 1 and Step 2 in our formalism,
if we are given a field theory with free energy $F_g(t)$
together with a propagator $\kappa(t)$,
we can construct a new field theory with free energy $\wF_g(t,\kappa)$.
If we want to understand this as constructing  a transformation on the `moduli space of theories',
we need consider the reversed procedure of reconstruct $F_g(t)$ from $\wF_g(t,\kappa)$
in a similar fashion.
In other words,
we want to express $F_{g,n}(t)$ as a Feynman sum over some contributions over
some stable graphs.
We solve this problem in this work.
The corresponding Feynman rules for the contribution of a vertex should be
defined by some derivatives of $\wF_g(t,\kappa)$.
One does not expect the derivatives are simply $\pd_t^n \wF_g(t,\kappa)$,
because even though $t$ is the ``flat coordinate" for the theory with free energies given by
$F_g(t)$,
it may not be also the flat coordinate for the new theory whose free energies are given by
$\wF_g(t,\kappa)$.
The correct replacement for $\pd_t$  should be some kind of ``covariant derivative" $D_t$.
It turns out that this $D_t$ can be taken to be the operator $D_t$ that has already appeared
in the physics literatures and in \cite{wz}.
See also \S \ref{sec:KleZas}.
A priori, the propagator for this dual Feynman rule should be determined by $F_g(t)$ and $\kappa$.
We find it to be simply $-\kappa$.

Parallel to our original formalism in \cite{wz},
we develop a dual formalism in this work.
It follow the same steps.

Step $1^\vee$.
We first develop what we call the {\em dual abstract QFT} and its realizations.
For any stable graph $\Gamma$,
we draw another stable graph $\Gamma^\vee$.
It is the same graph as $\Gamma$ but drawn in dotted lines.
We will call such a graph a dotted stable graph.
Next, we will develop a way to express each dotted stable graph
as a linear combination of ordinary stable graphs.
For a dotted stable vertex $\Gamma^\vee$ of genus $g$ with $n$ external edges,
we assign a formal sum:
\be
\Gamma^\vee = n!\cdot \sum_{\Gamma' \in \cG_{g,n}^c} \frac{1}{|\Aut (\Gamma')|} \Gamma',
\ee
where $\cG_{g,n}^c$ is the set of connected ordinary stable graphs of genus $g$ with
$n$ external edges.
For a general dotted stable graph $\Gamma^\vee$,
similar expressions are derived by  cutting $\Gamma^\vee$ into dotted stable vertices,
writing them into sums of ordinary stable graphs,
and then suitably glue back.
(For details, see \S \ref{sec:Dotted}.)
In this way,
we express every dotted stable graph in terms of a linear combination of ordinary
stable graph.
The linear map defined by $\Gamma \to \Gamma^\vee$ is called the {\em duality map}.
Our duality theorem (Theorem \ref{thm-dual-FE} ) shows that this can be done reversely in a particular simple fashion:
The formulas for both of the expressions are exactly the same up to some signs determined by
the numbers of internal edges.
(This is the underlying reason why the propagator for the   dual theory should be $-\kappa$.)
Furthermore,
the duality map is an involution (cf. Theorem \ref{thm:Involution}).
We also define operators $K^\vee$, $\pd^\vee$ and $\cD^\vee$ acting on the space
spanned by the dotted stable graphs.
They are dual version of the operators $K$, $\pd$ and $\cD$ acting on the space
spanned by the ordinary stable graphs correspondingly.
We show in Theorem \ref{thm1} that
\be
\cD^\vee=\pd,\qquad \pd^\vee=\cD.
\ee
This explains that why $\pd_t$ in the original theory corresponds to $D_t$ in the dual theory.

In Step $2^\vee$,
we develop the duality between realizations of abstract QFT and realizations of
dual abstract QFT.
Since $F_{g,n}$ and $F_{g,n}^\vee$ are both realized as formal Gaussian integrals,
we formulate this duality as analog of the Fourier transform.
In Step $3^\vee$,
we use the duality derived in Step $2^\vee$ to derive an assumption (Independence Assumption)
that guarantees a derivation of the quadratic  recursion relations satisfied by $\wF_{g,n}$.

As an application, we will solve the problem of constructing
a nonholomorphic free energy and deriving
the holomorphic anomaly equation in the general case of
\be
\kappa = -\frac{1}{2\sqrt{-1}}\big(Im(\tau)\big)^{-1}+\cE(t),
\ee
with $\cE\not=0$.
For the $\cE = 0$ case, see \cite{wz}.
We also make generalizations that treats the case of higher dimensional state-space
when  the propagators $\Delta^{IJ}$ is of the form
\be \label{intro-eqn:Delta}
\Delta^{IJ} = -\frac{1}{2\sqrt{-1}} ((\tau-\bar{\tau})^{-1})^{IJ} +\cE^{IJ},
\ee
where $\cE^{IJ}$ is a holomorphic function.

Even though our ultimate goal is to derive the holomorphic anomaly equations,
our approach through the dual diagrammatics of stable graphs 
and the idea of transformations on the moduli space of field theory 
lead us to a theory of transformations at the level of abstract quantum field theories.
More precisely,
we introduce some transformations on the space spanned by 
stable graphs with suitable coefficient ring
such that our duality  theory is just a special case,
and furthermore, 
such transformations naturally form a commutative group.
As a consequence,
at the level of realizations of the abstract QFT, 
we obtain a commutative group of transformations.
Using the representations by formal Gaussian integrals,
such transformations by are realized by Fourier-like transforms.
Therefore,
we call the transformations at the level of stable graphs
Fourier-like transforms on the space of stable graphs.
The duality that we use to derive the holomorphic anomaly equations
is just a special case of such more general transforms.
We will look for their applications in future investigations.

Many results in \cite{wz} and this work can be 
generalized to fat graphs. 
We will report on such generalizations in a separate paper.

The rest of this paper is arranged as follows.
We recall in \S \ref{sec-pre} our formalism of abstract QFT, its realizations
and the corresponding quadratic recursion relations.
The duality theory for the abstract QFT is developed in \S \ref{sec:dual-diagrammatics}.
The realizations of dual abstract QFT is studied in \S \ref{sec:realization}.
We make generalizations   to multi-labelled graphs in  \S \ref{sec:higher-dim}.
The Fourier-like transforms on the space of stable graphs and their realizations
are discussed in \S \ref{sec:Fourier}.
In the final \S \ref{sec:HAE}
we present the application of the duality theory to holomorphic anomaly equations.

\section{Preliminaries}
\label{sec-pre}

In this section, we recall the formalism of abstract quantum field theory
and its realizations developed in \cite{wz}.

\subsection{Stable graphs and edge-cutting operators}
\label{sec-pre-def}

A stable graph $\Gamma$ is defined to be a graph
whose vertices are marked by nonnegative integers $g_v \geq 0$,
called the genus of the vertex,
such that all vertices are stable.
I.e.,
the valence of each vertex of genus $0$ is at least three,
and the valence of each vertex of genus $1$ is at least one.
The genus of a stable graph $\Gamma$ is defined to be
\ben
g(\Gamma):=h^1(\Gamma)+\sum_{v\in V(\Gamma)}g_v,
\een
where $h^1(\Gamma)$ is the number of loops in $\Gamma$,
and $V(\Gamma)$ is the set of all vertices of $\Gamma$.

Stable graphs can be used to describe the natural stratification
of the Deligne-Mumford moduli space
of stable curves $\overline\cM_{g,n}$ (cf. \cite{dm, kn}).
Let the group $S_n$ act on this moduli space by permutations of the marked points,
then the stratification of $\overline\cM_{g,n}/S_n$ is given by
\ben
\overline\cM_{g,n}/S_n=\bigsqcup_{\Gamma\in\cG_{g,n}^c}\cM_\Gamma,
\een
where $\cG_{g,n}^c$ is the set of all connected stable graphs of genus $g$
with $n$ external edges,
and $\cM_\Gamma$ is the moduli space of stable curves
whose dual graph is $\Gamma$.

Denote by $\cG_{g,n}$ be the set of all stable graphs of genus $g$
with $n$ external edges (not necessarily connected).
In \cite[\S 2]{wz},
we introduced some operators acting on the vector spaces
\ben
\cV := \bigoplus_{\substack{\Gamma \in \cG_{g,n}\\ 2g-2+n > 0}} \bQ\Gamma,
\een
namely the `edge-cutting operator' and `edge-adding operators'.

The edge-cutting operator $K$ is defined to be
\ben
K:\cV \to \cV, \qquad \Gamma \in \cG_{g,n}\mapsto
K(\Gamma)= \sum_{e\in E(\Gamma)} \Gamma_e,
\een
where $E(\Gamma)$ is the set of internal edges of $\Gamma$,
and $\Gamma_e$ is obtained from $\Gamma$ by cutting the internal edge $e$.
If $\Gamma$ is a stable graph without internal edges,
then $K(\Gamma)$ is assigned to be $0$.

The edge-adding operator $\pd$ is defined to be
\ben
\pd:\cV \to \cV,  \qquad \Gamma \in \cG_{g,n}\mapsto
\pd(\Gamma)= \sum_{e\in E(\Gamma)} \Gamma_e '+
\sum_{v\in V(\Gamma)}\Gamma_v,
\een
where $\Gamma_e '$ is obtained from $\Gamma$ by breaking up an internal edge and
inserting a trivalent vertex of genus $0$,
and $\Gamma_v$ is obtained from $\Gamma$ by
attaching an additional external edge to the vertex $v\in V(\Gamma)$.

Another edge-adding operator $\gamma$ is defined to be
\ben
\pd:\cV \to \cV, \qquad \Gamma \in \cG_{g,n}\mapsto
\pd(\Gamma)= \sum_{e\in E^{ext}(\Gamma)} \Gamma_e '',
\een
where $E^{ext}(\Gamma)$ is the set of external edges of $\Gamma$,
and $\Gamma_e ''$ is obtained form $\Gamma$ by attaching
a trivalent vertex of genus $0$ to the external edge $e$.
If $\Gamma$ is a stable graph $\Gamma$ without external edges,
then we assign $\gamma(\Gamma)=0$.

Define $\cD:=\pd+\gamma$, then $\cD$ preserves the subspace
\ben
\cV^c := \bigoplus_{\substack{\Gamma \in \cG^c_{g,n}\\ 2g-2+n > 0}} \bQ\Gamma
\quad \subset \cV.
\een

\begin{Example}
Here we give some examples of these operators.

\begin{flalign*}
\begin{tikzpicture}
\node [align=center,align=center] at (-0.4,0) {$K$};
\draw (1,0) circle [radius=0.2];
\draw (0.4,0) circle [radius=0.2];
\draw (0.6,0)--(0.8,0);
\draw (1.16,0.1) .. controls (1.5,0.2) and (1.5,-0.2) ..  (1.16,-0.1);
\draw (0.24,0.1) .. controls (-0.1,0.2) and (-0.1,-0.2) ..  (0.24,-0.1);
\node [align=center,align=center] at (1,0) {$0$};
\node [align=center,align=center] at (0.4,0) {$0$};
\node [align=center,align=center] at (1.8,0) {$=2$};
\draw (3.3,0) circle [radius=0.2];
\draw (2.7,0) circle [radius=0.2];
\draw (2.9,0)--(3.1,0);
\draw (3.46,0.1)--(3.7,0.15);
\draw (3.46,-0.1)--(3.7,-0.15);
\draw (2.54,0.1) .. controls (2.2,0.2) and (2.2,-0.2) ..  (2.54,-0.1);
\node [align=center,align=center] at (3.3,0) {$0$};
\node [align=center,align=center] at (2.7,0) {$0$};
\node [align=center,align=center] at (4,0) {$+$};
\draw (5.7,0) circle [radius=0.2];
\draw (4.7,0) circle [radius=0.2];
\draw (4.9,0)--(5.1,0);
\draw (5.3,0)--(5.5,0);
\draw (5.86,0.1) .. controls (6.2,0.2) and (6.2,-0.2) ..  (5.86,-0.1);
\draw (4.54,0.1) .. controls (4.2,0.2) and (4.2,-0.2) ..  (4.54,-0.1);
\node [align=center,align=center] at (5.7,0) {$0$};
\node [align=center,align=center] at (4.7,0) {$0$};
\end{tikzpicture},&&
\end{flalign*}

\begin{flalign*}
\begin{tikzpicture}
\node [align=center,align=center] at (0.4,0) {$K$};
\draw (1,0) circle [radius=0.2];
\draw (1.2,0)--(1.4,0);
\draw (1.16,0.1)--(1.44,0.1);
\draw (1.16,-0.1)--(1.44,-0.1);
\draw (1.6,0) circle [radius=0.2];
\node [align=center,align=center] at (1,0) {$0$};
\node [align=center,align=center] at (1.6,0) {$1$};
\node [align=center,align=center] at (2.2,0) {$=3$};
\draw (3.1,0) circle [radius=0.2];
\draw (3.28,0.07)--(3.52,0.07);
\draw (3.28,-0.07)--(3.52,-0.07);
\draw (2.6,0)--(2.9,0);
\draw (3.9,0)--(4.2,0);
\draw (3.7,0) circle [radius=0.2];
\node [align=center,align=center] at (3.1,0) {$0$};
\node [align=center,align=center] at (3.7,0) {$1$};
\end{tikzpicture};&&
\end{flalign*}

\begin{flalign*}
\begin{tikzpicture}
\node [align=center,align=center] at (-0.4,0) {$\pd$};
\draw (1,0) circle [radius=0.2];
\draw (0.4,0) circle [radius=0.2];
\draw (0.6,0)--(0.8,0);
\draw (1.16,0.1) .. controls (1.5,0.2) and (1.5,-0.2) ..  (1.16,-0.1);
\draw (0.24,0.1) .. controls (-0.1,0.2) and (-0.1,-0.2) ..  (0.24,-0.1);
\node [align=center,align=center] at (1,0) {$0$};
\node [align=center,align=center] at (0.4,0) {$0$};
\node [align=center,align=center] at (1.8,0) {$=2$};
\draw (3.2,0) circle [radius=0.2];
\draw (2.6,0) circle [radius=0.2];
\draw (2.8,0)--(3,0);
\draw (2.6,0.2)--(2.6,0.4);
\draw (3.36,0.1) .. controls (3.7,0.2) and (3.7,-0.2) ..  (3.36,-0.1);
\draw (2.44,0.1) .. controls (2.1,0.2) and (2.1,-0.2) ..  (2.44,-0.1);
\node [align=center,align=center] at (2.6,0) {$0$};
\node [align=center,align=center] at (3.2,0) {$0$};
\node [align=center,align=center] at (4,0) {$+2$};
\draw (5.4,0) circle [radius=0.2];
\draw (4.8,0) circle [radius=0.2];
\draw (5,0)--(5.2,0);
\draw (4.64,0.1) .. controls (4.3,0.2) and (4.3,-0.2) ..  (4.64,-0.1);
\draw (5.58,0.07)--(5.82,0.07);
\draw (5.58,-0.07)--(5.82,-0.07);
\draw (6.2,-0)--(6.5,0);
\draw (6,0) circle [radius=0.2];
\node [align=center,align=center] at (5.4,0) {$0$};
\node [align=center,align=center] at (4.8,0) {$0$};
\node [align=center,align=center] at (6,0) {$0$};
\node [align=center,align=center] at (7,0) {$+$};
\draw (8.4,0) circle [radius=0.2];
\draw (7.8,0) circle [radius=0.2];
\draw (9,0) circle [radius=0.2];
\draw (8,0)--(8.2,0);
\draw (8.6,0)--(8.8,0);
\draw (8.4,0.2)--(8.4,0.4);
\draw (9.16,0.1) .. controls (9.5,0.2) and (9.5,-0.2) ..  (9.16,-0.1);
\draw (7.64,0.1) .. controls (7.3,0.2) and (7.3,-0.2) ..  (7.64,-0.1);
\node [align=center,align=center] at (8.4,0) {$0$};
\node [align=center,align=center] at (7.8,0) {$0$};
\node [align=center,align=center] at (9,0) {$0$};
\end{tikzpicture},&&
\end{flalign*}

\begin{flalign*}
\begin{tikzpicture}
\node [align=center,align=center] at (0.4,0) {$\pd$};
\draw (1,0) circle [radius=0.2];
\draw (1.2,0)--(1.4,0);
\draw (1.16,0.1)--(1.44,0.1);
\draw (1.16,-0.1)--(1.44,-0.1);
\draw (1.6,0) circle [radius=0.2];
\node [align=center,align=center] at (1,0) {$0$};
\node [align=center,align=center] at (1.6,0) {$1$};
\node [align=center,align=center] at (2.1,0) {$=$};
\draw (3,0) circle [radius=0.2];
\draw (3.2,0)--(3.4,0);
\draw (2.5,0)--(2.8,0);
\draw (3.16,0.1)--(3.44,0.1);
\draw (3.16,-0.1)--(3.44,-0.1);
\draw (3.6,0) circle [radius=0.2];
\node [align=center,align=center] at (3,0) {$0$};
\node [align=center,align=center] at (3.6,0) {$1$};
\node [align=center,align=center] at (4.1,0) {$+$};
\draw (4.7,0) circle [radius=0.2];
\draw (4.9,0)--(5.1,0);
\draw (5.5,0)--(5.8,0);
\draw (4.86,0.1)--(5.14,0.1);
\draw (4.86,-0.1)--(5.14,-0.1);
\draw (5.3,0) circle [radius=0.2];
\node [align=center,align=center] at (4.7,0) {$0$};
\node [align=center,align=center] at (5.3,0) {$1$};
\node [align=center,align=center] at (6.2,0) {$+3$};
\draw (6.8,-0.2) circle [radius=0.2];
\draw (6.98,-0.13)--(7.42,-0.13);
\draw (6.98,-0.27)--(7.42,-0.27);
\draw (7.6,-0.2) circle [radius=0.2];
\draw (7.2,0.25) circle [radius=0.2];
\draw (6.94,-0.06)--(7.06,0.11);
\draw (7.46,-0.06)--(7.34,0.11);
\node [align=center,align=center] at (6.8,-0.2) {$0$};
\node [align=center,align=center] at (7.6,-0.2) {$1$};
\node [align=center,align=center] at (7.2,0.25) {$0$};
\draw (7.2,0.45)--(7.2,0.65);
\end{tikzpicture};&&
\end{flalign*}

\begin{flalign*}
\begin{tikzpicture}
\node [align=center,align=center] at (0.2,0) {$\gamma$};
\draw (0.9,0) circle [radius=0.2];
\draw (1.06,0.1)--(1.3,0.15);
\draw (1.06,-0.1)--(1.3,-0.15);
\draw (0.74,0.1) .. controls (0.4,0.2) and (0.4,-0.2) ..  (0.74,-0.1);
\node [align=center,align=center] at (0.9,0) {$0$};
\node [align=center,align=center] at (1.8,0) {$=2$};
\draw (2.7,0) circle [radius=0.2];
\draw (3.3,0) circle [radius=0.2];
\draw (2.9,0)--(3.1,0);
\draw (2.7,0.2)--(2.7,0.4);
\draw (3.47,0.1)--(3.7,0.15);
\draw (3.47,-0.1)--(3.7,-0.15);
\draw (2.54,0.1) .. controls (2.2,0.2) and (2.2,-0.2) ..  (2.54,-0.1);
\node [align=center,align=center] at (2.7,0) {$0$};
\node [align=center,align=center] at (3.3,0) {$0$};
\end{tikzpicture}.&&
\end{flalign*}

\end{Example}

\subsection{Quadratic recursion relation for the abstract free energy}

In this subsection let us recall the abstract quantum field theory
and various types of recursion relations developed in \cite[\S 2]{wz}.

The abstract free energy of genus $g$ for the abstract quantum field theory
is defined to be the linear combination
\be\label{pre-abs-FE}
\widehat{\cF}_g:=\sum_{\Gamma\in\cG_{g,0}^c}\frac{1}{|\Aut(\Gamma)|}\Gamma,
\quad g\geq 2,
\ee
and the abstract $n$-point function of genus $g$ is defined to be
\be\label{pre-abs-npt}
\widehat{\cF}_{g,n}:=\sum_{\Gamma\in\cG_{g,n}^c}\frac{1}{|\Aut(\Gamma)|}\Gamma,
\quad 2g-2+n>0.
\ee
For example,
\begin{flalign*}
\begin{tikzpicture}
\node [align=center,align=center] at (0.3,0) {$\widehat{\cF}_{0,3}=\frac{1}{6}$};
\draw (1.6,0) circle [radius=0.2];
\draw (1.1,0)--(1.4,0);
\draw (1.76,0.1)--(2,0.15);
\draw (1.76,-0.1)--(2,-0.15);
\node [align=center,align=center] at (1.6,0) {$0$};
\end{tikzpicture},&&
\end{flalign*}

\begin{flalign*}
\begin{tikzpicture}
\node [align=center,align=center] at (0.3-0.2,0) {$\widehat{\cF}_{1,1}=$};
\draw (1,0) circle [radius=0.2];
\draw (1.2,0)--(1.5,0);
\node [align=center,align=center] at (1,0) {$1$};
\node [align=center,align=center] at (2,0) {$+\frac{1}{2}$};
\draw (1+1.8,0) circle [radius=0.2];
\draw (1.2+1.8,0)--(1.5+1.8,0);
\draw (0.84+1.8,0.1) .. controls (0.5+1.8,0.2) and (0.5+1.8,-0.2) ..  (0.84+1.8,-0.1);
\node [align=center,align=center] at (1+1.8,0) {$0$};
\end{tikzpicture},&&
\end{flalign*}

\begin{flalign*}
\begin{tikzpicture}
\node [align=center,align=center] at (0.1+0.4,0) {$\widehat{\cF}_{2}=$};
\draw (1+0.3,0) circle [radius=0.2];
\node [align=center,align=center] at (1+0.3,0) {$2$};
\node [align=center,align=center] at (1.6+0.2,0) {$+\frac{1}{2}$};
\draw (1+1.4+0.2,0) circle [radius=0.2];
\draw (0.84+1.4+0.2,0.1) .. controls (0.5+1.4+0.2,0.2) and (0.5+1.4+0.2,-0.2) ..  (0.84+1.4+0.2,-0.1);
\node [align=center,align=center] at (1+1.4+0.2,0) {$1$};
\node [align=center,align=center] at (3+0.2,0) {$+\frac{1}{2}$};
\draw (1+2.6+0.2,0) circle [radius=0.2];
\draw (1.2+2.6+0.2,0)--(1.4+2.6+0.2,0);
\draw (1.6+2.6+0.2,0) circle [radius=0.2];
\node [align=center,align=center] at (1+2.6+0.2,0) {$1$};
\node [align=center,align=center] at (1.6+2.6+0.2,0) {$1$};
\node [align=center,align=center] at (5,0) {$+\frac{1}{8}$};
\draw (1+4.8,0) circle [radius=0.2];
\draw (0.84+4.8,0.1) .. controls (0.5+4.8,0.2) and (0.5+4.8,-0.2) ..  (0.84+4.8,-0.1);
\draw (1.16+4.8,0.1) .. controls (1.5+4.8,0.2) and (1.5+4.8,-0.2) ..  (1.16+4.8,-0.1);
\node [align=center,align=center] at (1+4.8,0) {$0$};
\node [align=center,align=center] at (6.6,0) {$+\frac{1}{2}$};
\draw (1+6.8,0) circle [radius=0.2];
\draw (0.4+6.8,0) circle [radius=0.2];
\draw (0.6+6.8,0)--(0.8+6.8,0);
\draw (1.16+6.8,0.1) .. controls (1.5+6.8,0.2) and (1.5+6.8,-0.2) ..  (1.16+6.8,-0.1);
\node [align=center,align=center] at (1+6.8,0) {$0$};
\node [align=center,align=center] at (0.4+6.8,0) {$1$};
\node [align=center,align=center] at (8.6,0) {$+\frac{1}{8}$};
\draw (1+9,0) circle [radius=0.2];
\draw (0.4+9,0) circle [radius=0.2];
\draw (0.6+9,0)--(0.8+9,0);
\draw (1.16+9,0.1) .. controls (1.5+9,0.2) and (1.5+9,-0.2) ..  (1.16+9,-0.1);
\draw (0.24+9,0.1) .. controls (-0.1+9,0.2) and (-0.1+9,-0.2) ..  (0.24+9,-0.1);
\node [align=center,align=center] at (1+9,0) {$0$};
\node [align=center,align=center] at (0.4+9,0) {$0$};
\node [align=center,align=center] at (10.6+0.2,0) {$+\frac{1}{12}$};
\draw (1+10.2+0.2,0) circle [radius=0.2];
\draw (1.2+10.2+0.2,0)--(1.4+10.2+0.2,0);
\draw (1.16+10.2+0.2,0.1)--(1.44+10.2+0.2,0.1);
\draw (1.16+10.2+0.2,-0.1)--(1.44+10.2+0.2,-0.1);
\draw (1.6+10.2+0.2,0) circle [radius=0.2];
\node [align=center,align=center] at (1+10.2+0.2,0) {$0$};
\node [align=center,align=center] at (1.6+10.2+0.2,0) {$0$};
\end{tikzpicture}.&&
\end{flalign*}

Using the operators $\cD$ and $K$,
we are able to derive the following recursion relations for
$\wcF_g$ and $\wcF_{g,n}$:

\begin{Lemma}\label{lem-original-D}(\cite[Lemma 2.1]{wz})
For $2g-2+n>0$, we have
\be
\cD\widehat{\cF}_{g,n}=(n+1)\widehat{\cF}_{g,n+1}.
\ee
\end{Lemma}

\begin{Theorem}\label{thm-original-rec}(\cite[Theorem 2.2]{wz})
For $2g-2+n>0$, we have
\be\label{eq-thm2}
K\widehat{\cF}_{g,n}=\frac{1}{2}(\cD \cD\widehat{\cF}_{g-1,n}
+\sum_{\substack{g_1+g_2=g,\\n_1+n_2=n}}
\cD\widehat{\cF}_{g_1,n_1}\cD\widehat{\cF}_{g_2,n_2}).
\ee
In particular, by taking $n=0$ we get a quadratic recursion relation
for the abstract free energy for $g\geq 2$:
\be\label{thm-free}
K\widehat{\cF}_g=\frac{1}{2}(\cD \pd\widehat{\cF}_{g-1}
+\sum_{r=1}^{g-1}\pd\widehat{\cF}_{r}\pd\widehat{\cF}_{g-r}).
\ee
Here we use the convention
\be\label{convention-1}
\begin{split}
&\pd\wcF_1=\cD\widehat{\cF}_{1}:=\widehat{\cF}_{1,1},\\
&\cD\widehat{\cF}_{0,2}:=3\widehat{\cF}_{0,3},\\
&\cD \cD\widehat{\cF}_{0,1}:=6\widehat{\cF}_{0,3},
\end{split}
\ee
by formally applying Lemma \ref{lem-original-D}.

\end{Theorem}

This recursion can also be written in the following way:
\begin{Theorem}\label{original-rec-2}
(\cite[Theorem 2.1]{wz})
For $2g-2+n>0$, we have
\be
K\widehat{\cF}_{g,n}=\frac{1}{2}\biggl[(n+2)(n+1)\widehat{\cF}_{g-1,n+2}
+\sum_{\substack{g_1+g_2=g,\\n_1+n_2=n+2\\n_1\geq 1,n_2\geq 1}}
n_1n_2\widehat{\cF}_{g_1,n_1}\widehat{\cF}_{g_2,n_2}\biggr],
\ee
where the sum in the right-hand-side is over all stable cases.
\end{Theorem}

\begin{Example}
We have
\begin{flalign*}
\begin{split}
\begin{tikzpicture}
\node [align=center,align=center] at (-1.1,0) {$K\wcF_2=\frac{1}{2}$};
\draw (0,0) circle [radius=0.2];
\node [align=center,align=center] at (0,0) {$1$};
\draw (0.16,0.1)--(0.5,0.15);
\draw (0.16,-0.1)--(0.5,-0.15);
\node [align=center,align=center] at (1,0) {$+\frac{1}{2}\biggl($};
\draw (1.7,0) circle [radius=0.2];
\node [align=center,align=center] at (1.7,0) {$1$};
\draw (3,0) circle [radius=0.2];
\node [align=center,align=center] at (3,0) {$1$};
\draw (1.9,0)--(2.2,0);
\draw (2.5,0)--(2.8,0);
\node [align=center,align=center] at (3.8,0) {$\biggr)+\frac{1}{4}$};
\draw (4.9,0) circle [radius=0.2];
\node [align=center,align=center] at (4.9,0) {$0$};
\draw (4.74,0.1) .. controls (4.4,0.2) and (4.4,-0.2) ..  (4.74,-0.1);
\draw (5.06,0.1)--(5.4,0.15);
\draw (5.06,-0.1)--(5.4,-0.15);
\node [align=center,align=center] at (6,0) {$+\frac{1}{2}\biggl($};
\draw (7.1-0.4,0) circle [radius=0.2];
\node [align=center,align=center] at (7.1-0.4,0) {$1$};
\draw (7.3-0.4,0)--(7.6-0.4,0);
\draw (8.4-0.4,0) circle [radius=0.2];
\node [align=center,align=center] at (8.4-0.4,0) {$0$};
\draw (7.9-0.4,0)--(8.2-0.4,0);
\draw (8.56-0.4,0.1) .. controls (8.9-0.4,0.2) and (8.9-0.4,-0.2) ..  (8.56-0.4,-0.1);
\node [align=center,align=center] at (8.6,0) {$\biggr)$};
\end{tikzpicture}
\\
\begin{tikzpicture}
\node [align=center,align=center] at (-0.6,0) {$+\frac{1}{2}$};
\draw (0,0) circle [radius=0.2];
\node [align=center,align=center] at (0,0) {$1$};
\draw (0.6,0) circle [radius=0.2];
\node [align=center,align=center] at (0.6,0) {$0$};
\draw (0.2,0)--(0.4,0);
\draw (0.76,0.1)--(1.1,0.15);
\draw (0.76,-0.1)--(1.1,-0.15);
\node [align=center,align=center] at (1.6,0) {$+\frac{1}{4}$};
\draw (2.4,0) circle [radius=0.2];
\node [align=center,align=center] at (2.4,0) {$0$};
\draw (3,0) circle [radius=0.2];
\node [align=center,align=center] at (3,0) {$0$};
\draw (2.24,0.1) .. controls (1.9,0.2) and (1.9,-0.2) ..  (2.24,-0.1);
\draw (2.6,0)--(2.8,0);
\draw (3.16,0.1)--(3.5,0.15);
\draw (3.16,-0.1)--(3.5,-0.15);
\node [align=center,align=center] at (4.2,0) {$+\frac{1}{8}\biggl($};
\draw (5.1,0) circle [radius=0.2];
\node [align=center,align=center] at (5.1,0) {$0$};
\draw (6.4,0) circle [radius=0.2];
\node [align=center,align=center] at (6.4,0) {$0$};
\draw (5.3,0)--(5.6,0);
\draw (5.9,0)--(6.2,0);
\draw (4.94,0.1) .. controls (4.6,0.2) and (4.6,-0.2) ..  (4.94,-0.1);
\draw (6.56,0.1) .. controls (6.9,0.2) and (6.9,-0.2) ..  (6.56,-0.1);
\node [align=center,align=center] at (7.4,0) {$\biggr)+\frac{1}{4}$};
\draw (8.3+0.3,0) circle [radius=0.2];
\node [align=center,align=center] at (8.3+0.3,0) {$0$};
\draw (8.9+0.3,0) circle [radius=0.2];
\node [align=center,align=center] at (8.9+0.3,0) {$0$};
\draw (7.8+0.3,0)--(8.1+0.3,0);
\draw (9.1+0.3,0)--(9.4+0.3,0);
\draw (8.78,0.07)--(9.02,0.07);
\draw (8.78,-0.07)--(9.02,-0.07);
\end{tikzpicture}.
\end{split}
\end{flalign*}
Then using
\begin{flalign*}
\begin{tikzpicture}
\node [align=center,align=center] at (-0.3,0) {$\widehat{\cF}_{1,2}=\frac{1}{2}$};
\draw (1,0) circle [radius=0.2];
\draw (1.2,0)--(1.5,0);
\draw (0.5,0)--(0.8,0);
\node [align=center,align=center] at (1,0) {$1$};
\node [align=center,align=center] at (2,0) {$+\frac{1}{4}$};
\draw (1+1.8,0) circle [radius=0.2];
\draw (1.17+1.8,0.1)--(1.4+1.8,0.15);
\draw (1.17+1.8,-0.1)--(1.4+1.8,-0.15);
\draw (0.84+1.8,0.1) .. controls (0.5+1.8,0.2) and (0.5+1.8,-0.2) ..  (0.84+1.8,-0.1);
\node [align=center,align=center] at (1+1.8,0) {$0$};
\node [align=center,align=center] at (3.7,0) {$+\frac{1}{2}$};
\draw (1+3.9,0) circle [radius=0.2];
\draw (0.4+3.9,0) circle [radius=0.2];
\draw (1.17+3.9,0.1)--(1.4+3.9,0.15);
\draw (1.17+3.9,-0.1)--(1.4+3.9,-0.15);
\draw (0.6+3.9,0)--(0.8+3.9,0);
\node [align=center,align=center] at (1+3.9,0) {$0$};
\node [align=center,align=center] at (0.4+3.9,0) {$1$};
\node [align=center,align=center] at (5.8,0) {$+\frac{1}{4}$};
\draw (1+6.2,0) circle [radius=0.2];
\draw (0.4+6.2,0) circle [radius=0.2];
\draw (1.17+6.2,0.1)--(1.4+6.2,0.15);
\draw (1.17+6.2,-0.1)--(1.4+6.2,-0.15);
\draw (0.6+6.2,0)--(0.8+6.2,0);
\node [align=center,align=center] at (1+6.3-0.1,0) {$0$};
\node [align=center,align=center] at (0.4+6.3-0.1,0) {$0$};
\draw (0.24+6.3-0.1,0.1) .. controls (-0.1+6.3-0.1,0.2) and (-0.1+6.3-0.1,-0.2) ..  (0.24+6.3-0.1,-0.1);
\node [align=center,align=center] at (8.2-0.1,0) {$+\frac{1}{4}$};
\draw (1+8.1-0.1,0) circle [radius=0.2];
\draw (0.5+8.1-0.1,0)--(0.8+8.1-0.1,0);
\draw (1.18+8.1-0.1,0.07)--(1.42+8.1-0.1,0.07);
\draw (1.18+8.1-0.1,-0.07)--(1.42+8.1-0.1,-0.07);
\draw (1.8+8.1-0.1,-0)--(2.1+8.1-0.1,0);
\draw (1.6+8.1-0.1,0) circle [radius=0.2];
\node [align=center,align=center] at (1+8.1-0.1,0) {$0$};
\node [align=center,align=center] at (1.6+8.1-0.1,0) {$0$};
\end{tikzpicture},&&
\end{flalign*}
we can check:
\ben
K\wcF_2=\wcF_{1,2}+\frac{1}{2}\wcF_{1,1}\wcF_{1,1}
=\frac{1}{2}\biggl(\cD\cD\wcF_1+\cD\wcF_1\cD\wcF_1\biggr).
\een
\end{Example}

\subsection{A generalization to labelled graphs}
\label{sec-pre-label}

In this subsection, we recall the labelled graphs defined in \cite[\S 3]{wz}.
The abstract quantum field theory can be generalized to the case of labelled graphs.

First we fix a positive integer $N$.
For a stable graph $\Gamma\in\cG_{g,n}^c$, we label indices in
$\{1,2,\cdots,N\}$ on all `half edges' of $\Gamma$ as follows.
For an external edge, we label an index in $\{1,2,\cdots,N\}$;
and for an internal edge,
we label an index in $\{1,2,\cdots,N\}$ at each of its end points.

Then we generalize the operators $K$, $\pd$, $\cD$ as follows.
We define the edge-cutting operators $K_{ij}$ as the operator
that cuts off an internal edges with two labels $i$ and $j$,
then sums over all such internal edges.
The operator $\partial_{i}$ has two parts,
one is to add an external edge labelled by $i$,
and the other is to break up an internal edge and
attach a trivalent vertex of genus $0$
such that the new external edge has label $i$.
The operator $\gamma_{i}$ is to attach a trivalent vertex of genus $0$
to an external edge $e$,
and move the label of $e$ to be the label of one of the new external edge,
while the other new external edge has label $i$.
The two labels of the new internal edge can be chosen arbitrarily.
Also we define $\cD_i=\pd_i+\gamma_i$.

\begin{Example}
For the case $N=2$, we present some examples of the operators
$K_{ij}$, $\pd_i$, and $\gamma_i$.

\begin{flalign*}
\begin{tikzpicture}
\node [align=center,align=center] at (-0.5,0) {$K_{11}$};
\draw (1.2,0) circle [radius=0.2];
\draw (0.4,0) circle [radius=0.2];
\draw (0.6,0)--(1,0);
\draw (1.36,0.1) .. controls (1.7,0.2) and (1.7,-0.2) ..  (1.36,-0.1);
\draw (0.24,0.1) .. controls (-0.1,0.2) and (-0.1,-0.2) ..  (0.24,-0.1);
\node [align=center,align=center] at (1.2,0) {$0$};
\node [align=center,align=center] at (0.4,0) {$0$};
\node [above,align=center] at (0.6,0) {$1$};
\node [above,align=center] at (1,0) {$2$};
\node [above,align=center] at (0.24,0.1) {$1$};
\node [below,align=center] at (0.24,-0.1) {$1$};
\node [above,align=center] at (1.36,0.1) {$1$};
\node [below,align=center] at (1.36,-0.1) {$1$};
\draw (3.7,0) circle [radius=0.2];
\draw (2.9,0) circle [radius=0.2];
\draw (3.1,0)--(3.5,0);
\draw (2.5,0.15)--(2.74,0.1);
\draw (2.5,-0.15)--(2.74,-0.1);
\draw (3.86,0.1) .. controls (4.2,0.2) and (4.2,-0.2) ..  (3.86,-0.1);
\node [align=center,align=center] at (3.7,0) {$0$};
\node [align=center,align=center] at (2.9,0) {$0$};
\node [above,align=center] at (3.1,0) {$1$};
\node [above,align=center] at (3.5,0) {$2$};
\node [above,align=center] at (2.6,0.1) {$1$};
\node [below,align=center] at (2.6,-0.1) {$1$};
\node [above,align=center] at (3.86,0.1) {$1$};
\node [below,align=center] at (3.86,-0.1) {$1$};
\node [align=center,align=center] at (2,0) {$=$};
\node [align=center,align=center] at (4.5,0) {$+$};
\draw (3.7+2.4,0) circle [radius=0.2];
\draw (2.9+2.4,0) circle [radius=0.2];
\draw (3.1+2.4,0)--(3.5+2.4,0);
\draw (3.86+2.4,0.1)--(4.1+2.4,0.15);
\draw (3.86+2.4,-0.1)--(4.1+2.4,-0.15);
\draw (2.74+2.4,0.1) .. controls (2.4+2.4,0.2) and (2.4+2.4,-0.2) ..  (2.74+2.4,-0.1);
\node [align=center,align=center] at (3.7+2.4,0) {$0$};
\node [align=center,align=center] at (2.9+2.4,0) {$0$};
\node [above,align=center] at (3.1+2.4,0) {$1$};
\node [above,align=center] at (3.5+2.4,0) {$2$};
\node [above,align=center] at (2.7+2.4,0.1) {$1$};
\node [below,align=center] at (2.7+2.4,-0.1) {$1$};
\node [above,align=center] at (3.96+2.4,0.1) {$1$};
\node [below,align=center] at (3.96+2.4,-0.1) {$1$};
\end{tikzpicture},&&
\end{flalign*}

\begin{flalign*}
\begin{tikzpicture}
\node [align=center,align=center] at (-0.5,0) {$K_{12}$};
\draw (1.2,0) circle [radius=0.2];
\draw (0.4,0) circle [radius=0.2];
\draw (0.6,0)--(1,0);
\draw (1.36,0.1) .. controls (1.7,0.2) and (1.7,-0.2) ..  (1.36,-0.1);
\draw (0.24,0.1) .. controls (-0.1,0.2) and (-0.1,-0.2) ..  (0.24,-0.1);
\node [align=center,align=center] at (1.2,0) {$0$};
\node [align=center,align=center] at (0.4,0) {$0$};
\node [above,align=center] at (0.6,0) {$1$};
\node [above,align=center] at (1,0) {$2$};
\node [above,align=center] at (0.24,0.1) {$1$};
\node [below,align=center] at (0.24,-0.1) {$1$};
\node [above,align=center] at (1.36,0.1) {$1$};
\node [below,align=center] at (1.36,-0.1) {$1$};
\draw (3.8,0) circle [radius=0.2];
\draw (2.8,0) circle [radius=0.2];
\draw (3.0,0)--(3.2,0);
\draw (3.4,0)--(3.6,0);
\draw (3.96,0.1) .. controls (4.3,0.2) and (4.3,-0.2) ..  (3.96,-0.1);
\draw (2.64,0.1) .. controls (2.3,0.2) and (2.3,-0.2) ..  (2.64,-0.1);
\node [align=center,align=center] at (3.8,0) {$0$};
\node [align=center,align=center] at (2.8,0) {$0$};
\node [above,align=center] at (3.1,0) {$1$};
\node [above,align=center] at (3.5,0) {$2$};
\node [above,align=center] at (2.64,0.1) {$1$};
\node [below,align=center] at (2.64,-0.1) {$1$};
\node [above,align=center] at (3.96,0.1) {$1$};
\node [below,align=center] at (3.96,-0.1) {$1$};
\node [align=center,align=center] at (2,0) {$=$};
\end{tikzpicture};&&
\end{flalign*}

\begin{flalign*}
\begin{tikzpicture}
\node [align=center,align=center] at (0.4,0) {$\partial_{1}$};
\draw (1.2,0) circle [radius=0.2];
\draw (1.4,0)--(1.7,0);
\draw (1.04,0.1) .. controls (0.7,0.2) and (0.7,-0.2) ..  (1.04,-0.1);
\node [align=center,align=center] at (1.2,0) {$0$};
\node [above,align=center] at (1.04,0.1) {$2$};
\node [below,align=center] at (1.04,-0.1) {$2$};
\node [align=center,above] at (1.7,0) {$1$};
\node [align=center,right] at (1.9,0) {$=$};
\draw (3,0) circle [radius=0.2];
\draw (3.17,0.1)--(3.4,0.15);
\draw (3.17,-0.1)--(3.4,-0.15);
\draw (2.84,0.1) .. controls (2.5,0.2) and (2.5,-0.2) ..  (2.84,-0.1);
\node [align=center,align=center] at (3,0) {$0$};
\node [above,align=center] at (2.84,0.1) {$2$};
\node [below,align=center] at (2.84,-0.1) {$2$};
\node [align=center,right] at (3.4,0.15) {$1$};
\node [align=center,right] at (3.4,-0.15) {$1$};
\node [align=center,right] at (3.8,0) {$+$};
\draw (5,0) circle [radius=0.2];
\draw (4.5,0)--(4.8,0);
\draw (5.18,0.07)--(5.62,0.07);
\draw (5.18,-0.07)--(5.62,-0.07);
\draw (6,0)--(6.3,0);
\draw (5.8,0) circle [radius=0.2];
\node [align=center,align=center] at (5,0) {$0$};
\node [align=center,align=center] at (5.8,0) {$0$};
\node [above,align=center] at (4.5,0) {$1$};
\node [above,align=center] at (6.3,0) {$1$};
\node [above,align=center] at (5.18,0.07) {$1$};
\node [above,align=center] at (5.62,0.07) {$2$};
\node [below,align=center] at (5.18,-0.07) {$1$};
\node [below,align=center] at (5.62,-0.07) {$2$};
\node [align=center,align=center] at (6.7,0) {$+$};
\draw (7.5,0) circle [radius=0.2];
\draw (7,0)--(7.3,0);
\draw (7.68,0.07)--(8.12,0.07);
\draw (7.68,-0.07)--(8.12,-0.07);
\draw (8.5,0)--(8.8,0);
\draw (8.3,0) circle [radius=0.2];
\node [align=center,align=center] at (7.5,0) {$0$};
\node [align=center,align=center] at (8.3,0) {$0$};
\node [above,align=center] at (7,0) {$1$};
\node [above,align=center] at (8.8,0) {$1$};
\node [above,align=center] at (7.68,0.07) {$2$};
\node [above,align=center] at (8.12,0.07) {$2$};
\node [below,align=center] at (7.68,-0.07) {$2$};
\node [below,align=center] at (8.12,-0.07) {$2$};
\node [align=center,align=center] at (9.2,0) {$+$};
\node [align=center,align=center] at (9.5,0) {$2$};
\draw (10.3,0) circle [radius=0.2];
\draw (9.8,0)--(10.1,0);
\draw (10.48,0.07)--(10.92,0.07);
\draw (10.48,-0.07)--(10.92,-0.07);
\draw (11.3,0)--(11.6,0);
\draw (11.1,0) circle [radius=0.2];
\node [align=center,align=center] at (10.3,0) {$0$};
\node [align=center,align=center] at (11.1,0) {$0$};
\node [above,align=center] at (9.8,0) {$1$};
\node [above,align=center] at (11.6,0) {$1$};
\node [above,align=center] at (10.48,0.07) {$1$};
\node [above,align=center] at (10.92,0.07) {$2$};
\node [below,align=center] at (10.48,-0.07) {$2$};
\node [below,align=center] at (10.92,-0.07) {$2$};
\end{tikzpicture};&&
\end{flalign*}

\begin{flalign*}
\begin{tikzpicture}
\node [align=center,align=center] at (0.5,0) {$\gamma_1$};
\draw (1.2,0) circle [radius=0.2];
\draw (1.4,0)--(1.7,0);
\draw (1.04,0.1) .. controls (0.7,0.2) and (0.7,-0.2) ..  (1.04,-0.1);
\node [align=center,align=center] at (1.2,0) {$0$};
\node [above,align=center] at (1.04,0.1) {$2$};
\node [below,align=center] at (1.04,-0.1) {$2$};
\node [align=center,above] at (1.7,0) {$1$};
\node [align=center,right] at (1.8,0) {$=$};
\draw (3.5,0) circle [radius=0.2];
\draw (2.8,0) circle [radius=0.2];
\draw (3,0)--(3.3,0);
\draw (3.66,0.1)--(3.9,0.15);
\draw (3.66,-0.1)--(3.9,-0.15);
\draw (2.64,0.1) .. controls (2.3,0.2) and (2.3,-0.2) ..  (2.64,-0.1);
\node [align=center,align=center] at (3.5,0) {$0$};
\node [align=center,align=center] at (2.8,0) {$0$};
\node [above,align=center] at (2.64,0.1) {$2$};
\node [below,align=center] at (2.64,-0.1) {$2$};
\node [above,align=center] at (3,0) {$1$};
\node [above,align=center] at (3.3,0) {$1$};
\node [align=center,right] at (3.9,0.15) {$1$};
\node [align=center,right] at (3.9,-0.15) {$1$};
\node [align=center,align=center] at (4.5,0) {$+$};
\draw (6,0) circle [radius=0.2];
\draw (5.3,0) circle [radius=0.2];
\draw (5.5,0)--(5.8,0);
\draw (6.16,0.1)--(6.4,0.15);
\draw (6.16,-0.1)--(6.4,-0.15);
\draw (5.14,0.1) .. controls (4.8,0.2) and (4.8,-0.2) ..  (5.14,-0.1);
\node [align=center,align=center] at (6,0) {$0$};
\node [align=center,align=center] at (5.3,0) {$0$};
\node [above,align=center] at (5.14,0.1) {$2$};
\node [below,align=center] at (5.14,-0.1) {$2$};
\node [above,align=center] at (5.5,0) {$1$};
\node [above,align=center] at (5.8,0) {$2$};
\node [align=center,right] at (6.4,0.15) {$1$};
\node [align=center,right] at (6.4,-0.15) {$1$};
\node [align=center,align=center] at (7,0) {$+$};
\draw (8.5,0) circle [radius=0.2];
\draw (7.8,0) circle [radius=0.2];
\draw (8,0)--(8.3,0);
\draw (8.66,0.1)--(8.9,0.15);
\draw (8.66,-0.1)--(8.9,-0.15);
\draw (7.64,0.1) .. controls (7.3,0.2) and (7.3,-0.2) ..  (7.64,-0.1);
\node [align=center,align=center] at (8.5,0) {$0$};
\node [align=center,align=center] at (7.8,0) {$0$};
\node [above,align=center] at (7.64,0.1) {$2$};
\node [below,align=center] at (7.64,-0.1) {$2$};
\node [above,align=center] at (8,0) {$2$};
\node [above,align=center] at (8.3,0) {$1$};
\node [align=center,right] at (8.9,0.15) {$1$};
\node [align=center,right] at (8.9,-0.15) {$1$};
\node [align=center,align=center] at (9.5,0) {$+$};
\draw (11,0) circle [radius=0.2];
\draw (10.3,0) circle [radius=0.2];
\draw (10.5,0)--(10.8,0);
\draw (11.16,0.1)--(11.4,0.15);
\draw (11.16,-0.1)--(11.4,-0.15);
\draw (10.14,0.1) .. controls (9.8,0.2) and (9.8,-0.2) ..  (10.14,-0.1);
\node [align=center,align=center] at (11,0) {$0$};
\node [align=center,align=center] at (10.3,0) {$0$};
\node [above,align=center] at (10.14,0.1) {$2$};
\node [below,align=center] at (10.14,-0.1) {$2$};
\node [above,align=center] at (10.5,0) {$2$};
\node [above,align=center] at (10.8,0) {$2$};
\node [align=center,right] at (11.4,0.15) {$1$};
\node [align=center,right] at (11.4,-0.15) {$1$};
\end{tikzpicture}.&&
\end{flalign*}

\end{Example}

For a fixed $N$,
let $\cG_{g,n}(N)$ be the set of all stable graphs
of genus $g$ with $n$ external edges,
labelled by $\{1,2,\cdots,N\}$;
and $\cG_{g,n}^c(N)$ be the subset consisting of connected graphs.
Then the abstract free energy can be generalized to
\be
\widehat{\cF}_g:=\sum_{\Gamma\in\cG_{g,0}^c(N)}\frac{1}{|\Aut(\Gamma)|}\Gamma.
\ee

Denote by $\cG_{g;l_1,\cdots,l_N}^c(N)\subset\cG_{g,n}^c(N)$ the subset
consisting of all connected labelled
stable graphs of genus $g$, with $l_j$ external edges
labelled by $j$ for every $j\in\{1,\cdots,N\}$
(we have $l_1+\cdots+\l_N=n$).
Define
\be
\wcF_{g;l_1,\cdots,l_N}:=
\sum_{\Gamma\in \cG_{g;l_1,\cdots,l_N}^c(N)}
\frac{1}{|\Aut(\Gamma)|}\Gamma,
\ee
for $2g-2+\sum\limits_{j=1}^N l_j>0$,
then Lemma \ref{lem-original-D} and Theorem \ref{thm-original-rec}
can be generalized to:

\begin{Lemma}\label{lem-original-N-rec}
(\cite[Lemma 3.1]{wz})
For every $2g-2+\sum\limits_{j=1}^N l_j>0$,
we have
\be
\cD_j\wcF_{g;l_1,\cdots,l_N}=(l_j+1)\cdot\wcF_{g;l_1,\cdots,l_j+1,\cdots,l_N}.
\ee
\end{Lemma}

\begin{Theorem}(\cite[Theorem 3.1, Theorem 3.2]{wz})
\label{thm-original-N-rec}
For every $2g-2+\sum\limits_{j=1}^N l_j>0$, we have
\begin{equation*}
\begin{split}
&K_{ij}\widehat{\cF}_{g;l_1,\cdots,l_N}=\cD_i\cD_j\widehat{\cF}_{g-1;l_1,\cdots,l_N}+
\sum_{\substack{g_1+g_2=g\\p_k+q_k=l_k}}
\cD_i\widehat{\cF}_{g_1;p_1,\cdots,p_N}\cD_j\widehat{\cF}_{g_2;q_1,\cdots,q_N},
\quad i\not=j;\\
&K_{ii}\widehat{\cF}_{g;l_1,\cdots,l_N}=\half\biggl(
\cD_i\cD_i\widehat{\cF}_{g-1;l_1,\cdots,l_N}+
\sum_{\substack{g_1+g_2=g\\p_k+q_k=l_k}}
\cD_i\widehat{\cF}_{g_1;p_1,\cdots,p_N}\cD_i\widehat{\cF}_{g_2;q_1,\cdots,q_N}\biggr).
\end{split}
\end{equation*}
In particular,
for $l_1=\cdots=l_N=0$,
we have the following recursions for the free energy
of genus $g$ ($g\geq 2$):
\be
\begin{split}
&K_{ij}\widehat{\cF}_g=\cD_i\partial_j\widehat{\cF}_{g-1}+
\sum_{r=1}^{g-1}\partial_i\widehat{\cF}_r\partial_j\widehat{\cF}_{g-r},
\quad i\not=j;\\
&K_{ii}\widehat{\cF}_g=\frac{1}{2}\biggl(\cD_i\partial_i\widehat{\cF}_{g-1}
+\sum_{r=1}^{g-1}\partial_i\widehat{\cF}_r\partial_i\widehat{\cF}_{g-r}\biggr).
\end{split}
\ee
Here we use the convention
\ben
&&\pd_j \wcF_{1;l_1,\cdots,l_N}:=\wcF_{1;l_1,\cdots,l_j+1,\cdots,l_N},
\quad \text{for } l_1=\cdots=l_N=0;\\
&&\cD_j \wcF_{0;l_1,\cdots,l_N}= (l_j+1)\wcF_{0;l_1,\cdots,l_j+1,\cdots,l_N},
\quad \text{for } l_1+\cdots+l_N=2;\\
&&\cD_i\cD_j \wcF_{0;l_1,\cdots,l_N}:=(l_j+1)\cD_i \wcF_{0;l_1,\cdots,l_j+1,\cdots,l_N},
\quad \text{for } l_1+\cdots+l_N=1.\\
\een
\end{Theorem}

\subsection{Realization of the abstract QFT by Feynman rules}
\label{sec-pre-realization}

In this subsection we recall the realization of this abstract quantum
field theory introduced in \cite[\S 4]{wz}.
Given a sequence of functions $\{F_{g;l_1,\cdots,l_N}(t)\}_{2g-2+\sum l_j>0}$
and a choice of propagator
$\kappa=(\kappa_{ij})_{1\leq i,j\leq N}$,
we assign Feynman rules to our stable graphs, and in such a way
we are able to construct a new sequence $\wF_{g;l_1,\cdots,l_N}(t,\kappa)$.

Fix a positive integer $N$.
and let $\kappa=(\kappa_{ij}(t))$ be a non-degenerate symmetric matrix
of size $N\times N$, where $t=(t_1,\cdots,t_N)$.
Let $F_{g;l_1,\cdots,l_N}(t)$ be a sequence of holomorphic functions
in $t$ for every $2g-2+\sum\limits_{j=1}^N l_j>0$.
We assign the following Feynman rule to a stable graph $\Gamma\in\cG_{g,n}^c(N)$:
\be\label{eqn:FR}
\Gamma \mapsto \omega_\Gamma = \prod_{v\in V(\Gamma)} \omega_v \cdot
\prod_{e\in E(\Gamma)} \omega_e,
\ee
where $V(\Gamma)$ is the set of vertices of $\Gamma$,
and $E(\Gamma)$ is the set of internal edges of $\Gamma$.
The weight of a vertex $v\in V(\Gamma)$ is defined to be
\be
\omega_v:=F_{g_v;\val_1(v),\cdots,\val_N(v)}(t),
\ee
where $g_v$ is the genus associated to $v$,
and $\val_j(v)$ is the number of half-edges labelled by $j$ incident at $v$.
The weight of an internal edge $e$ is defined to be
\ben
\omega_e=\kappa_{j_1(e)j_2(e)},
\een
where $j_1(e)$, $j_2(e)$ are the two labels of the internal edge $e$.

Then the abstract free energy $\wcF_g(t,\kappa)$
of genus $g$ ($g\geq 2$) is realized by
\be\label{realization-FE}
\wF_g(t,\kappa)=\sum_{\Gamma\in\cG_{g,0}^c(N)}\frac{1}{|\Aut(\Gamma)|}\omega_\Gamma,
\ee
and the abstract $n$-point function $\wcF_{g;l_1,\cdots,l_N}$
($2g-2+\sum\limits_{j=1}^N l_j>0$) is realized by
\be\label{realization-npt-real-N}
\wF_{g;l_1,\cdots,l_N}(t,\kappa)=\sum_{\Gamma\in\cG_{g;l_1,\cdots,l_N}^c(N)}
\frac{1}{|\Aut(\Gamma)|}\omega_\Gamma.
\ee

This realization of the abstract free energies
can be represented by a formal Gaussian integral as follows:

\begin{Theorem}\label{thm-Gaussian}
(\cite[Theorem 4.1]{wz})
Define a partition function $\widehat{Z}(t,\kappa)$ to be
\be\label{eq-thm-Gaussian}
\begin{split}
\widehat{Z}(t,\kappa)=\frac{1}{(2\pi \lambda^2)^\frac{N}{2}}
\int \exp\biggl\{&
\sum_{2g-2+\sum\limits_{j=1}^N l_j>0}\lambda^{2g-2}\cdot
\frac{F_{g;l_1,\cdots,l_N}(t)}{l_1!\cdots l_N!}
\cdot\prod_{j=1}^{N}(\eta_j-t_j)^{l_j}\\
&\qquad-\frac{\lambda^{-2}}{2}\cdot(\eta-t)^T\kappa^{-1}
(\eta-t)\biggr\}d\eta,
\end{split}
\ee
then we have
\ben
\frac{\widehat{Z}(t,\kappa)}{\sqrt{\det(\kappa)}}
=\exp\biggl(\sum_{g=2}^{\infty}\lambda^{2g-2}\wF_g(t,\kappa)\biggr).
\een
\end{Theorem}

\subsection{Realization of the recursion relations}
\label{pre-realization-rec}

Let us recall the realization of the recursion relations in this subsection
(see \cite[\S 4]{wz}).
In order to realize the recursions in Lemma \ref{lem-original-N-rec}
and Theorem \ref{thm-original-N-rec},
first we need to realize the edge-cutting operator $K_{ij}$
and the edge-adding operators $\cD_i=\pd_i+\gamma_i$.

The edge-cutting operator $K_{ij}$ can be realized by
the partial derivative $\frac{\pd}{\pd\kappa_{ij}}$,
since using the Feynman rule \eqref{eqn:FR} we can easily check:
\ben
\pd_{\kappa_{ij}}\omega_\Gamma=\omega_{K_{ij}\Gamma}.
\een
The operator $\pd_i$ is realized by an operator
$\tilde\pd_i$ with the following properties
(see \cite[Definition 2.1]{wz2}):

\begin{itemize}
\item[(1)]
$\tilde\pd_i(F_{g;l_1,\cdots,l_N})=F_{g;l_1,\cdots,l_i+1,\cdots,l_N}$
for every $2g-2+\sum\limits_{j=1}^N l_j>0$;

\item[(2)]
$\tilde\pd_i(\kappa_{pq})=\sum\limits_{r,s=1}^N\kappa_{pr}\kappa_{qs}F_{0;l_1,\cdots,l_N}$,
where $l_1+\cdots+l_N=3$ and $l_j$ is the number of $j$ appearing in the sequence $(i,r,s)$.

\item[(3)]
$\tilde\pd_k$ acts on a polynomial in $F_{g;l_1,\cdots,l_N}$ and $\kappa_{pq}$ via
Leibniz rule.

\end{itemize}

\begin{Example}\label{eg-HAE-special}
Let $\{F_g(t_1,\cdots,t_N)\}$ be a sequence of holomorphic functions
in $t$ for $g\geq 0$. We take
\ben
F_{g;l_1,\cdots,l_N}:=\big(\frac{\pd}{\pd t_1}\big)^{l_1}\cdots
\big(\frac{\pd}{\pd t_N}\big)^{l_N}F_g(t),
\een
and choose the propagator $\kappa$ to be
\ben
\kappa=\biggl(C-Hess(F_0)\biggr)^{-1},
\een
where $C$ is either a constant symmetric matrix
or an anti-holomorphic symmetric matrix function in $t$
such that the matrix $(C-Hess(F_0))$ is invertible,
and $Hess(F_0)$ is the Hessian of $F_0(t_1,\cdots,t_N)$.
Then clearly $\pd_j$ can be realized by the partial derivative $\frac{\pd}{\pd t_j}$.
\end{Example}

The realization of $\gamma_i$ is more complicated.
Here we only describe the realization of the operator $\cD_i\pd_j$
acting on $\omega_\Gamma$ for $\Gamma\in\cG_{g,0}^c(N)$.
For a stable graph $\Gamma\in\cG_{g,0}^c(N)$,
we have
\ben
D_i\tilde\pd_j\omega_\Gamma:=\omega_{\cD_i\partial_j\Gamma}
=\tilde\pd_i\tilde\pd_j\omega_\Gamma+
\sum_{l,m=0}^{N}\tilde\pd_l\omega_\Gamma\cdot\kappa_{lm}
F_{0;l_1,\cdots,l_N},
\een
where $\tilde\pd_i$ and $\tilde\pd_j$ on the right-hand-side
are realizations of $\pd_i$ and $\pd_j$ respectively,
and $l_1+\cdots+l_N=3$ and $l_k$ is the number of $k$ appearing in the sequence $(m,i,j)$.

For the special case $N=1$,
the operator $\gamma$ acting on $\omega_\Gamma$ for $\Gamma\in\cG_{g,n}^c$
can be realized by simply multiplying by $n\cdot\kappa F_{0,3}$.

Now the quadratic recursion can be realized by:
\begin{Theorem}\label{thm-N-rec-realization} (\cite[Theorem 4.3]{wz})
Let $\tilde\pd_i$ and $D_i\tilde\pd_j$ be the operators defined above,
then for $g\geq 2$, we have the recursions
\ben
&&\partial_{\kappa_{ij}}\wF_g=D_i\tilde\pd_j\wF_{g-1}
+\sum_{r=1}^{g-1}\tilde\pd_i\wF_r\tilde\pd_{j}\wF_{g-r},\quad i\not=j;\\
&&\partial_{\kappa_{ii}}\wF_g=\frac{1}{2}\big(D_i\tilde\pd_i\wF_{g-1}
+\sum_{r=1}^{g-1}\tilde\pd_{i}\wF_r\tilde\pd_{i}\wF_{g-r}\big),
\een
for the free energy $\wF_g$ defined by \eqref{realization-FE}.
\end{Theorem}

\section{Dual Diagrammatics of Stable Graphs and Dual Abstract Quantum Field Theory}
\label{sec:dual-diagrammatics}

In this section we reverse the construction of $\wcF_{g,n}$ recalled in last section
and present a dual construction of $\cF_{g,n}$ from $\wcF_{g,n}$
based on the duality of the dual diagrammatics of stable graphs.
We will focus on the one-dimensional case (i.e., $N=1$) in this section.

\subsection{Dotted stable graphs}
\label{sec:Dotted}

In order to reverse the construction of $\{\wF_{g,n}\}$ from $\{F_g\}$,
one needs to reverse the role that they play.
To construct $\{\wF_{g,n}\}$ from $\{F_g\}$,
one first finds $\{F_{g,n}\}$ and uses them to obtain Feynman rules
that associate weights to the vertices of some suitable Feynman graphs,
each $\wF_{g,n}$ is then given by a summation over stable graphs of type $(g,n)$.
To a vertex of genus $g$ of valence $n$,
we associate a contribution of $F_{g,n}$.
If the reversed construction exists,
then $\wF_{g,n}$ should be used to associate a weight to a vertex of genus $g$
and valence $n$, and $F_{g,n}$ should be given by a summation over stable graphs of type $(g,n)$.
Since both the original construction and the reversed construction involve
stable graphs,
we will make a distinction between their graphic representations.
The stable graphs in the original construction will be represented as graphs with solid lines,
while the stable graphs in the reversed construction will be  represented as graphs with dotted lines.
They will be referred to as the ordinary stable graphs and the dotted stable graphs respectively.

Now we present the reversed construction.
We first represent $\wcF_{g,n}$ as $\frac{1}{|\Aut \Gamma^\vee|} \Gamma^\vee$,
here $\Gamma^\vee$ is a dotted stable graph consisting of one vertex of genus $g$ and $n$ external edges.
It is clear that in this case $|\Aut(\Gamma^\vee)| = n!$,
and  this graph is still a stable graph.
It will be called a dotted stable vertex.

\begin{Example}
\label{eg-new-vert}
Using \cite[Example 2.5]{wz},
we present some explicit examples of dotted stable graphs.
\begin{flalign}\label{eg1-eq1}
\begin{tikzpicture}
\draw [dotted,thick](0,0) circle [radius=0.2];
\node [align=center,align=center] at (0,0) {$0$};
\draw [dotted,thick](-0.5,0)--(-0.2,0);
\draw [dotted,thick](0.16,0.1)--(0.5,0.15);
\draw [dotted,thick](0.16,-0.1)--(0.5,-0.15);
\node [align=center,align=center] at (0.8,0) {$=$};
\draw (0+1.6,0) circle [radius=0.2];
\node [align=center,align=center] at (0+1.6,0) {$0$};
\draw (-0.5+1.6,0)--(-0.2+1.6,0);
\draw (0.16+1.6,0.1)--(0.5+1.6,0.15);
\draw (0.16+1.6,-0.1)--(0.5+1.6,-0.15);
\end{tikzpicture},&&
\end{flalign}

\begin{flalign}\label{eg1-eq2}
\begin{tikzpicture}
\draw [dotted,thick](1.6-1.6,0) circle [radius=0.2];
\draw [dotted,thick](1.1-1.6,0.15)--(1.44-1.6,0.1);
\draw [dotted,thick](1.1-1.6,-0.15)--(1.44-1.6,-0.1);
\draw [dotted,thick](1.76-1.6,0.1)--(2.1-1.6,0.15);
\draw [dotted,thick](1.76-1.6,-0.1)--(2.1-1.6,-0.15);
\node [align=center,align=center] at (1.6-1.6,0) {$0$};
\node [align=center,align=center] at (0.8,0) {$=$};
\draw (1.6+0.1,0) circle [radius=0.2];
\draw (1.1+0.1,0.15)--(1.44+0.1,0.1);
\draw (1.1+0.1,-0.15)--(1.44+0.1,-0.1);
\draw (1.76+0.1,0.1)--(2.1+0.1,0.15);
\draw (1.76+0.1,-0.1)--(2.1+0.1,-0.15);
\node [align=center,align=center] at (1.6+0.1,0) {$0$};
\node [align=center,align=center] at (2.6,0) {$+3$};
\draw (1.6+1.9,0) circle [radius=0.2];
\draw (2.2+1.9,0) circle [radius=0.2];
\draw (1.1+1.9,0.15)--(1.44+1.9,0.1);
\draw (1.1+1.9,-0.15)--(1.44+1.9,-0.1);
\draw (1.1+1.9,0.15)--(1.44+1.9,0.1);
\draw (2.36+1.9,0.1)--(2.7+1.9,0.15);
\draw (2.36+1.9,-0.1)--(2.7+1.9,-0.15);
\draw (1.8+1.9,0)--(2+1.9,0);
\node [align=center,align=center] at (1.6+1.9,0) {$0$};
\node [align=center,align=center] at (2.2+1.9,0) {$0$};
\end{tikzpicture},&&
\end{flalign}

\begin{flalign}\label{eg1-eq3}
\begin{tikzpicture}
\draw [dotted,thick](1-1.6,0) circle [radius=0.2];
\draw [dotted,thick](1.2-1.6,0)--(1.5-1.6,0);
\node [align=center,align=center] at (1-1.6,0) {$1$};
\node [align=center,align=center] at (0.3,0) {$=$};
\draw (1,0) circle [radius=0.2];
\draw (1.2,0)--(1.5,0);
\node [align=center,align=center] at (1,0) {$1$};
\node [align=center,align=center] at (2,0) {$+\frac{1}{2}$};
\draw (1+1.8,0) circle [radius=0.2];
\draw (1.2+1.8,0)--(1.5+1.8,0);
\draw (0.84+1.8,0.1) .. controls (0.5+1.8,0.2) and (0.5+1.8,-0.2) ..  (0.84+1.8,-0.1);
\node [align=center,align=center] at (1+1.8,0) {$0$};
\end{tikzpicture},&&
\end{flalign}

\begin{flalign}\label{eg1-eq4}
\begin{tikzpicture}
\draw [dotted,thick](1-2+0.4,0) circle [radius=0.2];
\draw [dotted,thick](1.2-2+0.4,0)--(1.5-2+0.4,0);
\draw [dotted,thick](0.5-2+0.4,0)--(0.8-2+0.4,0);
\node [align=center,align=center] at (1-2+0.4,0) {$1$};
\node [align=center,align=center] at (0.3-0.1,0) {$=$};
\draw (1,0) circle [radius=0.2];
\draw (1.2,0)--(1.5,0);
\draw (0.5,0)--(0.8,0);
\node [align=center,align=center] at (1,0) {$1$};
\node [align=center,align=center] at (2,0) {$+\frac{1}{2}$};
\draw (1+1.8,0) circle [radius=0.2];
\draw (1.17+1.8,0.1)--(1.4+1.8,0.15);
\draw (1.17+1.8,-0.1)--(1.4+1.8,-0.15);
\draw (0.84+1.8,0.1) .. controls (0.5+1.8,0.2) and (0.5+1.8,-0.2) ..  (0.84+1.8,-0.1);
\node [align=center,align=center] at (1+1.8,0) {$0$};
\node [align=center,align=center] at (3.7,0) {$+$};
\draw (1+3.9,0) circle [radius=0.2];
\draw (0.4+3.9,0) circle [radius=0.2];
\draw (1.17+3.9,0.1)--(1.4+3.9,0.15);
\draw (1.17+3.9,-0.1)--(1.4+3.9,-0.15);
\draw (0.6+3.9,0)--(0.8+3.9,0);
\node [align=center,align=center] at (1+3.9,0) {$0$};
\node [align=center,align=center] at (0.4+3.9,0) {$1$};
\node [align=center,align=center] at (5.8,0) {$+\frac{1}{2}$};
\draw (1+6.2,0) circle [radius=0.2];
\draw (0.4+6.2,0) circle [radius=0.2];
\draw (1.17+6.2,0.1)--(1.4+6.2,0.15);
\draw (1.17+6.2,-0.1)--(1.4+6.2,-0.15);
\draw (0.6+6.2,0)--(0.8+6.2,0);
\node [align=center,align=center] at (1+6.3-0.1,0) {$0$};
\node [align=center,align=center] at (0.4+6.3-0.1,0) {$0$};
\draw (0.24+6.3-0.1,0.1) .. controls (-0.1+6.3-0.1,0.2) and (-0.1+6.3-0.1,-0.2) ..  (0.24+6.3-0.1,-0.1);
\node [align=center,align=center] at (8.2-0.1,0) {$+\frac{1}{2}$};
\draw (1+8.1-0.1,0) circle [radius=0.2];
\draw (0.5+8.1-0.1,0)--(0.8+8.1-0.1,0);
\draw (1.18+8.1-0.1,0.07)--(1.42+8.1-0.1,0.07);
\draw (1.18+8.1-0.1,-0.07)--(1.42+8.1-0.1,-0.07);
\draw (1.8+8.1-0.1,-0)--(2.1+8.1-0.1,0);
\draw (1.6+8.1-0.1,0) circle [radius=0.2];
\node [align=center,align=center] at (1+8.1-0.1,0) {$0$};
\node [align=center,align=center] at (1.6+8.1-0.1,0) {$0$};
\end{tikzpicture},&&
\end{flalign}

\be\label{eg1-eq5}
\begin{split}
&\begin{tikzpicture}
\draw [dotted,thick](0.3,0) circle [radius=0.2];
\node [align=center,align=center] at (0.3,0) {$2$};
\node [align=center,align=center] at (0.1+0.7,0) {$=$};
\draw (1+0.3,0) circle [radius=0.2];
\node [align=center,align=center] at (1+0.3,0) {$2$};
\node [align=center,align=center] at (1.6+0.2+0.1,0) {$+\frac{1}{2}$};
\draw (1+1.4+0.2+0.1,0) circle [radius=0.2];
\draw (0.84+1.4+0.2+0.1,0.1) .. controls (0.5+1.4+0.2+0.1,0.2) and (0.5+1.4+0.2+0.1,-0.2) ..  (0.84+1.4+0.2+0.1,-0.1);
\node [align=center,align=center] at (1+1.4+0.2+0.1,0) {$1$};
\node [align=center,align=center] at (3+0.2+0.2,0) {$+\frac{1}{2}$};
\draw (1+2.6+0.2+0.2,0) circle [radius=0.2];
\draw (1.2+2.6+0.2+0.2,0)--(1.4+2.6+0.2+0.2,0);
\draw (1.6+2.6+0.2+0.2,0) circle [radius=0.2];
\node [align=center,align=center] at (1+2.6+0.2+0.2,0) {$1$};
\node [align=center,align=center] at (1.6+2.6+0.2+0.2,0) {$1$};
\node [align=center,align=center] at (5+0.3,0) {$+\frac{1}{8}$};
\draw (1+4.8+0.3,0) circle [radius=0.2];
\draw (0.84+4.8+0.3,0.1) .. controls (0.5+4.8+0.3,0.2) and (0.5+4.8+0.3,-0.2) ..  (0.84+4.8+0.3,-0.1);
\draw (1.16+4.8+0.3,0.1) .. controls (1.5+4.8+0.3,0.2) and (1.5+4.8+0.3,-0.2) ..  (1.16+4.8+0.3,-0.1);
\node [align=center,align=center] at (1+4.8+0.3,0) {$0$};
\node [align=center,align=center] at (6.6+0.4,0) {$+\frac{1}{2}$};
\draw (1+6.8+0.4,0) circle [radius=0.2];
\draw (0.4+6.8+0.4,0) circle [radius=0.2];
\draw (0.6+6.8+0.4,0)--(0.8+6.8+0.4,0);
\draw (1.16+6.8+0.4,0.1) .. controls (1.5+6.8+0.4,0.2) and (1.5+6.8+0.4,-0.2) ..  (1.16+6.8+0.4,-0.1);
\node [align=center,align=center] at (1+6.8+0.4,0) {$0$};
\node [align=center,align=center] at (0.4+6.8+0.4,0) {$1$};
\node [align=center,align=center] at (8.6+0.5,0) {$+\frac{1}{8}$};
\draw (1+9+0.5,0) circle [radius=0.2];
\draw (0.4+9+0.5,0) circle [radius=0.2];
\draw (0.6+9+0.5,0)--(0.8+9+0.5,0);
\draw (1.16+9+0.5,0.1) .. controls (1.5+9+0.5,0.2) and (1.5+9+0.5,-0.2) ..  (1.16+9+0.5,-0.1);
\draw (0.24+9+0.5,0.1) .. controls (-0.1+9+0.5,0.2) and (-0.1+9+0.5,-0.2) ..  (0.24+9+0.5,-0.1);
\node [align=center,align=center] at (1+9+0.5,0) {$0$};
\node [align=center,align=center] at (0.4+9+0.5,0) {$0$};
\end{tikzpicture}\qquad
\\
&\qquad\quad\begin{tikzpicture}
\node [align=center,align=center] at (10.6+0.2,0) {$+\frac{1}{12}$};
\draw (1+10.2+0.2,0) circle [radius=0.2];
\draw (1.2+10.2+0.2,0)--(1.4+10.2+0.2,0);
\draw (1.16+10.2+0.2,0.1)--(1.44+10.2+0.2,0.1);
\draw (1.16+10.2+0.2,-0.1)--(1.44+10.2+0.2,-0.1);
\draw (1.6+10.2+0.2,0) circle [radius=0.2];
\node [align=center,align=center] at (1+10.2+0.2,0) {$0$};
\node [align=center,align=center] at (1.6+10.2+0.2,0) {$0$};
\end{tikzpicture}.
\end{split}
\ee

In particular, we see that the dotted vertex of genus $0$ and valence three
is the same as the ordinary vertex of genus $0$ and valence three.
\end{Example}

Note the dotted stable vertex has the same number of external edges as the ordinary stable graphs
in the summation.
We say that the external edges of each ordinary graphs in the summation are inherited from those of the dotted graphs.

Now to each  dotted stable graph,
we associate a  linear combination
of ordinary stable graphs in the following way.

Given a dotted stable graph $\Gamma^\vee$,
we need to first cut all the internal edges to reduce to the case of the dotted vertex,
rewrite each of the dotted vertices in terms of a summation over ordinary stable graphs.
This gives a bunch of disconnected stable graphs,
with a connected component from each of the dotted vertex.
Next, in order to obtain a summation over some connected stable graphs,
we need to suitably glue some half edges so that we have the same number
of external edges for the dotted stable graph and the ordinary stable graphs
that we sum over.
Since the external edges of the stable graphs in the summation for each dotted vertex
are inherited from the dotted vertex,
this gives us a rule how they can be glued together:
They can be glued together if and only if the ancestor half-edges in the dotted graphs are glued together.

More precisely,
we first distinguish half-edges of $\Gamma^\vee$ by
giving every half-edge a name.
We require that different half-edges have different names.
Next, we cut off every internal edges of $\Gamma^\vee$ to obtain
some dotted stable vertices with all external edges
distinguished by different names.

Next we express dotted stable vertices with names as a summation
over ordinary stable graphs whose external edges have inherited names.
Denote by $\Gamma'^\vee$ a dotted stable vertex of genus $g$ and valence $n$,
whose external edges have different names.
Then we define
\be\label{def-mark}
\Gamma'^\vee=\sum_{\Gamma'}\frac{1}{|\Aut(\Gamma')|}\Gamma',
\ee
where the sum is over all possible stable graphs $\Gamma'$ of type $(g,n)$,
whose external edges have the same $n$ different names as the external edges of  $\Gamma'^\vee$.
The automorphisms of such $\Gamma'$ should preserves the names of
all the external edges,
hence they fix each of the external edges.

Then we glue the external edges of these ordinary stable graphs together
in such a way that two external edges are joined together if and only if
their ancestors  in $\Gamma^\vee$ are joined together.
Finally, we forget all these names,
and multiply by a factor $(-1)^{|E^\vee(\Gamma^\vee)|}$ to the whole expression,
where $E^\vee(\Gamma^\vee)$ is the set of all dotted internal edges of $\Gamma^\vee$.
In this way we obtain a linear combination
of ordinary stable graphs associated to
the dotted stable graph $\Gamma^\vee$.
In the following we will abuse the notations by using $\Gamma^\vee$
to denote both a dotted graph and this linear combination.

\begin{Example}
Now let us give some examples.
First consider the following graph,
and we give names $a,b,c,d$ to its half-edges then cut the internal edge:
\ben
\begin{tikzpicture}
\draw [dotted,thick](1.6-1.6,0) circle [radius=0.2];
\draw [dotted,thick](-0.16,0.1) .. controls (-0.5,0.2) and (-0.5,-0.2) ..  (-0.16,-0.1);
\draw [dotted,thick](1.76-1.6,0.1)--(2.1-1.6,0.15);
\draw [dotted,thick](1.76-1.6,-0.1)--(2.1-1.6,-0.15);
\node [align=center,align=center] at (1.6-1.6,0) {$0$};
\draw [dotted,thick](1.6-1.6+3,0) circle [radius=0.2];
\draw [dotted,thick](-0.16+3,0.1) .. controls (-0.5+3,0.2) and (-0.5+3,-0.2) ..  (-0.16+3,-0.1);
\draw [dotted,thick](1.76-1.6+3,0.1)--(2.1-1.6+3,0.15);
\draw [dotted,thick](1.76-1.6+3,-0.1)--(2.1-1.6+3,-0.15);
\node [align=center,align=center] at (1.6-1.6+3,0) {$0$};
\node [above,align=center] at (-0.16+3,0.1,0) {$a$};
\node [below,align=center] at (-0.16+3,-0.1,0) {$b$};
\node [align=center,right] at (2.1-1.6+3,0.15) {$c$};
\node [align=center,right] at (2.1-1.6+3,-0.15) {$d$};
\draw [dotted,thick](1.6-1.6+6,0) circle [radius=0.2];
\draw [dotted,thick](5.84,0.1)--(5.5,0.15);
\draw [dotted,thick](5.84,-0.1)--(5.5,-0.15);
\draw [dotted,thick](1.76-1.6+6,0.1)--(2.1-1.6+6,0.15);
\draw [dotted,thick](1.76-1.6+6,-0.1)--(2.1-1.6+6,-0.15);
\node [align=center,align=center] at (1.6-1.6+6,0) {$0$};
\node [align=center,left] at (5.5,0.15) {$a$};
\node [align=center,left] at (5.5,-0.15) {$b$};
\node [align=center,right] at (2.1-1.6+6,0.15) {$c$};
\node [align=center,right] at (2.1-1.6+6,-0.15) {$d$};
\node [align=center,right] at (1.3,0) {$\to$};
\node [align=center,right] at (4.2,0) {$\to$};
\end{tikzpicture}.
\een
Then by definition, the third graph above equals
\be\label{eg-mark-eq1}
\begin{tikzpicture}
\draw (0,0) circle [radius=0.2];
\draw (-0.16,0.1)--(-0.5,0.15);
\draw (-0.16,-0.1)--(-0.5,-0.15);
\draw (0.16,0.1)--(0.5,0.15);
\draw (0.16,-0.1)--(0.5,-0.15);
\node [align=center,align=center] at (0,0) {$0$};
\node [align=center,left] at (-0.5,0.15) {$a$};
\node [align=center,left] at (-0.5,-0.15) {$b$};
\node [align=center,right] at (0.5,0.15) {$c$};
\node [align=center,right] at (0.5,-0.15) {$d$};
\node [align=center,align=center] at (1.1,0) {$+$};
\draw (2.2,0) circle [radius=0.2];
\node [align=center,align=center] at (2.2,0) {$0$};
\draw (2.8,0) circle [radius=0.2];
\node [align=center,align=center] at (2.8,0) {$0$};
\draw (2.4,0)--(2.6,0);
\draw (1.7,0.15)--(2.04,0.1);
\draw (1.7,-0.15)--(2.04,-0.1);
\draw (2.96,0.1)--(3.3,0.15);
\draw (2.96,-0.1)--(3.3,-0.15);
\node [align=center,left] at (1.7,0.15) {$a$};
\node [align=center,left] at (1.7,-0.15) {$b$};
\node [align=center,right] at (3.3,0.15) {$c$};
\node [align=center,right] at (3.3,-0.15) {$d$};
\node [align=center,align=center] at (3.9,0) {$+$};
\draw (2.2+2.8,0) circle [radius=0.2];
\node [align=center,align=center] at (2.2+2.8,0) {$0$};
\draw (2.8+2.8,0) circle [radius=0.2];
\node [align=center,align=center] at (2.8+2.8,0) {$0$};
\draw (2.4+2.8,0)--(2.6+2.8,0);
\draw (1.7+2.8,0.15)--(2.04+2.8,0.1);
\draw (1.7+2.8,-0.15)--(2.04+2.8,-0.1);
\draw (2.96+2.8,0.1)--(3.3+2.8,0.15);
\draw (2.96+2.8,-0.1)--(3.3+2.8,-0.15);
\node [align=center,left] at (1.7+2.8,0.15) {$a$};
\node [align=center,left] at (1.7+2.8,-0.15) {$c$};
\node [align=center,right] at (3.3+2.8,0.15) {$b$};
\node [align=center,right] at (3.3+2.8,-0.15) {$d$};
\node [align=center,align=center] at (3.9+2.8,0) {$+$};
\draw (2.2+2.8+2.8,0) circle [radius=0.2];
\node [align=center,align=center] at (2.2+2.8+2.8,0) {$0$};
\draw (2.8+2.8+2.8,0) circle [radius=0.2];
\node [align=center,align=center] at (2.8+2.8+2.8,0) {$0$};
\draw (2.4+2.8+2.8,0)--(2.6+2.8*2,0);
\draw (1.7+2.8+2.8,0.15)--(2.04+2.8*2,0.1);
\draw (1.7+2.8+2.8,-0.15)--(2.04+2.8*2,-0.1);
\draw (2.96+2.8+2.8,0.1)--(3.3+2.8*2,0.15);
\draw (2.96+2.8+2.8,-0.1)--(3.3+2.8+2.8,-0.15);
\node [align=center,left] at (1.7+2.8+2.8,0.15) {$a$};
\node [align=center,left] at (1.7+2.8+2.8,-0.15) {$d$};
\node [align=center,right] at (3.3+2.8+2.8,0.15) {$b$};
\node [align=center,right] at (3.3+2.8+2.8,-0.15) {$c$};
\end{tikzpicture}.
\ee
Glue the edges named by $a$ and $b$ together,
then forget the names and multiply by $(-1)$,
we get
\ben
\begin{tikzpicture}
\draw [dotted,thick](1.6-1.6-0.2,0) circle [radius=0.2];
\draw [dotted,thick](-0.16-0.2,0.1) .. controls (-0.5-0.2,0.2) and (-0.5-0.2,-0.2) ..  (-0.16-0.2,-0.1);
\draw [dotted,thick](1.76-1.6-0.2,0.1)--(2.1-1.6-0.2,0.15);
\draw [dotted,thick](1.76-1.6-0.2,-0.1)--(2.1-1.6-0.2,-0.15);
\node [align=center,align=center] at (1.6-1.6-0.2,0) {$0$};
\node [align=center,align=center] at (0.8,0) {$=-$};
\draw (1.6+0.1,0) circle [radius=0.2];
\draw (1.54,0.1) .. controls (1.2,0.2) and (1.2,-0.2) ..  (1.54,-0.1);
\draw (1.76+0.1,0.1)--(2.1+0.1,0.15);
\draw (1.76+0.1,-0.1)--(2.1+0.1,-0.15);
\node [align=center,align=center] at (1.6+0.1,0) {$0$};
\node [align=center,align=center] at (2.7,0) {$-$};
\draw (1.7+2,0) circle [radius=0.2];
\draw (2.3+2,0) circle [radius=0.2];
\draw (3.54,0.1) .. controls (3.2,0.2) and (3.2,-0.2) ..  (3.54,-0.1);
\draw (2.46+2,0.1)--(2.8+2,0.15);
\draw (2.46+2,-0.1)--(2.8+2,-0.15);
\draw (1.9+2,0)--(2.1+2,0);
\node [align=center,align=center] at (1.7+2,0) {$0$};
\node [align=center,align=center] at (2.3+2,0) {$0$};
\node [align=center,align=center] at (5.3,0) {$-2$};
\draw (6.2,0) circle [radius=0.2];
\draw (5.7,0)--(6,0);
\draw (6.38,0.07)--(6.62,0.07);
\draw (6.38,-0.07)--(6.62,-0.07);
\draw (7,-0)--(7.3,0);
\draw (6.8,0) circle [radius=0.2];
\node [align=center,align=center] at (6.2,0) {$0$};
\node [align=center,align=center] at (6.8,0) {$0$};
\end{tikzpicture}.
\een

Now let us consider the following dotted stable graph with two internal edges:
\ben
\begin{tikzpicture}
\draw [dotted,thick](1.6-1.6,0) circle [radius=0.2];
\draw [dotted,thick](-0.16,0.1) .. controls (-0.5,0.2) and (-0.5,-0.2) ..  (-0.16,-0.1);
\draw [dotted,thick](0.16,0.1) .. controls (0.5,0.2) and (0.5,-0.2) ..  (0.16,-0.1);
\node [align=center,align=center] at (1.6-1.6,0) {$0$};
\draw [dotted,thick](1.6-1.6+3,0) circle [radius=0.2];
\draw [dotted,thick](-0.16+3,0.1) .. controls (-0.5+3,0.2) and (-0.5+3,-0.2) ..  (-0.16+3,-0.1);
\draw [dotted,thick](0.16+3,0.1) .. controls (0.5+3,0.2) and (0.5+3,-0.2) ..  (0.16+3,-0.1);
\node [align=center,align=center] at (1.6-1.6+3,0) {$0$};
\node [above,align=center] at (-0.16+3,0.1,0) {$a$};
\node [below,align=center] at (-0.16+3,-0.1,0) {$b$};
\node [above,align=center] at (0.16+3,0.1) {$c$};
\node [below,align=center] at (0.16+3,-0.1) {$d$};
\draw [dotted,thick](1.6-1.6+6,0) circle [radius=0.2];
\draw [dotted,thick](5.84,0.1)--(5.5,0.15);
\draw [dotted,thick](5.84,-0.1)--(5.5,-0.15);
\draw [dotted,thick](1.76-1.6+6,0.1)--(2.1-1.6+6,0.15);
\draw [dotted,thick](1.76-1.6+6,-0.1)--(2.1-1.6+6,-0.15);
\node [align=center,align=center] at (1.6-1.6+6,0) {$0$};
\node [align=center,left] at (5.5,0.15) {$a$};
\node [align=center,left] at (5.5,-0.15) {$b$};
\node [align=center,right] at (2.1-1.6+6,0.15) {$c$};
\node [align=center,right] at (2.1-1.6+6,-0.15) {$d$};
\node [align=center,right] at (1.3,0) {$\to$};
\node [align=center,right] at (4.2,0) {$\to$};
\end{tikzpicture}.
\een
Glue $a$ with $b$, and $c$ with $d$
in the expression \eqref{eg-mark-eq1},
then forget the names and multiply by $(-1)^2$,
we get
\ben
\begin{tikzpicture}
\draw [dotted,thick](1.6-1.6,0) circle [radius=0.2];
\draw [dotted,thick](-0.16,0.1) .. controls (-0.5,0.2) and (-0.5,-0.2) ..  (-0.16,-0.1);
\draw [dotted,thick](0.16,0.1) .. controls (0.5,0.2) and (0.5,-0.2) ..  (0.16,-0.1);
\node [align=center,align=center] at (1.6-1.6,0) {$0$};
\node [align=center,align=center] at (0.8,0) {$=$};
\draw (1.6+0.1,0) circle [radius=0.2];
\draw (1.54,0.1) .. controls (1.2,0.2) and (1.2,-0.2) ..  (1.54,-0.1);
\draw (1.86,0.1) .. controls (2.2,0.2) and (2.2,-0.2) ..  (1.86,-0.1);
\node [align=center,align=center] at (1.6+0.1,0) {$0$};
\node [align=center,align=center] at (2.8,0) {$+$};
\draw (1.7+2,0) circle [radius=0.2];
\draw (2.3+2,0) circle [radius=0.2];
\draw (3.54,0.1) .. controls (3.2,0.2) and (3.2,-0.2) ..  (3.54,-0.1);
\draw (4.46,0.1) .. controls (4.8,0.2) and (4.8,-0.2) ..  (4.46,-0.1);
\draw (1.9+2,0)--(2.1+2,0);
\node [align=center,align=center] at (1.7+2,0) {$0$};
\node [align=center,align=center] at (2.3+2,0) {$0$};
\node [align=center,align=center] at (5.3,0) {$+2$};
\draw (5.9,0) circle [radius=0.2];
\draw (6.06,0.1)--(6.34,0.1);
\draw (6.06,-0.1)--(6.34,-0.1);
\draw (6.1,0)--(6.3,0);
\draw (6.5,0) circle [radius=0.2];
\node [align=center,align=center] at (5.9,0) {$0$};
\node [align=center,align=center] at (6.5,0) {$0$};
\end{tikzpicture}.
\een

Let us see a third example containing two dotted vertices:
\ben
\begin{tikzpicture}
\draw [dotted,thick](0,0) circle [radius=0.2];
\node [align=center,align=center] at (0,0) {$1$};
\draw [dotted,thick](0.7,0) circle [radius=0.2];
\node [align=center,align=center] at (0.7,0) {$0$};
\draw [dotted,thick](0.2,0)--(0.5,0);
\draw [dotted,thick](0.86,0.1)--(1.2,0.15);
\draw [dotted,thick](0.86,-0.1)--(1.2,-0.15);
\draw [dotted,thick](0+3,0) circle [radius=0.2];
\node [align=center,align=center] at (0+3,0) {$1$};
\draw [dotted,thick](0.7+3,0) circle [radius=0.2];
\node [align=center,align=center] at (0.7+3,0) {$0$};
\draw [dotted,thick](0.2+3,0)--(0.5+3,0);
\draw [dotted,thick](0.86+3,0.1)--(1.2+3,0.15);
\draw [dotted,thick](0.86+3,-0.1)--(1.2+3,-0.15);
\node [above,align=center] at (0.2+3,0) {$a$};
\node [above,align=center] at (0.5+3,0) {$b$};
\node [align=center,right] at (1.2+3,0.15) {$c$};
\node [align=center,right] at (1.2+3,-0.15) {$d$};
\draw [dotted,thick](0+3+3,0) circle [radius=0.2];
\node [align=center,align=center] at (0+3+3,0) {$1$};
\draw [dotted,thick](0.7+3+4-0.2,0) circle [radius=0.2];
\node [align=center,align=center] at (0.7+3+4-0.2,0) {$0$};
\draw [dotted,thick](0.2+3+3,0)--(0.5+3+3,0);
\draw [dotted,thick](7.2-0.2,0)--(7.5-0.2,0);
\draw [dotted,thick](0.86+3+4-0.2,0.1)--(1.2+3+4-0.2,0.15);
\draw [dotted,thick](0.86+3+4-0.2,-0.1)--(1.2+3+4-0.2,-0.15);
\node [above,align=center] at (0.2+3+3+0.1,0) {$a$};
\node [above,align=center] at (0.5+3+4-0.3,0) {$b$};
\node [align=center,right] at (1.2+3+4-0.2,0.15) {$c$};
\node [align=center,right] at (1.2+3+4-0.2,-0.15) {$d$};
\node [align=center,right] at (1.6,0) {$\to$};
\node [align=center,right] at (4.8,0) {$\to$};
\end{tikzpicture}.
\een
The two dotted vertices with names on external edges are given by
\ben
\begin{tikzpicture}
\draw [dotted,thick](1-1.6,0) circle [radius=0.2];
\draw [dotted,thick](1.2-1.6,0)--(1.5-1.6,0);
\node [align=center,align=center] at (1-1.6,0) {$1$};
\node [above,align=center] at (-0.4+0.1,0) {$a$};
\node [align=center,align=center] at (0.3,0) {$=$};
\draw (1,0) circle [radius=0.2];
\draw (1.2,0)--(1.5,0);
\node [align=center,align=center] at (1,0) {$1$};
\node [above,align=center] at (1.2+0.1,0) {$a$};
\node [align=center,align=center] at (2,0) {$+\frac{1}{2}$};
\draw (1+1.8,0) circle [radius=0.2];
\draw (1.2+1.8,0)--(1.5+1.8,0);
\draw (0.84+1.8,0.1) .. controls (0.5+1.8,0.2) and (0.5+1.8,-0.2) ..  (0.84+1.8,-0.1);
\node [align=center,align=center] at (1+1.8,0) {$0$};
\node [above,align=center] at (3+0.1,0) {$a$};
\end{tikzpicture}
\een
and
\ben
\begin{tikzpicture}
\draw [dotted,thick](0-0.3,0) circle [radius=0.2];
\node [align=center,align=center] at (0-0.3,0) {$0$};
\draw [dotted,thick](-0.5-0.3,0)--(-0.2-0.3,0);
\draw [dotted,thick](0.16-0.3,0.1)--(0.5-0.3,0.15);
\draw [dotted,thick](0.16-0.3,-0.1)--(0.5-0.3,-0.15);
\node [above,align=center] at (-0.3-0.3,0) {$b$};
\node [align=center,right] at (0.5-0.3,0.15,0) {$c$};
\node [align=center,right] at (0.5-0.3,-0.15,0) {$d$};
\node [align=center,align=center] at (0.8,0) {$=$};
\draw (0+1.6,0) circle [radius=0.2];
\node [align=center,align=center] at (0+1.6,0) {$0$};
\draw (-0.5+1.6,0)--(-0.2+1.6,0);
\draw (0.16+1.6,0.1)--(0.5+1.6,0.15);
\draw (0.16+1.6,-0.1)--(0.5+1.6,-0.15);
\node [above,align=center] at (1.3,0) {$b$};
\node [align=center,right] at (0.5+1.6,0.15) {$c$};
\node [align=center,right] at (0.5+1.6,-0.15) {$d$};
\end{tikzpicture}.
\een
Glue the half edges $a$ and $b$,
then forget the names and multiply by $(-1)$,
we get
\ben
\begin{tikzpicture}
\draw [dotted,thick](0,0) circle [radius=0.2];
\node [align=center,align=center] at (0,0) {$1$};
\draw [dotted,thick](0.6,0) circle [radius=0.2];
\node [align=center,align=center] at (0.6,0) {$0$};
\draw [dotted,thick](0.2,0)--(0.4,0);
\draw [dotted,thick](0.76,0.1)--(1,0.15);
\draw [dotted,thick](0.76,-0.1)--(1,-0.15);
\node [align=center,align=center] at (1.5,0) {$=-$};
\draw (2.1,0) circle [radius=0.2];
\node [align=center,align=center] at (2.1,0) {$1$};
\draw (2.7,0) circle [radius=0.2];
\node [align=center,align=center] at (2.7,0) {$0$};
\draw (2.3,0)--(2.5,0);
\draw (2.86,0.1)--(3.1,0.15);
\draw (2.86,-0.1)--(3.1,-0.15);
\node [align=center,align=center] at (3.5,0) {$-\frac{1}{2}$};
\draw (4.4,0) circle [radius=0.2];
\node [align=center,align=center] at (4.4,0) {$0$};
\draw (5,0) circle [radius=0.2];
\node [align=center,align=center] at (5,0) {$0$};
\draw (4.24,0.1) .. controls (3.9,0.2) and (3.9,-0.2) ..  (4.24,-0.1);
\draw (4.6,0)--(4.8,0);
\draw (5.16,0.1)--(5.5,0.15);
\draw (5.16,-0.1)--(5.5,-0.15);
\end{tikzpicture}.
\een

\end{Example}

\begin{Example}
\label{eg-genus2-graph}
Explicit expressions for all possible dotted stable graphs of genus $2$
without external edges are listed here:

\begin{flalign*}
\begin{tikzpicture}
\draw [dotted,thick](0.3+0.2,0) circle [radius=0.2];
\node [align=center,align=center] at (0.3+0.2,0) {$2$};
\node [align=center,align=center] at (0.1+0.7+0.2,0) {$=$};
\draw (1+0.3+0.2,0) circle [radius=0.2];
\node [align=center,align=center] at (1+0.3+0.2,0) {$2$};
\node [align=center,align=center] at (1.6+0.2+0.1+0.2,0) {$+\frac{1}{2}$};
\draw (1+1.4+0.2+0.1+0.2,0) circle [radius=0.2];
\draw (0.84+1.4+0.2+0.1+0.2,0.1) .. controls (0.5+1.4+0.2+0.1+0.2,0.2) and (0.5+1.4+0.2+0.1+0.2,-0.2) ..  (0.84+1.4+0.2+0.1+0.2,-0.1);
\node [align=center,align=center] at (1+1.4+0.2+0.1+0.2,0) {$1$};
\node [align=center,align=center] at (3+0.2+0.2+0.1,0) {$+\frac{1}{2}$};
\draw (1+2.6+0.2+0.2+0.1,0) circle [radius=0.2];
\draw (1.2+2.6+0.2+0.2+0.1,0)--(1.4+2.6+0.2+0.2+0.1,0);
\draw (1.6+2.6+0.2+0.2+0.1,0) circle [radius=0.2];
\node [align=center,align=center] at (1+2.6+0.2+0.2+0.1,0) {$1$};
\node [align=center,align=center] at (1.6+2.6+0.2+0.2+0.1,0) {$1$};
\node [align=center,align=center] at (5+0.3,0) {$+\frac{1}{8}$};
\draw (1+4.8+0.3,0) circle [radius=0.2];
\draw (0.84+4.8+0.3,0.1) .. controls (0.5+4.8+0.3,0.2) and (0.5+4.8+0.3,-0.2) ..  (0.84+4.8+0.3,-0.1);
\draw (1.16+4.8+0.3,0.1) .. controls (1.5+4.8+0.3,0.2) and (1.5+4.8+0.3,-0.2) ..  (1.16+4.8+0.3,-0.1);
\node [align=center,align=center] at (1+4.8+0.3,0) {$0$};
\node [align=center,align=center] at (6.6+0.4,0) {$+\frac{1}{2}$};
\draw (1+6.8+0.4,0) circle [radius=0.2];
\draw (0.4+6.8+0.4,0) circle [radius=0.2];
\draw (0.6+6.8+0.4,0)--(0.8+6.8+0.4,0);
\draw (1.16+6.8+0.4,0.1) .. controls (1.5+6.8+0.4,0.2) and (1.5+6.8+0.4,-0.2) ..  (1.16+6.8+0.4,-0.1);
\node [align=center,align=center] at (1+6.8+0.4,0) {$0$};
\node [align=center,align=center] at (0.4+6.8+0.4,0) {$1$};
\node [align=center,align=center] at (8.6+0.5,0) {$+\frac{1}{8}$};
\draw (1+9+0.5,0) circle [radius=0.2];
\draw (0.4+9+0.5,0) circle [radius=0.2];
\draw (0.6+9+0.5,0)--(0.8+9+0.5,0);
\draw (1.16+9+0.5,0.1) .. controls (1.5+9+0.5,0.2) and (1.5+9+0.5,-0.2) ..  (1.16+9+0.5,-0.1);
\draw (0.24+9+0.5,0.1) .. controls (-0.1+9+0.5,0.2) and (-0.1+9+0.5,-0.2) ..  (0.24+9+0.5,-0.1);
\node [align=center,align=center] at (1+9+0.5,0) {$0$};
\node [align=center,align=center] at (0.4+9+0.5,0) {$0$};
\node [align=center,align=center] at (11.4,0) {$+\frac{1}{12}$};
\draw (1+10.2+0.2+0.6,0) circle [radius=0.2];
\draw (1.2+10.2+0.2+0.6,0)--(1.4+10.2+0.2+0.6,0);
\draw (1.16+10.2+0.2+0.6,0.1)--(1.44+10.2+0.2+0.6,0.1);
\draw (1.16+10.2+0.2+0.6,-0.1)--(1.44+10.2+0.2+0.6,-0.1);
\draw (1.6+10.2+0.2+0.6,0) circle [radius=0.2];
\node [align=center,align=center] at (1+10.2+0.2+0.6,0) {$0$};
\node [align=center,align=center] at (1.6+10.2+0.2+0.6,0) {$0$};
\end{tikzpicture},&&
\end{flalign*}

\begin{flalign*}
\begin{tikzpicture}
\draw [dotted,thick](0-0.5,0) circle [radius=0.2];
\node [align=center,align=center] at (0-0.5,0) {$1$};
\draw [dotted,thick](0.16-0.5,0.1) .. controls (0.5-0.5,0.2) and (0.5-0.5,-0.2) ..  (0.16-0.5,-0.1);
\node [align=center,align=center] at (0.8-0.3,0) {$=-$};
\draw (1.3,0) circle [radius=0.2];
\node [align=center,align=center] at (1.3,0) {$1$};
\draw (1.46,0.1) .. controls (1.8,0.2) and (1.8,-0.2) ..  (1.46,-0.1);
\node [align=center,align=center] at (2.1,0) {$-\frac{1}{2}$};
\draw (2.9,0) circle [radius=0.2];
\node [align=center,align=center] at (2.9,0) {$0$};
\draw (3.06,0.1) .. controls (3.4,0.2) and (3.4,-0.2) ..  (3.06,-0.1);
\draw (2.74,0.1) .. controls (2.4,0.2) and (2.4,-0.2) ..  (2.74,-0.1);
\node [align=center,align=center] at (3.7,0) {$-$};
\draw (4.2,0) circle [radius=0.2];
\draw (4.8,0) circle [radius=0.2];
\node [align=center,align=center] at (4.2,0) {$1$};
\node [align=center,align=center] at (4.8,0) {$0$};
\draw (4.96,0.1) .. controls (5.3,0.2) and (5.3,-0.2) ..  (4.96,-0.1);
\draw (4.4,0)--(4.6,0);
\node [align=center,align=center] at (5.6,0) {$-\frac{1}{2}$};
\draw (6.4,0) circle [radius=0.2];
\draw (7,0) circle [radius=0.2];
\node [align=center,align=center] at (6.4,0) {$0$};
\node [align=center,align=center] at (7,0) {$0$};
\draw (6.6,0)--(6.8,0);
\draw (6.24,0.1) .. controls (5.9,0.2) and (5.9,-0.2) ..  (6.24,-0.1);
\draw (7.16,0.1) .. controls (7.5,0.2) and (7.5,-0.2) ..  (7.16,-0.1);
\node [align=center,align=center] at (7.8,0) {$-\frac{1}{2}$};
\draw (8.4,0) circle [radius=0.2];
\draw (9,0) circle [radius=0.2];
\node [align=center,align=center] at (8.4,0) {$0$};
\node [align=center,align=center] at (9,0) {$0$};
\draw (8.6,0)--(8.8,0);
\draw (8.56,0.1)--(8.84,0.1);
\draw (8.56,-0.1)--(8.84,-0.1);
\end{tikzpicture},&&
\end{flalign*}

\begin{flalign*}
\begin{tikzpicture}
\draw [dotted,thick](0-0.5,0) circle [radius=0.2];
\node [align=center,align=center] at (0-0.5,0) {$1$};
\draw [dotted,thick](0.6-0.5,0) circle [radius=0.2];
\node [align=center,align=center] at (0.6-0.5,0) {$1$};
\draw [dotted,thick](0.2-0.5,0)--(0.4-0.5,0);
\node [align=center,align=center] at (1.1-0.2,0) {$=-$};
\draw (1.6,0) circle [radius=0.2];
\node [align=center,align=center] at (1.6,0) {$1$};
\draw (2.2,0) circle [radius=0.2];
\node [align=center,align=center] at (2.2,0) {$1$};
\draw (1.8,0)--(2,0);
\node [align=center,align=center] at (2.7,0) {$-$};
\draw (3.2,0) circle [radius=0.2];
\node [align=center,align=center] at (3.2,0) {$1$};
\draw (3.8,0) circle [radius=0.2];
\node [align=center,align=center] at (3.8,0) {$0$};
\draw (3.4,0)--(3.6,0);
\draw (3.96,0.1) .. controls (4.3,0.2) and (4.3,-0.2) ..  (3.96,-0.1);
\node [align=center,align=center] at (4.6,0) {$-\frac{1}{4}$};
\draw (5.4,0) circle [radius=0.2];
\draw (6,0) circle [radius=0.2];
\node [align=center,align=center] at (5.4,0) {$0$};
\node [align=center,align=center] at (6,0) {$0$};
\draw (5.6,0)--(5.8,0);
\draw (5.24,0.1) .. controls (4.9,0.2) and (4.9,-0.2) ..  (5.24,-0.1);
\draw (6.16,0.1) .. controls (6.5,0.2) and (6.5,-0.2) ..  (6.16,-0.1);
\end{tikzpicture},&&
\end{flalign*}

\begin{flalign*}
\begin{tikzpicture}
\draw [dotted,thick](0,0) circle [radius=0.2];
\node [align=center,align=center] at (0,0) {$1$};
\draw [dotted,thick](0.6,0) circle [radius=0.2];
\node [align=center,align=center] at (0.6,0) {$0$};
\draw [dotted,thick](0.2,0)--(0.4,0);
\draw [dotted,thick](0.76,0.1) .. controls (1.1,0.2) and (1.1,-0.2) ..  (0.76,-0.1);
\node [align=center,align=center] at (1.4,0) {$=$};
\draw (1.9,0) circle [radius=0.2];
\node [align=center,align=center] at (1.9,0) {$1$};
\draw (2.5,0) circle [radius=0.2];
\node [align=center,align=center] at (2.5,0) {$0$};
\draw (2.1,0)--(2.3,0);
\draw (2.66,0.1) .. controls (3,0.2) and (3,-0.2) ..  (2.66,-0.1);
\node [align=center,align=center] at (3.4,0) {$+\frac{1}{2}$};
\draw (4.2,0) circle [radius=0.2];
\node [align=center,align=center] at (4.2,0) {$0$};
\draw (4.8,0) circle [radius=0.2];
\node [align=center,align=center] at (4.8,0) {$0$};
\draw (4.04,0.1) .. controls (3.7,0.2) and (3.7,-0.2) ..  (4.04,-0.1);
\draw (4.4,0)--(4.6,0);
\draw (4.96,0.1) .. controls (5.3,0.2) and (5.3,-0.2) ..  (4.96,-0.1);
\end{tikzpicture},&&
\end{flalign*}

\begin{flalign*}
\begin{tikzpicture}
\draw [dotted,thick](1.6-1.6,0) circle [radius=0.2];
\draw [dotted,thick](-0.16,0.1) .. controls (-0.5,0.2) and (-0.5,-0.2) ..  (-0.16,-0.1);
\draw [dotted,thick](0.16,0.1) .. controls (0.5,0.2) and (0.5,-0.2) ..  (0.16,-0.1);
\node [align=center,align=center] at (1.6-1.6,0) {$0$};
\node [align=center,align=center] at (0.8,0) {$=$};
\draw (1.6,0) circle [radius=0.2];
\draw (1.44,0.1) .. controls (1.1,0.2) and (1.1,-0.2) ..  (1.44,-0.1);
\draw (1.76,0.1) .. controls (2.1,0.2) and (2.12,-0.2) ..  (1.76,-0.1);
\node [align=center,align=center] at (1.6,0) {$0$};
\node [align=center,align=center] at (2.4,0) {$+$};
\draw (3.2,0) circle [radius=0.2];
\draw (3.8,0) circle [radius=0.2];
\draw (3.04,0.1) .. controls (2.7,0.2) and (2.7,-0.2) ..  (3.04,-0.1);
\draw (3.96,0.1) .. controls (4.3,0.2) and (4.3,-0.2) ..  (3.96,-0.1);
\draw (3.4,0)--(3.6,0);
\node [align=center,align=center] at (3.2,0) {$0$};
\node [align=center,align=center] at (3.8,0) {$0$};
\node [align=center,align=center] at (4.7,0) {$+2$};
\draw (5.3,0) circle [radius=0.2];
\draw (5.46,0.1)--(5.74,0.1);
\draw (5.46,-0.1)--(5.74,-0.1);
\draw (5.5,0)--(5.7,0);
\draw (5.9,0) circle [radius=0.2];
\node [align=center,align=center] at (5.3,0) {$0$};
\node [align=center,align=center] at (5.9,0) {$0$};
\end{tikzpicture},&&
\end{flalign*}

\begin{flalign*}
\begin{tikzpicture}
\draw [dotted,thick](0,0) circle [radius=0.2];
\draw [dotted,thick](0.6,0) circle [radius=0.2];
\node [align=center,align=center] at (0,0) {$0$};
\node [align=center,align=center] at (0.6,0) {$0$};
\draw [dotted,thick](0.2,0)--(0.4,0);
\draw [dotted,thick](-0.16,0.1) .. controls (-0.5,0.2) and (-0.5,-0.2) ..  (-0.16,-0.1);
\draw [dotted,thick](0.76,0.1) .. controls (1.1,0.2) and (1.1,-0.2) ..  (0.76,-0.1);
\node [align=center,align=center] at (1.4,0) {$=-$};
\draw (0+2.2,0) circle [radius=0.2];
\draw (0.6+2.2,0) circle [radius=0.2];
\node [align=center,align=center] at (0+2.2,0) {$0$};
\node [align=center,align=center] at (0.6+2.2,0) {$0$};
\draw (0.2+2.2,0)--(0.4+2.2,0);
\draw (-0.16+2.2,0.1) .. controls (-0.5+2.2,0.2) and (-0.5+2.2,-0.2) ..  (-0.16+2.2,-0.1);
\draw (0.76+2.2,0.1) .. controls (1.1+2.2,0.2) and (1.1+2.2,-0.2) ..  (0.76+2.2,-0.1);
\end{tikzpicture},&&
\end{flalign*}

\begin{flalign*}
\begin{tikzpicture}
\draw [dotted,thick](5.3-0.5,0) circle [radius=0.2];
\draw [dotted,thick](5.46-0.5,0.1)--(5.74-0.5,0.1);
\draw [dotted,thick](5.46-0.5,-0.1)--(5.74-0.5,-0.1);
\draw [dotted,thick](5.5-0.5,0)--(5.7-0.5,0);
\draw [dotted,thick](5.9-0.5,0) circle [radius=0.2];
\node [align=center,align=center] at (5.3-0.5,0) {$0$};
\node [align=center,align=center] at (5.9-0.5,0) {$0$};
\node [align=center,align=center] at (6.4-0.2,0) {$=-$};
\draw (5.3+1.6,0) circle [radius=0.2];
\draw (5.46+1.6,0.1)--(5.74+1.6,0.1);
\draw (5.46+1.6,-0.1)--(5.74+1.6,-0.1);
\draw (5.5+1.6,0)--(5.7+1.6,0);
\draw (5.9+1.6,0) circle [radius=0.2];
\node [align=center,align=center] at (5.3+1.6,0) {$0$};
\node [align=center,align=center] at (5.9+1.6,0) {$0$};
\end{tikzpicture}.&&
\end{flalign*}

\end{Example}

Note if $\Gamma^\vee$ is a dotted vertex,
then we have two different ways to express it as a summation over ordinary stable graphs.
The first way is given by \eqref{pre-abs-npt} which involves stable graphs
with no names on the external edges,
the second way is given by \eqref{def-mark}, which involves stable graphs with names on the
external edges,
up to sign and forgetting of the names.
These two definitions involve different kinds of stable graphs and hence different kinds
of automorphism groups.
The following Lemma shows that these two definitions match with each other.

\begin{Lemma}\label{lemma-mark-1}
Denote by $\Gamma$ an ordinary stable graph of type $(g,n)$ without names,
and $S_\Gamma$ the set of ordinary stable graph $\Gamma'$ with names on external edges,
such that $\Gamma$ can be obtained from $\Gamma'$ by forgetting all the names.
Then we have
\ben
\sum_{\Gamma'\in S_\Gamma}\frac{1}{|\Aut(\Gamma')|}
=\frac{n!}{|\Aut(\Gamma)|}.
\een
\end{Lemma}
\begin{proof}
First observe that for every $\Gamma',\Gamma''\in S_\Gamma$,
we have $|\Aut(\Gamma')|=|\Aut(\Gamma'')|$.
Thus we only need to check
\be\label{lem-mark-1}
\frac{|S_\Gamma|}{|\Aut(\Gamma')|}=\frac{n!}{|\Aut(\Gamma)|}
\ee
for an arbitrary $\Gamma'\in S_\Gamma$.

The group $S_n$ acts transitively on the set $S_\Gamma$
by permuting these names, thus
\be\label{lem-mark-2}
|S_\Gamma|=\frac{n!}{|H|},
\ee
where $H$ consists of the permutations of the names that preserves $\Gamma'$.

On the other hand, the group $\Aut(\Gamma)$ acts on $H$,
and the invariant subgroup is exactly isomorphic to $\Aut(\Gamma')$,
thus
\be\label{lem-mark-3}
\frac{|\Aut(\Gamma)|}{|\Aut(\Gamma')|}=|H|.
\ee
Then \eqref{lem-mark-2} and \eqref{lem-mark-3} give us \eqref{lem-mark-1}.
\end{proof}

\subsection{Duality map defined by dotted stable graphs}
\label{subsec3.2}

As in \cite{wz},
let $\cV_{g,n}$ be the linear space spanned by the basis consisting of
all ordinary stable graphs of genus $g$ with $n$ external edges.
We define a partial ordering on this basis as follows.
Let $\Gamma_i$ be an ordinary stable graph of type $(g,n)$ for $i=1,2$,
then we set $\Gamma_1>\Gamma_2$ if $|E(\Gamma_1)|<|E(\Gamma_2)|$,
where $E(\Gamma_i)$ is the set of all internal edges of $\Gamma_i$.
Since $\cV_{g,n}$ is finite-dimensional for a fixed pair $(g,n)$,
we can extend this partial ordering to a total ordering
by assign an arbitrary ordering among stable graphs with
the same number of internal edges.

There is a natural linear map
\be
\phi_{g,n}:\cV_{g,n}\to\cV_{g,n},
\ee
which simply changes an ordinary vertex of genus $g$ and valence $n$
to a dotted vertex of genus $g$ and valence $n$,
and change an internal edge to a dotted internal edge.
For example,
\ben
&&\begin{tikzpicture}
\draw (0,0) circle [radius=0.2];
\node [align=center,align=center] at (0,0) {$2$};
\node [align=center,align=center] at (0.6,0) {$\mapsto$};
\draw [dotted,thick](1.2,0) circle [radius=0.2];
\node [align=center,align=center] at (1.2,0) {$2$};
\end{tikzpicture},\\
&&\begin{tikzpicture}
\draw (0,0) circle [radius=0.2];
\node [align=center,align=center] at (0,0) {$1$};
\draw (0.6,0) circle [radius=0.2];
\node [align=center,align=center] at (0.6,0) {$1$};
\draw (0.2,0)--(0.4,0);
\node [align=center,align=center] at (1.2,0) {$\mapsto$};
\draw [dotted,thick](2.2-0.3,0) circle [radius=0.2];
\node [align=center,align=center] at (2.2-0.3,0) {$1$};
\draw [dotted,thick](2.8-0.3,0) circle [radius=0.2];
\node [align=center,align=center] at (2.8-0.3,0) {$1$};
\draw [dotted,thick](2.4-0.3,0)--(2.6-0.3,0);
\end{tikzpicture},\\
&&\begin{tikzpicture}
\draw (0,0) circle [radius=0.2];
\node [align=center,align=center] at (0,0) {$1$};
\draw (0.6,0) circle [radius=0.2];
\node [align=center,align=center] at (0.6,0) {$0$};
\draw (0.2,0)--(0.4,0);
\draw (0.76,0.1) .. controls (1.1,0.2) and (1.1,-0.2) ..  (0.76,-0.1);
\node [align=center,align=center] at (1.5,0) {$\mapsto$};
\draw [dotted,thick](0+2.2,0) circle [radius=0.2];
\node [align=center,align=center] at (0+2.2,0) {$1$};
\draw [dotted,thick](0.6+2.2,0) circle [radius=0.2];
\node [align=center,align=center] at (0.6+2.2,0) {$0$};
\draw [dotted,thick](0.2+2.2,0)--(0.4+2.2,0);
\draw [dotted,thick](0.76+2.2,0.1) .. controls (1.1+2.2,0.2) and (1.1+2.2,-0.2) ..  (0.76+2.2,-0.1);
\end{tikzpicture},\\
&&\begin{tikzpicture}
\draw (0,0) circle [radius=0.2];
\node [align=center,align=center] at (0,0) {$0$};
\draw (0.6,0) circle [radius=0.2];
\node [align=center,align=center] at (0.6,0) {$0$};
\draw (0.2,0)--(0.4,0);
\draw (0.76,0.1) .. controls (1.1,0.2) and (1.1,-0.2) ..  (0.76,-0.1);
\draw (-0.16,0.1) .. controls (-0.5,0.2) and (-0.5,-0.2) ..  (-0.16,-0.1);
\node [align=center,align=center] at (1.5,0) {$\mapsto$};
\draw [dotted,thick](0+2.7-0.3,0) circle [radius=0.2];
\node [align=center,align=center] at (0+2.7-0.3,0) {$0$};
\draw [dotted,thick](0.6+2.7-0.3,0) circle [radius=0.2];
\node [align=center,align=center] at (0.6+2.7-0.3,0) {$0$};
\draw [dotted,thick](0.2+2.7-0.3,0)--(0.4+2.7-0.3,0);
\draw [dotted,thick](0.76+2.7-0.3,0.1) .. controls (1.1+2.7-0.3,0.2) and (1.1+2.7-0.3,-0.2) ..  (0.76+2.7-0.3,-0.1);
\draw [dotted,thick](-0.16+2.7-0.3,0.1) .. controls (-0.5+2.7-0.3,0.2) and (-0.5+2.7-0.3,-0.2) ..  (-0.16+2.7-0.3,-0.1);
\end{tikzpicture}.
\een
An easy observation is that for every ordinary stable graph $\Gamma'$
appearing in $\phi_{g,n}(\Gamma)$,
we have $|E(\Gamma')|\geq |E(\Gamma)|$,
and $|E(\Gamma')|=|E(\Gamma)|$ if and only if $\Gamma'=\Gamma$.
Therefore using the total ordering defined above,
the endomorphism $\phi_{g,n}$ is a upper-triangular matrix
with $\pm 1$'s on the diagonal in such a basis.
Therefore $\phi_{g,n}$ is an isomorphism of the vector space $\cV_{g,n}$,
and all dotted stable graphs of type $(g,n)$ form another basis for $\cV_{g,n}$.
Let $\phi:\cV \to \cV$ be defined by:
\be\label{eq-phi-duality}
\phi:=\bigoplus_{2g-2+n > 0} \phi_{g,n}.
\ee

Now we can write the ordinary stable graphs in terms of linear combinations of
dotted stable graphs.

\begin{Example}\label{eg-inverse}
The following relations can be checked by explicit computations:
\begin{flalign*}
\begin{tikzpicture}
\draw (0,0) circle [radius=0.2];
\node [align=center,align=center] at (0,0) {$0$};
\draw (-0.5,0)--(-0.2,0);
\draw (0.16,0.1)--(0.5,0.15);
\draw (0.16,-0.1)--(0.5,-0.15);
\node [align=center,align=center] at (0.8,0) {$=$};
\draw [dotted,thick](0+1.6,0) circle [radius=0.2];
\node [align=center,align=center] at (0+1.6,0) {$0$};
\draw [dotted,thick](-0.5+1.6,0)--(-0.2+1.6,0);
\draw [dotted,thick](0.16+1.6,0.1)--(0.5+1.6,0.15);
\draw [dotted,thick](0.16+1.6,-0.1)--(0.5+1.6,-0.15);
\end{tikzpicture},&&
\end{flalign*}

\begin{flalign*}
\begin{tikzpicture}
\draw (1.6-1.6,0) circle [radius=0.2];
\draw (1.1-1.6,0.15)--(1.44-1.6,0.1);
\draw (1.1-1.6,-0.15)--(1.44-1.6,-0.1);
\draw (1.76-1.6,0.1)--(2.1-1.6,0.15);
\draw (1.76-1.6,-0.1)--(2.1-1.6,-0.15);
\node [align=center,align=center] at (1.6-1.6,0) {$0$};
\node [align=center,align=center] at (0.8,0) {$=$};
\draw [dotted,thick](1.6+0.1,0) circle [radius=0.2];
\draw [dotted,thick](1.1+0.1,0.15)--(1.44+0.1,0.1);
\draw [dotted,thick](1.1+0.1,-0.15)--(1.44+0.1,-0.1);
\draw [dotted,thick](1.76+0.1,0.1)--(2.1+0.1,0.15);
\draw [dotted,thick](1.76+0.1,-0.1)--(2.1+0.1,-0.15);
\node [align=center,align=center] at (1.6+0.1,0) {$0$};
\node [align=center,align=center] at (2.6,0) {$+3$};
\draw [dotted,thick](1.6+1.9,0) circle [radius=0.2];
\draw [dotted,thick](2.2+1.9,0) circle [radius=0.2];
\draw [dotted,thick](1.1+1.9,0.15)--(1.44+1.9,0.1);
\draw [dotted,thick](1.1+1.9,-0.15)--(1.44+1.9,-0.1);
\draw [dotted,thick](1.1+1.9,0.15)--(1.44+1.9,0.1);
\draw [dotted,thick](2.36+1.9,0.1)--(2.7+1.9,0.15);
\draw [dotted,thick](2.36+1.9,-0.1)--(2.7+1.9,-0.15);
\draw [dotted,thick](1.8+1.9,0)--(2+1.9,0);
\node [align=center,align=center] at (1.6+1.9,0) {$0$};
\node [align=center,align=center] at (2.2+1.9,0) {$0$};
\end{tikzpicture},&&
\end{flalign*}

\begin{flalign*}
\begin{tikzpicture}
\draw (1-1.6,0) circle [radius=0.2];
\draw (1.2-1.6,0)--(1.5-1.6,0);
\node [align=center,align=center] at (1-1.6,0) {$1$};
\node [align=center,align=center] at (0.3,0) {$=$};
\draw [dotted,thick](1,0) circle [radius=0.2];
\draw [dotted,thick](1.2,0)--(1.5,0);
\node [align=center,align=center] at (1,0) {$1$};
\node [align=center,align=center] at (2,0) {$+\frac{1}{2}$};
\draw [dotted,thick](1+1.8,0) circle [radius=0.2];
\draw [dotted,thick](1.2+1.8,0)--(1.5+1.8,0);
\draw [dotted,thick](0.84+1.8,0.1) .. controls (0.5+1.8,0.2) and (0.5+1.8,-0.2) ..  (0.84+1.8,-0.1);
\node [align=center,align=center] at (1+1.8,0) {$0$};
\end{tikzpicture},&&
\end{flalign*}

\begin{flalign*}
\begin{tikzpicture}
\draw (1-2+0.2,0) circle [radius=0.2];
\draw (1.2-2+0.2,0)--(1.5-2+0.2,0);
\draw (0.5-2+0.2,0)--(0.8-2+0.2,0);
\node [align=center,align=center] at (1-2+0.2,0) {$1$};
\node [align=center,align=center] at (0.3-0.2,0) {$=$};
\draw [dotted,thick](1,0) circle [radius=0.2];
\draw [dotted,thick](1.2,0)--(1.5,0);
\draw [dotted,thick](0.5,0)--(0.8,0);
\node [align=center,align=center] at (1,0) {$1$};
\node [align=center,align=center] at (2,0) {$+\frac{1}{2}$};
\draw [dotted,thick](1+1.8,0) circle [radius=0.2];
\draw [dotted,thick](1.17+1.8,0.1)--(1.4+1.8,0.15);
\draw [dotted,thick](1.17+1.8,-0.1)--(1.4+1.8,-0.15);
\draw [dotted,thick](0.84+1.8,0.1) .. controls (0.5+1.8,0.2) and (0.5+1.8,-0.2) ..  (0.84+1.8,-0.1);
\node [align=center,align=center] at (1+1.8,0) {$0$};
\node [align=center,align=center] at (3.7,0) {$+$};
\draw [dotted,thick](1+3.9,0) circle [radius=0.2];
\draw [dotted,thick](0.4+3.9,0) circle [radius=0.2];
\draw [dotted,thick](1.17+3.9,0.1)--(1.4+3.9,0.15);
\draw [dotted,thick](1.17+3.9,-0.1)--(1.4+3.9,-0.15);
\draw [dotted,thick](0.6+3.9,0)--(0.8+3.9,0);
\node [align=center,align=center] at (1+3.9,0) {$0$};
\node [align=center,align=center] at (0.4+3.9,0) {$1$};
\node [align=center,align=center] at (5.8,0) {$+\frac{1}{2}$};
\draw [dotted,thick](1+6.2,0) circle [radius=0.2];
\draw [dotted,thick](0.4+6.2,0) circle [radius=0.2];
\draw [dotted,thick](1.17+6.2,0.1)--(1.4+6.2,0.15);
\draw [dotted,thick](1.17+6.2,-0.1)--(1.4+6.2,-0.15);
\draw [dotted,thick](0.6+6.2,0)--(0.8+6.2,0);
\node [align=center,align=center] at (1+6.3-0.1,0) {$0$};
\node [align=center,align=center] at (0.4+6.3-0.1,0) {$0$};
\draw [dotted,thick](0.24+6.3-0.1,0.1) .. controls (-0.1+6.3-0.1,0.2) and (-0.1+6.3-0.1,-0.2) ..  (0.24+6.3-0.1,-0.1);
\node [align=center,align=center] at (8.2-0.1,0) {$+\frac{1}{2}$};
\draw [dotted,thick](1+8.1-0.1,0) circle [radius=0.2];
\draw [dotted,thick](0.5+8.1-0.1,0)--(0.8+8.1-0.1,0);
\draw [dotted,thick](1.18+8.1-0.1,0.07)--(1.42+8.1-0.1,0.07);
\draw [dotted,thick](1.18+8.1-0.1,-0.07)--(1.42+8.1-0.1,-0.07);
\draw [dotted,thick](1.8+8.1-0.1,-0)--(2.1+8.1-0.1,0);
\draw [dotted,thick](1.6+8.1-0.1,0) circle [radius=0.2];
\node [align=center,align=center] at (1+8.1-0.1,0) {$0$};
\node [align=center,align=center] at (1.6+8.1-0.1,0) {$0$};
\end{tikzpicture}.&&
\end{flalign*}

\end{Example}

From these examples we see a surprising duality:
$\phi_{g,n}$ is an involution,
i.e. $\phi_{g,n}^2 = \Id$.
In other words,
the coefficients in the identities that express ordinary stable graphs
in terms of dotted stable graphs coincide with the coefficients
in the identities that express dotted stable graphs in terms of ordinary stable graphs.
To prove this result,
we need to first consider the dual versions of the operators $K$, $\cD$, $\pd$ and $\gamma$
introduced in \cite{wz} (see also \S \ref{sec-pre}).

\subsection{Operators for the dotted stable graphs}

Let $K^\vee$ be the operator that acts on $\cV$ defined
by cutting one of the internal edges of a dotted stable graph
and sums over all internal edges.
It will be referred to as the {\em dual edge-cutting operator}.
It is clear from the definition of $\phi$ that
\be\label{phi-K}
\phi(K\Gamma)=K^\vee\phi(\Gamma).
\ee

Let us now introduce the operators $\cD^\vee$ and $\pd^\vee$ acting on
the dotted stable graphs. The definitions are identical to
that of the operators $\cD$ and $\pd$ for ordinary stable graphs,
except that one just replace ordinary graphs by dotted graphs.

\begin{Definition}
The operator $\pd^\vee$ acting on a dotted stable graph consists of two parts,
one is to add an external edge to a dotted vertex and sum over all dotted vertices,
and the other is  to insert a trivalent dotted vertex of genus $0$
on an internal edge and
then sum over all internal edges.

The operator $\gamma^\vee$ acts on a dotted stable graph by attaching a trivalent
dotted vertex of genus $0$ to an external edge and sum over all external edges.

The operator $\cD^\vee$ is defined to be $\cD^\vee:=\pd^\vee+\gamma^\vee$.
\end{Definition}

From these definitions and the definition of $\phi$,
it is clear that  for every ordinary
stable graph $\Gamma$, we have
\be\label{phi-pd&cD}
\phi(\pd\Gamma)=\pd^\vee\phi(\Gamma),\quad
\phi(\cD\Gamma)=\cD^\vee\phi(\Gamma).
\ee

\begin{Example}
Here are some examples for the operators $\pd^\vee$ and $\gamma^\vee$.
\begin{flalign*}
\begin{tikzpicture}
\node [align=center,align=center] at (-0.4,0) {$\pd^\vee$};
\draw [dotted,thick](1,0) circle [radius=0.2];
\draw [dotted,thick](0.4,0) circle [radius=0.2];
\draw [dotted,thick](0.6,0)--(0.8,0);
\draw [dotted,thick](1.16,0.1) .. controls (1.5,0.2) and (1.5,-0.2) ..  (1.16,-0.1);
\draw [dotted,thick](0.24,0.1) .. controls (-0.1,0.2) and (-0.1,-0.2) ..  (0.24,-0.1);
\node [align=center,align=center] at (1,0) {$0$};
\node [align=center,align=center] at (0.4,0) {$0$};
\node [align=center,align=center] at (1.8,0) {$=2$};
\draw [dotted,thick](3.2,0) circle [radius=0.2];
\draw [dotted,thick](2.6,0) circle [radius=0.2];
\draw [dotted,thick](2.8,0)--(3,0);
\draw [dotted,thick](2.6,0.2)--(2.6,0.4);
\draw [dotted,thick](3.36,0.1) .. controls (3.7,0.2) and (3.7,-0.2) ..  (3.36,-0.1);
\draw [dotted,thick](2.44,0.1) .. controls (2.1,0.2) and (2.1,-0.2) ..  (2.44,-0.1);
\node [align=center,align=center] at (2.6,0) {$0$};
\node [align=center,align=center] at (3.2,0) {$0$};
\node [align=center,align=center] at (4,0) {$+2$};
\draw [dotted,thick](5.4,0) circle [radius=0.2];
\draw [dotted,thick](4.8,0) circle [radius=0.2];
\draw [dotted,thick](5,0)--(5.2,0);
\draw [dotted,thick](4.64,0.1) .. controls (4.3,0.2) and (4.3,-0.2) ..  (4.64,-0.1);
\draw [dotted,thick](5.58,0.07)--(5.82,0.07);
\draw [dotted,thick](5.58,-0.07)--(5.82,-0.07);
\draw [dotted,thick](6.2,-0)--(6.5,0);
\draw [dotted,thick](6,0) circle [radius=0.2];
\node [align=center,align=center] at (5.4,0) {$0$};
\node [align=center,align=center] at (4.8,0) {$0$};
\node [align=center,align=center] at (6,0) {$0$};
\node [align=center,align=center] at (7,0) {$+$};
\draw [dotted,thick](8.4,0) circle [radius=0.2];
\draw [dotted,thick](7.8,0) circle [radius=0.2];
\draw [dotted,thick](9,0) circle [radius=0.2];
\draw [dotted,thick](8,0)--(8.2,0);
\draw [dotted,thick](8.6,0)--(8.8,0);
\draw [dotted,thick](8.4,0.2)--(8.4,0.4);
\draw [dotted,thick](9.16,0.1) .. controls (9.5,0.2) and (9.5,-0.2) ..  (9.16,-0.1);
\draw [dotted,thick](7.64,0.1) .. controls (7.3,0.2) and (7.3,-0.2) ..  (7.64,-0.1);
\node [align=center,align=center] at (8.4,0) {$0$};
\node [align=center,align=center] at (7.8,0) {$0$};
\node [align=center,align=center] at (9,0) {$0$};
\end{tikzpicture},&&
\end{flalign*}

\begin{flalign*}
\begin{tikzpicture}
\node [align=center,align=center] at (0.4,0) {$\pd^\vee$};
\draw [dotted,thick](1,0) circle [radius=0.2];
\draw [dotted,thick](1.2,0)--(1.4,0);
\draw [dotted,thick](1.16,0.1)--(1.44,0.1);
\draw [dotted,thick](1.16,-0.1)--(1.44,-0.1);
\draw [dotted,thick](1.6,0) circle [radius=0.2];
\node [align=center,align=center] at (1,0) {$0$};
\node [align=center,align=center] at (1.6,0) {$1$};
\node [align=center,align=center] at (2.1,0) {$=$};
\draw [dotted,thick](3,0) circle [radius=0.2];
\draw [dotted,thick](3.2,0)--(3.4,0);
\draw [dotted,thick](2.5,0)--(2.8,0);
\draw [dotted,thick](3.16,0.1)--(3.44,0.1);
\draw [dotted,thick](3.16,-0.1)--(3.44,-0.1);
\draw [dotted,thick](3.6,0) circle [radius=0.2];
\node [align=center,align=center] at (3,0) {$0$};
\node [align=center,align=center] at (3.6,0) {$1$};
\node [align=center,align=center] at (4.1,0) {$+$};
\draw [dotted,thick](4.7,0) circle [radius=0.2];
\draw [dotted,thick](4.9,0)--(5.1,0);
\draw [dotted,thick](5.5,0)--(5.8,0);
\draw [dotted,thick](4.86,0.1)--(5.14,0.1);
\draw [dotted,thick](4.86,-0.1)--(5.14,-0.1);
\draw [dotted,thick](5.3,0) circle [radius=0.2];
\node [align=center,align=center] at (4.7,0) {$0$};
\node [align=center,align=center] at (5.3,0) {$1$};
\node [align=center,align=center] at (6.3,0) {$+3$};
\draw [dotted,thick](6.8,-0.2) circle [radius=0.2];
\draw [dotted,thick](6.98,-0.13)--(7.42,-0.13);
\draw [dotted,thick](6.98,-0.27)--(7.42,-0.27);
\draw [dotted,thick](7.6,-0.2) circle [radius=0.2];
\draw [dotted,thick](7.2,0.25) circle [radius=0.2];
\draw [dotted,thick](6.94,-0.06)--(7.06,0.11);
\draw [dotted,thick](7.46,-0.06)--(7.34,0.11);
\node [align=center,align=center] at (6.8,-0.2) {$0$};
\node [align=center,align=center] at (7.6,-0.2) {$1$};
\node [align=center,align=center] at (7.2,0.25) {$0$};
\draw [dotted,thick](7.2,0.45)--(7.2,0.65);
\end{tikzpicture}.&&
\end{flalign*}

\begin{flalign*}
\begin{tikzpicture}
\node [align=center,align=center] at (0.4,0) {$\gamma^\vee$};
\draw [dotted,thick](1,0) circle [radius=0.2];
\draw [dotted,thick](1.2,0)--(1.5,0);
\node [align=center,align=center] at (1,0) {$1$};
\node [align=center,align=center] at (1.9,0) {$=$};
\draw [dotted,thick](2.4,0) circle [radius=0.2];
\draw [dotted,thick](3,0) circle [radius=0.2];
\draw [dotted,thick](2.6,0)--(2.8,0);
\draw [dotted,thick](3.16,0.1)--(3.4,0.15);
\draw [dotted,thick](3.16,-0.1)--(3.4,-0.15);
\node [align=center,align=center] at (2.4,0) {$1$};
\node [align=center,align=center] at (3,0) {$0$};
\end{tikzpicture},&&
\end{flalign*}

\begin{flalign*}
\begin{tikzpicture}
\node [align=center,align=center] at (0.1,0) {$\gamma^\vee$};
\draw [dotted,thick](0.9,0) circle [radius=0.2];
\draw [dotted,thick](1.06,0.1)--(1.3,0.15);
\draw [dotted,thick](1.06,-0.1)--(1.3,-0.15);
\draw [dotted,thick](0.74,0.1) .. controls (0.4,0.2) and (0.4,-0.2) ..  (0.74,-0.1);
\node [align=center,align=center] at (0.9,0) {$0$};
\node [align=center,align=center] at (1.8,0) {$=2$};
\draw [dotted,thick](2.7,0) circle [radius=0.2];
\draw [dotted,thick](3.3,0) circle [radius=0.2];
\draw [dotted,thick](2.9,0)--(3.1,0);
\draw [dotted,thick](2.7,0.2)--(2.7,0.4);
\draw [dotted,thick](3.47,0.1)--(3.7,0.15);
\draw [dotted,thick](3.47,-0.1)--(3.7,-0.15);
\draw [dotted,thick](2.54,0.1) .. controls (2.2,0.2) and (2.2,-0.2) ..  (2.54,-0.1);
\node [align=center,align=center] at (2.7,0) {$0$};
\node [align=center,align=center] at (3.3,0) {$0$};
\end{tikzpicture}.&&
\end{flalign*}

\end{Example}

Our next goal is to study the relations between $\cD^\vee$, $\pd^\vee$ and
the ordinary operators $\cD$, $\pd$.
First, we have the following:

\begin{Lemma}
\label{lem-gamma}
We have $\gamma=-\gamma^\vee$.
\end{Lemma}

This lemma holds trivially by \eqref{eg1-eq1}.

Next, consider the actions of the ordinary operators $\cD$ and $\pd$
on a dotted stable vertex.
Recall that the dotted stable vertex of genus $g$ with $n$ external edges
is defined to be $n!\cdot\wcF_{g,n}$, thus Lemma \ref{lem-original-D} gives
\ben
\begin{tikzpicture}
\node [align=center,align=center] at (-0.6,0) {$\cD\biggl($};
\draw [dotted,thick](0,0) circle [radius=0.2];
\node [align=center,align=center] at (0,0) {$g$};
\draw [dotted,thick](0.16,0.1)--(0.5,0.35);
\draw [dotted,thick](0.16,-0.1)--(0.5,-0.35);
\node [align=center,align=center] at (0.5,0.1) {$\vdots$};
\node [align=center,align=center] at (0.7,-0.15) {$n$};
\node [align=center,align=center] at (1.2,0) {$\biggr)=$};
\draw [dotted,thick](2,0) circle [radius=0.2];
\node [align=center,align=center] at (2,0) {$g$};
\draw [dotted,thick](2.16,0.1)--(2.5,0.35);
\draw [dotted,thick](2.16,-0.1)--(2.5,-0.35);
\node [align=center,align=center] at (2.5,0.1) {$\vdots$};
\node [align=center,align=center] at (3,-0.15) {$n+1$};
\node [align=center,align=center] at (4.1,0) {$=\pd^\vee\biggl($};
\draw [dotted,thick](4.9,0) circle [radius=0.2];
\node [align=center,align=center] at (0+4.9,0) {$g$};
\draw [dotted,thick](0.16+4.9,0.1)--(0.5+4.9,0.35);
\draw [dotted,thick](0.16+4.9,-0.1)--(0.5+4.9,-0.35);
\node [align=center,align=center] at (0.5+4.9,0.1) {$\vdots$};
\node [align=center,align=center] at (0.7+4.9,-0.15) {$n$};
\node [align=center,align=center] at (5.9,0) {$\biggr)$};
\end{tikzpicture}.
\een
Therefore, the operator $\pd=\cD-\gamma=\cD+\gamma^\vee$
acts on this dotted stable vertex by
\ben
\begin{tikzpicture}
\node [align=center,align=center] at (-0.6,0) {$\pd\biggl($};
\draw [dotted,thick](0,0) circle [radius=0.2];
\node [align=center,align=center] at (0,0) {$g$};
\draw [dotted,thick](0.16,0.1)--(0.5,0.35);
\draw [dotted,thick](0.16,-0.1)--(0.5,-0.35);
\node [align=center,align=center] at (0.5,0.1) {$\vdots$};
\node [align=center,align=center] at (0.65,-0.15) {$n$};
\node [align=center,align=center] at (1.2,0) {$\biggr)=$};
\draw [dotted,thick](2,0) circle [radius=0.2];
\node [align=center,align=center] at (2,0) {$g$};
\draw [dotted,thick](2.16,0.1)--(2.5,0.35);
\draw [dotted,thick](2.16,-0.1)--(2.5,-0.35);
\node [align=center,align=center] at (2.5,0.1) {$\vdots$};
\node [align=center,align=center] at (2.95,-0.15) {$n+1$};
\node [align=center,align=center] at (3.95,0) {$+n\biggr($};
\draw [dotted,thick](4.8,0) circle [radius=0.2];
\node [align=center,align=center] at (4.8,0) {$0$};
\draw [dotted,thick](5,0)--(5.2,0);
\draw [dotted,thick](4.64,0.1)--(4.3,0.15);
\draw [dotted,thick](4.64,-0.1)--(4.3,-0.15);
\draw [dotted,thick](2+3.4,0) circle [radius=0.2];
\node [align=center,align=center] at (2+3.4,0) {$g$};
\draw [dotted,thick](2.16+3.4,0.1)--(2.5+3.4,0.35);
\draw [dotted,thick](2.16+3.4,-0.1)--(2.5+3.4,-0.35);
\node [align=center,align=center] at (2.5+3.4,0.1) {$\vdots$};
\node [align=center,align=center] at (3+3.35,-0.15) {$n-1$};
\node [align=center,align=center] at (6.9,0) {$\biggr)$};
\node [align=center,align=center] at (7.7,0) {$=\cD^\vee\biggl($};
\draw [dotted,thick](0+8.5,0) circle [radius=0.2];
\node [align=center,align=center] at (0+8.5,0) {$g$};
\draw [dotted,thick](0.16+8.5,0.1)--(0.5+8.5,0.35);
\draw [dotted,thick](0.16+8.5,-0.1)--(0.5+8.5,-0.35);
\node [align=center,align=center] at (0.5+8.5,0.1) {$\vdots$};
\node [align=center,align=center] at (0.65+8.5,-0.15) {$n$};
\node [align=center,align=center] at (9.5,0) {$\biggr)$};
\end{tikzpicture}.
\een
In general, we will prove the following:

\begin{Theorem}\label{thm1}
We have
\be\label{thm1-eq}
\cD^\vee=\pd,\qquad \pd^\vee=\cD.
\ee
\end{Theorem}

\begin{proof}
We already know that
\ben
\cD^\vee\Gamma^\vee=\pd\Gamma^\vee,
\qquad \pd^\vee\Gamma^\vee=\cD\Gamma^\vee
\een
for the simplest case where $\Gamma^\vee$ is a dotted stable vertex.

Now consider an arbitrary dotted stable graph $\Gamma^\vee$.
Let us recall the definition of $\Gamma^\vee$ expressed as a linear combination of
stable graphs.
Denote by $V^\vee(\Gamma^\vee)$ the set of all its vertices.
We give different names to all the half-edges of $\Gamma^\vee$
to get a dotted stable graphs with names $\Gamma'^\vee$.
Then cut off all internal edges of $\Gamma'^\vee$,
we obtain a (not necessarily connected) dotted stable graph
\ben
\bigsqcup_{v\in V^\vee(\Gamma^\vee)}V_v^\vee,
\een
where $V_v^\vee$ is a dotted stable vertex with names
of genus $g_v$ and valence $n_v$.
Then
\ben
V_v^\vee=n_v!\cdot\sum_{\Gamma\in \cG_v}
\frac{1}{|\Aut(\Gamma)|}\Gamma,
\een
where $\cG_v$ is the set of all ordinary stable graphs of
genus $g_v$ with $n_v$ external edges which have the same names
with the dotted vertex $V_v^\vee$.
By definition $\Gamma^\vee$ is obtained by gluing two external edges together
if their ancestors in $\Gamma^\vee$ are joined together,
then forgetting the names and multiplying by $(-1)^{|E^\vee(\Gamma^\vee)|}$.

Let $\Gamma_v^\vee$ be the graph obtained from $\Gamma^\vee$ by
adding an external edge to the dotted vertex $v\in V^\vee(\Gamma^\vee)$.
And for every $e\in E^\vee(\Gamma^\vee)$, let $\Gamma_e^\vee$ be
the dotted stable graph obtained from $\Gamma^\vee$ by replacing $e$
with a trivalent dotted vertex of genus $0$ together with
two additional internal edges.
Then
\be\label{thm1-eq1}
\pd^\vee\Gamma^\vee=\sum_{v\in V^\vee(\Gamma^\vee)}\Gamma_v^\vee
+\sum_{e\in E^\vee(\Gamma^\vee)}\Gamma_e^\vee.
\ee

We cut off all internal edges and give different names to $\Gamma_v^\vee$.
The result is a bunch  of dotted stable vertices with names.
We denote them by $\{V_{v,v'}^\vee\}_{v'\in V^\vee(\Gamma^\vee)}$.
Notice that each $V_{v,v'}^\vee$ comes from a vertex $v'$ of $\Gamma^\vee$,
and if $v'\not= v$,
then $V_{v,v'}^\vee$ is just $V_{v'}^\vee$,
but $V_{v,v}^\vee$ is obtained from $V_v^\vee$ by adding  one more external edge.
By Lemma \ref{lem-original-D} we know that
\ben
V_{v,v}^\vee=\cD_m V_v^\vee
=(\pd_m+\gamma_m) V_v^\vee,
\een
where $\cD_m,\pd_m$ is to apply $\cD,\pd$ respectively
and then give the new external edge a different name,
and $\gamma_m$ is to apply $\gamma$, then
move the original name of this external edge of $V_v^\vee$
to one of the new external edge,
and the other new external edge and the new internal edge
will be given three new names such that
all names appearing on this graph are different.
Moreover, the names of external edges should match with $V_{v,v}^\vee$.
Then the dotted graph $\Gamma_v^\vee$ is obtained by gluing
$\{V_{v'}^\vee\}_{v'\not= v}$ and $(\pd_m+\gamma_m) V_v^\vee$
together according to the information of names,
then forgetting the names and multiplying by $(-1)^{|E^\vee(V_v^\vee)|}$.

On the other hand, we can regard $\Gamma^\vee$ as a linear combination of
ordinary stable graphs and analyse
$\cD\Gamma^\vee=\pd\Gamma^\vee+\gamma\Gamma^\vee$.

Notice that $\pd$ is a `differential operator'
in the sense that it acts on an ordinary stable vertex
as adding an external edge, and acts on an internal edge as
replacing it with a trivalent vertex of genus $0$
together with two additional internal edges;
moreover, it acts on an ordinary stable graph $\Gamma$
via Leibniz rule if we regard $\Gamma$ as a product of
all vertices $v\in V(\Gamma)$ and
all internal edges $e\in E(\Gamma)$.

Now we write $\Gamma^\vee$ in terms of a linear combination of
ordinary stable graphs, then we have a new natural way
to factorize each term $\Gamma$;
that is, we regard a term in one of the dotted stable vertex
as a factor of $\Gamma$,
then $\Gamma$ can be regarded as a product of such factors and
internal edges $e\in E^\vee(\Gamma^\vee)$.
Thus Leibniz rule gives us
\ben
\pd_m\Gamma'^\vee=\sum_{v\in V^\vee(\Gamma^\vee)} \Gamma'^\vee_v
+\sum_{e\in E^\vee(\Gamma^\vee)}\Gamma'^\vee_e,
\een
where $\Gamma_v'^\vee$ is obtained by gluing the dotted vertices with names
$\{V_{v'}^\vee\}_{v'\in V^\vee(\Gamma^\vee)}$ and $\pd_m V_v^\vee$ together,
and $\Gamma'^\vee_e$ is just $\Gamma^\vee_e$ with different names
on half-edges. Therefore
\be
\cD_m\Gamma'^\vee=\sum_{v\in V^\vee(\Gamma^\vee)} \Gamma'^\vee_v
+\sum_{e\in E^\vee(\Gamma^\vee)}\Gamma'^\vee_e
+\gamma_m \Gamma'^\vee.
\ee

Recall that when we forget the names,
we need an additional factor $(-1)^{|E^\vee|}$.
Therefore the above relation gives
\be\label{thm1-eq2}
\cD\Gamma^\vee=\sum_{v\in V^\vee(\Gamma^\vee)} \widetilde\Gamma^\vee_v
-\sum_{e\in E^\vee(\Gamma^\vee)}\Gamma^\vee_e
-\gamma \Gamma^\vee,
\ee
where $\widetilde\Gamma^\vee_v$ is obtained by gluing
$\{V_{v,v'}^\vee\}_{v'\not= v}$ and $\pd_m V_v^\vee$ together
according to the information of names,
then forgetting the names and multiplying by $(-1)^{|E^\vee(V_v^\vee)|}$.

Now we compare the two relations \eqref{thm1-eq1} and \eqref{thm1-eq2}.
In order to prove $\pd^\vee\Gamma^\vee=\cD\Gamma^\vee$,
all we need to do is to check
\be\label{claim-proof1}
\sum_{v\in V^\vee(\Gamma^\vee)}(\widetilde\Gamma_v^\vee
-\Gamma_v^\vee)
=2 \sum_{e\in E^\vee(\Gamma^\vee)}\Gamma^\vee_e+\gamma\Gamma^\vee.
\ee

We already know that $(\widetilde\Gamma_v^\vee
-\Gamma_v^\vee)$ is obtained by gluing
$\{V_{v'}^\vee\}_{v'\not= v}$ and $\gamma_m V_v^\vee$ together
according to the information of names,
then forgetting the names and multiplying by $(-1)^{|E^\vee|}$.
Recall that the operator $\gamma$ is to attach a trivalent vertex of genus $0$
to an external edge,
and the set of all external edges of the graph
$\bigsqcup_{v\in V^\vee(\Gamma^\vee)}V_v^\vee$ equals the union of
the set of all external edges of $\Gamma^\vee$
and two copies of $E^\vee(\Gamma)$,
thus the relation \eqref{claim-proof1} trivially holds.
Therefore we've proved
\be
\pd^\vee=\cD.
\ee

The other equality $\cD^\vee=\pd$ follows from
\ben
\cD^\vee=\pd^\vee+\gamma^\vee=\pd^\vee-\gamma=\cD-\gamma=\pd.
\een
\end{proof}

The relation between the operators $K$ and $K^\vee$
is more complicated.
Let $\Gamma^\vee$ be a dotted stable graph,
and let $K$ act on $\Gamma^\vee$.
As a linear combination of ordinary stable graphs,
$\Gamma^\vee$ is obtained by gluing some ordinary stable graphs togrther
along the edges $e\in E^\vee(\Gamma^\vee)$,
where these ordinary graphs are of the same types as
dotted vertices of $\Gamma^\vee$.
Therefore the operator $K$ can act on both the edges $e^\vee \in E^\vee(\Gamma^\vee)$
and the dotted vertices of $\Gamma^\vee$.
We regard $\Gamma^\vee$ as a product of
internal edges $e^\vee \in E^\vee(\Gamma^\vee)$
and dotted vertices $v^\vee \in V^\vee(\Gamma^\vee)$,
then the action of $K$ on $\Gamma^\vee$ satisfies the Leibniz rule.
The action of $K$ on internal edges $e^\vee \in E^\vee(\Gamma^\vee)$
is the same as the action of $-K^\vee$, where the appearance of the factor $(-1)$
is due to the fact that the number of internal edges decreases by one.
Let $V_{g,n}^\vee$ be a dotted vertex of genus $g$ and valence $n$,
then Theorem \ref{thm-original-rec} tells us
\ben
KV_{g,n}^\vee=\frac{1}{2}\biggl(V_{g-1,n+2}^\vee+
\sum_{\substack{g_1+g_2=g,n_1+n_2=n+2\\n_1\geq 1,n_2\geq 1}}
V_{g_1,n_1}^\vee\sqcup V_{g_2,n_2}^\vee\biggr),
\een
where the sum above is over all stable cases.

\subsection{Dual abstract free energies and the duality theorem}
Let $\cG_g^{\vee,c}$ be the set of all connected dotted stable graphs of genus $g$
without external edges, and $\cG_{g,n}^{\vee,c}$ be the set of
all connected dotted stable graphs of genus $g$ with $n$ external edges.
Then similar to the case of ordinary stable graphs,
we can define a dual abstract free energies for dotted stable graphs.

\begin{Definition}
\label{def-dual-absfe}
For $g\geq 2$, we define the dual abstract free energy $\wcF_g^\vee$ to be
\be
\wcF_g^\vee:=\sum_{\Gamma^\vee\in\cG_g^{\vee,c}}
\frac{1}{|\Aut(\Gamma^\vee)|}\Gamma^\vee.
\ee
Also, we define the dual abstract $n$-point functions to be
\be
\wcF_{g,n}^\vee:=\sum_{\Gamma^\vee\in\cG_{g,n}^{\vee,c}}
\frac{1}{|\Aut(\Gamma^\vee)|}\Gamma^\vee
\ee
for $2g-2+n>0$.

\end{Definition}

Clearly we have
\ben
\phi(\wcF_{g,n})=\wcF_{g,n}^\vee
\een
for any $2g-2+n>0$, thus if we apply $\phi$ to
Lemma \ref{lem-original-D} and Theorem \ref{thm-original-rec},
by the properties \eqref{phi-K} and \eqref{phi-pd&cD}
we obtain the following recursion relations for the
dual abstract free energies and dual abstract $n$-point functions:

\begin{Lemma}\label{lem-new-D}
For $2g-2+n>0$, we have
\be
\cD^\vee\wcF_{g,n}^\vee=(n+1)\wcF_{g,n+1}^\vee.
\ee
\end{Lemma}

\begin{Theorem}\label{thm-new-rec}
For $2g-2+n>0$, we have
\be\label{eq-thm-new-rec}
K^\vee\widehat{\cF}^\vee_{g,n}=\frac{1}{2}(\cD^\vee \cD^\vee\widehat{\cF}^\vee_{g-1,n}
+\sum_{\substack{g_1+g_2=g,\\n_1+n_2=n}}
\cD^\vee\widehat{\cF}^\vee_{g_1,n_1}\cD^\vee\widehat{\cF}^\vee_{g_2,n_2}).
\ee

In particular, by taking $n=0$ we get a recursion relation
for the dual abstract free energy for $g\geq 2$:
\be\label{eq-thm-new-rec}
K^\vee\widehat{\cF}^\vee_g=\frac{1}{2}
(\cD^\vee \pd^\vee\widehat{\cF}^\vee_{g-1}+\sum_{r=1}^{g-1}
\pd^\vee\widehat{\cF}^\vee_{r}\pd^\vee\widehat{\cF}^\vee_{g-r}).
\ee
\end{Theorem}

\begin{Example}
The dual free energy of genus $2$ is given by
\begin{flalign*}
\begin{tikzpicture}
\node [align=center,align=center] at (0.1+0.7,0) {$\wcF_2^\vee=$};
\draw [dotted,thick](1+0.3+0.2,0) circle [radius=0.2];
\node [align=center,align=center] at (1+0.3+0.2,0) {$2$};
\node [align=center,align=center] at (1.6+0.2+0.1+0.2,0) {$+\frac{1}{2}$};
\draw [dotted,thick](1+1.4+0.2+0.1+0.2,0) circle [radius=0.2];
\draw [dotted,thick](0.84+1.4+0.2+0.1+0.2,0.1) .. controls (0.5+1.4+0.2+0.1+0.2,0.2) and (0.5+1.4+0.2+0.1+0.2,-0.2) ..  (0.84+1.4+0.2+0.1+0.2,-0.1);
\node [align=center,align=center] at (1+1.4+0.2+0.1+0.2,0) {$1$};
\node [align=center,align=center] at (3+0.2+0.2+0.1,0) {$+\frac{1}{2}$};
\draw [dotted,thick](1+2.6+0.2+0.2+0.1,0) circle [radius=0.2];
\draw [dotted,thick](1.2+2.6+0.2+0.2+0.1,0)--(1.4+2.6+0.2+0.2+0.1,0);
\draw [dotted,thick](1.6+2.6+0.2+0.2+0.1,0) circle [radius=0.2];
\node [align=center,align=center] at (1+2.6+0.2+0.2+0.1,0) {$1$};
\node [align=center,align=center] at (1.6+2.6+0.2+0.2+0.1,0) {$1$};
\node [align=center,align=center] at (5+0.3,0) {$+\frac{1}{8}$};
\draw [dotted,thick](1+4.8+0.3,0) circle [radius=0.2];
\draw [dotted,thick](0.84+4.8+0.3,0.1) .. controls (0.5+4.8+0.3,0.2) and (0.5+4.8+0.3,-0.2) ..  (0.84+4.8+0.3,-0.1);
\draw [dotted,thick](1.16+4.8+0.3,0.1) .. controls (1.5+4.8+0.3,0.2) and (1.5+4.8+0.3,-0.2) ..  (1.16+4.8+0.3,-0.1);
\node [align=center,align=center] at (1+4.8+0.3,0) {$0$};
\node [align=center,align=center] at (6.6+0.4,0) {$+\frac{1}{2}$};
\draw [dotted,thick](1+6.8+0.4,0) circle [radius=0.2];
\draw [dotted,thick](0.4+6.8+0.4,0) circle [radius=0.2];
\draw [dotted,thick](0.6+6.8+0.4,0)--(0.8+6.8+0.4,0);
\draw [dotted,thick](1.16+6.8+0.4,0.1) .. controls (1.5+6.8+0.4,0.2) and (1.5+6.8+0.4,-0.2) ..  (1.16+6.8+0.4,-0.1);
\node [align=center,align=center] at (1+6.8+0.4,0) {$0$};
\node [align=center,align=center] at (0.4+6.8+0.4,0) {$1$};
\node [align=center,align=center] at (8.6+0.4,0) {$+\frac{1}{8}$};
\draw [dotted,thick](1+9+0.4,0) circle [radius=0.2];
\draw [dotted,thick](0.4+9+0.4,0) circle [radius=0.2];
\draw [dotted,thick](0.6+9+0.4,0)--(0.8+9+0.4,0);
\draw [dotted,thick](1.16+9+0.4,0.1) .. controls (1.5+9+0.4,0.2) and (1.5+9+0.4,-0.2) ..  (1.16+9+0.4,-0.1);
\draw [dotted,thick](0.24+9+0.4,0.1) .. controls (-0.1+9+0.4,0.2) and (-0.1+9+0.4,-0.2) ..  (0.24+9+0.4,-0.1);
\node [align=center,align=center] at (1+9+0.4,0) {$0$};
\node [align=center,align=center] at (0.4+9+0.4,0) {$0$};
\node [align=center,align=center] at (11.3,0) {$+\frac{1}{12}$};
\draw [dotted,thick](1+10.2+0.2+0.5,0) circle [radius=0.2];
\draw [dotted,thick](1.2+10.2+0.2+0.5,0)--(1.4+10.2+0.2+0.5,0);
\draw [dotted,thick](1.16+10.2+0.2+0.5,0.1)--(1.44+10.2+0.2+0.5,0.1);
\draw [dotted,thick](1.16+10.2+0.2+0.5,-0.1)--(1.44+10.2+0.2+0.5,-0.1);
\draw [dotted,thick](1.6+10.2+0.2+0.5,0) circle [radius=0.2];
\node [align=center,align=center] at (1+10.2+0.2+0.5,0) {$0$};
\node [align=center,align=center] at (1.6+10.2+0.2+0.5,0) {$0$};
\end{tikzpicture}.&&
\end{flalign*}

Using the expressions in Example \ref{eg-genus2-graph},
we are able to express $\wcF_2^\vee$ entirely in terms of ordinary stable graphs.
The final result is simple:
\be
\begin{tikzpicture}
\node [align=center,align=center] at (0.2,0) {$\wcF_2^\vee=$};
\draw (1,0) circle [radius=0.2];
\node [align=center,align=center] at (1,0) {$2$};
\end{tikzpicture}.
\ee

\end{Example}

This is indeed another example of the duality we mentioned
at the end of subsection \ref{subsec3.2}.
In general, we have the following duality theorem.

\begin{Theorem}
\label{thm-dual-FE}
For $2g-2+n>0$, $\wcF_{g,n}^\vee$ equals $\frac{1}{n!}$ times
the ordinary stable vertex of genus $g$ with $n$ external edges:
\be\label{thm2-eq1}
\begin{tikzpicture}
\node [align=center,align=center] at (-1.2,0) {$\wcF_{g,n}^\vee=\frac{1}{n!}$};
\draw (0,0) circle [radius=0.2];
\node [align=center,align=center] at (0,0) {$g$};
\draw (0.16,0.1)--(0.5,0.35);
\draw (0.16,-0.1)--(0.5,-0.35);
\node [align=center,align=center] at (0.5,0.1) {$\vdots$};
\node [align=center,align=center] at (0.7,-0.15) {$n$};
\end{tikzpicture}.
\ee
In particular, for $g\geq 2$, we have
\be\label{thm2-eq2}
\begin{tikzpicture}
\node [align=center,align=center] at (0.2,0) {$\wcF_g^\vee=$};
\draw (1,0) circle [radius=0.2];
\node [align=center,align=center] at (1,0) {$g$};
\end{tikzpicture}.
\ee
\end{Theorem}

\begin{proof}
The cases of $(g,n)=(0,3)$ and $(1,1)$ have already been checked
in Example \ref{eg-inverse}.
For all other cases,
we only need to prove \eqref{thm2-eq2} since by Theorem \ref{thm1}
and Lemma \ref{lem-new-D} we have for $g \geq 2$,
\be
\wcF_{g,n}^\vee=\frac{1}{n!}(\cD^\vee)^n \wcF_g^\vee=\frac{1}{n!}\pd^n \wcF_g^\vee,
\ee
which gives \eqref{thm2-eq1} if we have \eqref{thm2-eq2}.
The general $g=0$ and $g=1$ cases can be proved in the same fashion.

For a given $g\geq 2$,
by definition,
\ben
\wcF_g^\vee=\sum_{\Gamma^\vee\in\cG_g^{\vee,c}}
\frac{1}{|\Aut(\Gamma^\vee)|}\Gamma^\vee
\een
can be written as a linear combination of ordinary graphs
of genus $g$ without external edges.
For every ordinary graph $\Gamma$ of genus $g$ without external edges,
the coefficient of $\Gamma$ in $\wcF_g^\vee$ is
\be\label{pf-thm2-coeff}
\sum_{\Gamma^\vee\in\cG_g^{\vee,c}}
\frac{1}{|\Aut(\Gamma^\vee)|}A_{\Gamma^\vee}^\Gamma,
\ee
where $A_{\Gamma^\vee}^\Gamma$ is the coefficient of $\Gamma$
in the expression of $\Gamma^\vee$.

Given an ordinary graph $\Gamma$, we first consider all possible graphs $\Gamma^\vee$
whose expansion in ordinary graphs contains a nonzero multiple of $\Gamma$.
Let us explain how to reconstruct a possible $\Gamma^\vee$
from $\Gamma$ by `choosing internal edges'.
In fact, from the definition of the elements in $\cV$ associated to dotted stable graphs,
we see that the set of internal edges of $\Gamma^\vee$ can be
mapped injectively into the set $E(\Gamma)$ of internal edges of $\Gamma$.
Conversely, for every subset $E$ of $E(\Gamma)$, we can construct
a dotted stable graph $\Gamma_E^\vee$
by inductively contract the edges in $E(\Gamma)\backslash E$ as follows.
Denote the elements in $E(\Gamma)\backslash E$ by $\{e_1, \dots, e_n\}$.
We first contract $e_1$ to get a graph $\Gamma_1$.
We have two cases to consider.
If $e_1$ connects two vertices $v_1$ and $v_2$,
labelled by  $g(v_1)$ and $g(v_2)$ respectively,
then the vertex in $\Gamma_1$ corresponding to these two vertices
is labelled by $g(v_1) +g(v_2)$;
if $e_1$ is a loop incident at a vertex $v$ labelled by $g(v)$,
then we label the corresponding vertex in $\Gamma_1$ by $g(v) +1$.
Repeat this procedure $n$ times,
we end up with a graph $\Gamma_n$.
It is not hard to see that $\Gamma_n$ is independent of the choice
of the ordering of the set $E(\Gamma)-E$.
The dotted graph $\Gamma^\vee_E$ is just $\phi(\Gamma)$,
i.e., it is obtained from $\Gamma_n$ by
changing the edges in $\Gamma_n$ to dotted edges,
and vertices to dotted vertices.
For example, for the stable graph
\ben
\begin{tikzpicture}
\draw (0,0) circle [radius=0.2];
\node [align=center,align=center] at (0,0) {$2$};
\draw (0.8,0) circle [radius=0.2];
\node [align=center,align=center] at (0.8,0) {$1$};
\draw (0.2,0)--(0.6,0);
\draw (0.96,0.1) .. controls (1.3,0.2) and (1.3,-0.2) ..  (0.96,-0.1);
\node [above,align=center] at (0.4,0) {$\alpha$};
\node [align=center,right] at (1.3,0) {$\beta$};
\end{tikzpicture}
\een
of genus $4$ with two internal edges $\alpha,\beta$,
for $E=\emptyset$, $E=\{\alpha\}$, $E=\{\beta\}$, $E=\{\alpha,\beta\}$,
the corresponding dotted stable graphs $\Gamma^\vee_E$ are
\ben
\begin{tikzpicture}
\draw [dotted,thick](0,0) circle [radius=0.2];
\node [align=center,align=center] at (0,0) {$4$};
\draw [dotted,thick](2,0) circle [radius=0.2];
\node [align=center,align=center] at (2,0) {$2$};
\draw [dotted,thick](2.6,0) circle [radius=0.2];
\node [align=center,align=center] at (2.6,0) {$2$};
\draw [dotted,thick](2.2,0)--(2.4,0);
\draw [dotted,thick](4.6,0) circle [radius=0.2];
\node [align=center,align=center] at (4.6,0) {$3$};
\draw [dotted,thick](4.76,0.1) .. controls (5.1,0.2) and (5.1,-0.2) ..  (4.76,-0.1);
\draw [dotted,thick](0+7,0) circle [radius=0.2];
\node [align=center,align=center] at (0+7,0) {$2$};
\draw [dotted,thick](0.8+6.8,0) circle [radius=0.2];
\node [align=center,align=center] at (0.8+6.8,0) {$1$};
\draw [dotted,thick](0.2+7,0)--(0.6+6.8,0);
\draw [dotted,thick](0.96+6.8,0.1) .. controls (1.3+6.8,0.2) and (1.3+6.8,-0.2) ..  (0.96+6.8,-0.1);
\end{tikzpicture}
\een
respectively.

Now we fix a stable graph $\Gamma\in\cG_{g,0}^c$
together with a dotted stable graph $\Gamma^\vee \in\cG_{g,0}^{\vee,c}$.
Define $S_{\Gamma,\Gamma^\vee}$ to be the following set of
stable graphs with different names on some of the half-edges:
\ben
S_{\Gamma,\Gamma^\vee}:=\{\widetilde\Gamma
&|&
\widetilde\Gamma\text{ is a stable graph with names, obtained from $\Gamma$ by adding}\\
&&\text{$2|E|$ different names on half-edges of $E$,
where $E\subset E(\Gamma)$}\\
&&\text{such that $\Gamma^\vee$ can be obtained from $\Gamma$ and $E$}\}.
\een
Then we claim:
\be\label{thm2-claim1}
\frac{A_{\Gamma^\vee}^\Gamma}{|\Aut(\Gamma^\vee)|}=
\frac{(-1)^{|E^\vee(\Gamma^\vee)|}}{(2|E^\vee(\Gamma^\vee)|)!}
\cdot\sum_{\widetilde\Gamma\in S_{\Gamma,\Gamma^\vee}}
\frac{1}{|\Aut(\widetilde\Gamma)|}.
\ee

Now let us prove the above claim.
First we define
\ben
S_{\Gamma,\Gamma^\vee}^\vee:=\{
\widetilde\Gamma^\vee&|&
\text{$\widetilde\Gamma^\vee$ is obtained from $\Gamma^\vee$
by adding $2|E^\vee(\Gamma^\vee)|$ different}\\
&&\text{names to every half edge of $\Gamma^\vee$}\}.
\een
Then similar to the case of Lemma \ref{lemma-mark-1},
using the argument of `forgetting names',
we can prove:
\be\label{thm2-claim1-pf}
\frac{(2|E^\vee(\Gamma^\vee)|)!}{|\Aut(\Gamma^\vee)|}=
\sum_{\widetilde\Gamma^\vee\in S_{\Gamma,\Gamma^\vee}^\vee}
\frac{1}{|\Aut(\widetilde\Gamma^\vee)|}
=\sum_{\widetilde\Gamma^\vee\in S_{\Gamma,\Gamma^\vee}^\vee}1
=|S_{\Gamma,\Gamma^\vee}^\vee|.
\ee

Recall that the coefficients $A_{\Gamma^\vee}^\Gamma$ are defined in \S \ref{subsec3.2}
by first assigning names to the half-edges of $\Gamma^\vee$,
and representing the dotted stable vertices of $\Gamma^\vee$
as linear combinations of ordinary graphs
$\widetilde\Gamma_1,\cdots,\widetilde\Gamma_k$ with names on external edges
(with coefficients $\frac{1}{|\Aut(\widetilde\Gamma_1)|},
\cdots,\frac{1}{|\Aut(\widetilde\Gamma_k)|}$ respectively),
then gluing these graphs together according to the names
and multiplying by $(-1)^{|E^\vee(\Gamma^\vee)|}$.
Therefore
\ben
|S_{\Gamma,\Gamma^\vee}^\vee |\cdot A_{\Gamma^\vee}^\Gamma
=(-1)^{|E^\vee(\Gamma^\vee)|}\cdot\sum
\frac{1}{|\Aut(\widetilde\Gamma_1)|\cdots|\Aut(\widetilde\Gamma_k)|},
\een
where the sum on the right-hand-side is over all possible ways
to choose $\widetilde\Gamma_1,\cdots,\widetilde\Gamma_k$.
Notice that we only glue the half-edges with names together,
thus
\ben
\frac{1}{|\Aut(\widetilde\Gamma_1)|\cdots|\Aut(\widetilde\Gamma_k)|}=
\frac{1}{|\Aut(\widetilde\Gamma)|}
\een
where $\widetilde\Gamma$ is the resulted graph with names
on half-edges which correspond to external edges of
$\widetilde\Gamma_1,\cdots,\widetilde\Gamma_k$.
Thus we have
\ben
|S_{\Gamma,\Gamma^\vee}^\vee |\cdot A_{\Gamma^\vee}^\Gamma=
(-1)^{|E^\vee(\Gamma^\vee)|}\cdot\sum_{\widetilde\Gamma\in S_{\Gamma,\Gamma^\vee}}
\frac{1}{|\Aut(\widetilde\Gamma)|}.
\een
By this equation and equation \eqref{thm2-claim1-pf},
we have proved the claim \eqref{thm2-claim1}.

Now recall that the coefficient of a stable graph $\Gamma$ in $\wcF_g^\vee$
is given by \eqref{pf-thm2-coeff},
thus by the claim \eqref{thm2-claim1} we can rewrite this coefficient as
\be\label{thm2-pf-coeff2}
\sum_{\Gamma^\vee\in\cG_g^{\vee,c}}
\frac{1}{|\Aut(\Gamma^\vee)|}A_{\Gamma^\vee}^\Gamma
=
\sum_{\Gamma^\vee\in\cG_g^{\vee,c}}
\frac{(-1)^{|E^\vee(\Gamma^\vee)|}}{(2|E^\vee(\Gamma^\vee)|)!}
\sum_{\widetilde\Gamma\in S_{\Gamma,\Gamma^\vee}}
\frac{1}{|\Aut(\widetilde\Gamma)|}.
\ee

Let $S_{\Gamma,\Gamma^\vee}'$ be the following set of stable graphs
with a same name on some of the half-edges:
\ben
S_{\Gamma,\Gamma^\vee}':=\{\Gamma'
&|&
\text{$\Gamma'$ is obtained from $\Gamma$ by adding the same name to}\\
&&\text{$2|E|$ half-edges of $E$,
where $E\subset E(\Gamma)$ such that }\\
&&\text{$\Gamma^\vee$ can be obtained from $\Gamma$ and $E$}\}.
\een
Then we can apply the argument of `forgetting names' again,
we get
\ben
\frac{1}{(2|E^\vee(\Gamma^\vee)|)!}
\sum_{\widetilde\Gamma\in T_{\Gamma'}}
\frac{1}{|\Aut(\widetilde\Gamma)|}=
\frac{1}{|\Aut(\Gamma')|}
\een
for every $\Gamma'\in S_{\Gamma,\Gamma^\vee}'$,
where the set $T_{\Gamma'}$ is defined to be the subset of $S_{\Gamma,\Gamma^\vee}$,
consisting of those $\Gamma^\vee$ such that $\Gamma^\vee$ can be obtained from $\Gamma'$
by changing the names on $\Gamma'$ to $2|E^\vee(\Gamma^\vee)|$ different names.
Therefore,
we can rewrite the coefficient \eqref{thm2-pf-coeff2} as
\ben
\sum_{\Gamma^\vee\in\cG_g^{\vee,c}}
\frac{(-1)^{|E^\vee(\Gamma^\vee)|}}{(2|E^\vee(\Gamma^\vee)|)!}
\sum_{\widetilde\Gamma\in S_{\Gamma,\Gamma^\vee}}
\frac{1}{|\Aut(\widetilde\Gamma)|}=
\sum_{\Gamma^\vee\in\cG_g^{\vee,c}}(-1)^{|E^\vee(\Gamma^\vee)|}
\sum_{\Gamma'\in S_{\Gamma,\Gamma^\vee}'}\frac{1}{|\Aut(\Gamma')|}.
\een
Notice that we can also regard $\Gamma'$ as a stable graph
with a same name on some of its internal edges (instead of half-edges),
thus applying the argument of `forgetting names' one more time,
we get
\be\label{eq-coeff-final}
\sum_{\Gamma^\vee\in\cG_g^{\vee,c}}(-1)^{|E^\vee(\Gamma^\vee)|}
\sum_{\Gamma'\in S_{\Gamma,\Gamma^\vee}'}\frac{1}{|\Aut(\Gamma')|}
=
\frac{1}{|\Aut(\Gamma)|}\cdot\sum_{E\subset E(\Gamma)}(-1)^{|E|} .
\ee
Clearly $\sum\limits_{E\subset E(\Gamma)}(-1)^{|E|}=0$ if $E(\Gamma)$ is not empty.
And if $E(\Gamma)=\emptyset$,
i.e., $\Gamma$ is a single ordinary stable vertex,
then the coefficient of $\Gamma$ is $\frac{1}{|\Aut(\Gamma)|}$.
This proves the conclusion \eqref{thm2-eq2}.
\end{proof}

Denote by $V_{g,n}$ an ordinary stable vertex
of genus $g$ and valence $n$,
and $V_{g,n}^\vee$ a dotted stable vertex of genus $g$ and valence $n$.
The above theorem simply tells us
$\phi(V_{g,n}^\vee)=n!\cdot\wcF_{g,n}^\vee=V_{g,n}$,
i.e., $\phi^2(V_{g,n})=V_{g,n}$ for every $2g-2+n>0$.
We will see that this relation is true for a stable graph in general.

\begin{Theorem} \label{thm:Involution}
For any ordinary stable graph $\Gamma$,
we have $\phi^2(\Gamma)=\Gamma$.
I.e., the linear map $\phi: \cV\to \cV$ is an involution.
\end{Theorem}

\begin{proof}
Let $\Gamma^\vee:=\phi (\Gamma)$,
then $\Gamma^\vee$ is a dotted stable graph whose vertices
are all of the same type as the ordinary vertices of $\Gamma$.
By giving different names to half-edges of $\Gamma^\vee$
and cutting off all its internal edges,
we get a set of dotted stable vertices $\{V_{g_i,n_i}^\vee\}$.
We already know that $\phi(V_{g_i,n_i}^\vee)=V_{g_i,n_i}$.
Notice that $\phi(\Gamma^\vee)$ can be obtained by
gluing external edges with names of
$\phi(V_{g_i,n_i}^\vee)=V_{g_i,n_i}$ together
if their ancestors are joined together,
and finally forgetting the names and multiplying by $(-1)^{|E^\vee(\Gamma^\vee)|}$,
therefore we indeed get $\phi(\Gamma^\vee)=\Gamma$,
i.e., $\phi^2(\Gamma)=\Gamma$.
\end{proof}

\section{Realizations of the Abstract QFT vs. Dual Realizations of the Dual Abstract QFT}
\label{sec:realization}

Given a realization $\{F_{g,n},\kappa\}$ of the abstract quantum field theory,
we are able to construct a natural realization of the dual abstract quantum field theory.
In this section we describe such a procedure.
We will focus on the one-dimensional case (i.e., stable graphs without labels)
in this section.

\subsection{Realizations of the dual abstract QFT}

Let us first recall the construction of a realization of abstract QFT \cite{wz}
(see also \S \ref{sec-pre-realization}).
Let $\{F_{g,n}(t, \kappa)\}_{2g-2+n>0}$ be a sequence of functions,
and let $\kappa$ be a formal variable
(it can be taken to be arbitrary smooth function in $t$ in examples).
Then the Feynman rule \eqref{eqn:FR} gives us
a realization of the abstract quantum field theory
(see \S \ref{sec-pre-realization}).
The abstract free energy $\wcF_g$ ($g\geq2$)
and abstract $n$-point functions $\wcF_{g,n}$ ($2g-2+n>0$)
can be realized by functions $\wF_{g}(t,\kappa)$ and $\wF_{g,n}(t,\kappa)$
respectively.

Now this construction can be applied to the dual abstract QFT,
the main difference is that the realizations of the dual abstract QFT use the dotted
stable graphs,
and to get the Feynman rules,
we need a collection $\{F^\vee_{g,n}(t, \kappa^\vee)\}_{2g-2+n>0}$ that we use for
 the vertices,
and a formal variable $\kappa^\vee$ which we use as the propagator.
For the present,
we do not require that $\{F^\vee_{g,n}(t, \kappa^\vee)\}_{2g-2+n>0}$
or $\kappa^\vee$ to be related to $\{F_{g,n}(t, \kappa)\}_{2g-2+n>0}$
or $\kappa$.
Then the dual version of the Feynman rules \eqref{eqn:FR} gives us
a realization of the dual abstract quantum field theory.
Denote by $\wF^\vee_{g}(t,\kappa^\vee)$ and $\wF^\vee_{g,n}(t,\kappa^\vee)$
the resulting free energy and $n$-point functions respectively.
They realize
the dual abstract free energy $\wcF^\vee_g$ ($g\geq2$)
and dual abstract $n$-point functions $\wcF^\vee_{g,n}$ ($2g-2+n>0$)
respectively.

\subsection{Induced realizations of the dual abstract QFT}

Let us now consider a special case of the realization of
the dual abstract QFT.
To motivate it,
let us first examine an example.

\begin{Example}\label{eg-realization}
We give some examples of $\wF_{g,n}(t,\kappa)$ for small $g$, $n$:
\be\label{eq1-eg-realization}
\begin{split}
&\wF_{0,3}=\frac{1}{6}F_{0,3},\\
&\wF_{0,4}=\frac{1}{24}F_{0,4}+
\kappa\frac{1}{8}F_{0,3}^2,\\
&\wF_{1,1}=F_{1,1}+\frac{1}{2}\kappa F_{0,3},\\
&\wF_{1,2}=\frac{1}{2}F_{1,2}+
\frac{1}{4}\kappa F_{0,4}+\frac{1}{2}\kappa F_{1,1}F_{0,3}
+\frac{1}{2}\kappa^2 F_{0,3}^2.
\end{split}
\ee
Using the above expressions, we can also solve $\{F_{g,n}\}$
from $\{\wF_{g,n}\}$ and $\kappa$:
\be\label{eq2-eg-realization}
\begin{split}
&F_{0,3}=6\wF_{0,3},\\
&F_{0,4}=24\wF_{0,4}-
108 \kappa\wF_{0,3}^2,\\
&F_{1,1}=\wF_{1,1}-3\kappa \wF_{0,3},\\
&F_{1,2}=2\wF_{1,2}-12
\kappa \wF_{0,4}-6\kappa \wF_{1,1}\wF_{0,3}
+36\kappa^2 \wF_{0,3}^2.
\end{split}
\ee
\end{Example}

By comparing the expressions in \eqref{eq1-eg-realization} and \eqref{eq2-eg-realization},
we are led to the following construction.
Suppose that we are given $\{F_{g,n}(t, \kappa)\}_{2g-2+n>0}$
and $\kappa$,
and we have used the realization of the abstract QFT to construct
$\{\wF_{g,n}(t, \kappa)\}_{2g-2+n>0}$,
then we can take
\be
F^\vee_{g,n}(t, \kappa)= n!\wF_{g,n}(t, \kappa),
\ee
and for the propagator $\kappa^\vee$,  we take
\be
\kappa^\vee = -\kappa,
\ee
and use them to carry out the realization
of the dual abstract quantum field theory as discussed in last subsection.
We refer to it as
the {\em induced realization of the dual abstract QFT}.
For a dotted stable graph $\Gamma^\vee$, the contribution
of a dotted vertex $v^\vee\in V^\vee(\Gamma^\vee)$
of genus $g_v^\vee$ and valence $n_v^\vee$ is given by
\be
\omega_{v^\vee}=n!\cdot\wF_{g_v^\vee,n_v^\vee}
=\sum_{\Gamma\in\cG_{g_v,n_v}^c}\frac{1}{|\Aut(\Gamma)|}\omega_\Gamma,
\ee
and the contribution of an internal edge $e^\vee\in E^\vee(\Gamma^\vee)$
is
\be
\omega_{e^\vee}=-\kappa,
\ee
where the factor $-1$ comes from the factor $(-1)^{|E^\vee(\Gamma)|}$
in the definition of dotted stable graphs.

In what follows we will use the notation
\be\label{eq-notation}
\tF_{g,n}:=n!\cdot\wF_{g,n}.
\ee
Then the induced Feynman rule for a dotted stable graph
$\Gamma^\vee\in\cG_{g,n}^{\vee,c}$ is
\be
\begin{split}
\Gamma^\vee\mapsto \omega_{\Gamma^\vee}=&
\prod_{e^\vee\in E^\vee(\Gamma^\vee)}\omega_{e^\vee}\cdot
\prod_{v^\vee\in V^\vee(\Gamma^\vee)}\omega_{v^\vee}\\
=&(-\kappa)^{|E^\vee(\Gamma^\vee)|}\cdot
\prod_{v^\vee\in V^\vee(\Gamma^\vee)}\tF_{g_v^\vee,n_v^\vee}.
\end{split}
\ee

Using this induced Feynman rule,
the dual abstract free energy and dual abstract $n$-point functions
can be realized as follows.
The dual abstract free energy $\wcF^\vee_g$ of genus $g\geq 2$
is realized by
\be
\wF^\vee_g=\sum_{\Gamma^\vee\in\cG_{g,0}^{\vee,c}}
\frac{1}{|\Aut(\Gamma^\vee)|}
\omega_{\Gamma^\vee},
\ee
and $\wcF^\vee_{g,n}$ for $2g-2+n>0$ is realized by
\be
\wF^\vee_{g,n}=\sum_{\Gamma^\vee\in\cG_{g,n}^{\vee,c}}
\frac{1}{|\Aut(\Gamma^\vee)|}
\omega_{\Gamma^\vee}.
\ee

\subsection{Realization of the duality theorem}

Given a realization of the abstract QFT and the induced realization
of the dual abstract QFT,
Theorem \ref{thm-dual-FE} gives us the following:

\begin{Theorem}\label{thm-dual-realization}
For every $2g-2+n>0$, we have
\be\label{eq-duality-onedim}
F_{g,n}=n!\cdot\wF^\vee_{g,n}=n!\cdot
\sum_{\Gamma^\vee\in\cG_{g,n}^{\vee,c}}
\frac{1}{|\Aut(\Gamma^\vee)|}
\omega_{\Gamma^\vee}.
\ee
In particular, for every $g\geq2$ we have
\be
F_{g,0}=\wF^\vee_g
=
\sum_{\Gamma^\vee\in\cG_{g,n}^{\vee,c}}
\frac{1}{|\Aut(\Gamma^\vee)|}
\omega_{\Gamma^\vee}.
\ee
\end{Theorem}

Now Theorem \ref{thm-dual-realization} provides an explicit way to
express $F_{g,n}$ in terms of $\tF_{g,n}$ and the propagator $\kappa$.

\begin{Example}\label{eq-inverse}
We have
\be
\begin{split}
&\frac{1}{6}F_{0,3}=\frac{1}{6}\tF_{0,3},\\
&\frac{1}{24}F_{0,4}=\frac{1}{24}\tF_{0,4}-
\frac{1}{8}\kappa\tF_{0,3}^2,\\
&F_{1,1}=\tF_{1,1}-\frac{1}{2}\kappa \tF_{0,3},\\
&\frac{1}{2}F_{1,2}=\frac{1}{2}\tF_{1,2}-
\frac{1}{4}\kappa \tF_{0,4}-\frac{1}{2}\kappa \tF_{1,1}\tF_{0,3}
+\frac{1}{2}\kappa^2 \tF_{0,3}^2,
\end{split}
\ee
these match with the equalities \eqref{eq2-eg-realization}.
Moreover, the free energies $F_2(t)$ and $F_3(t)$ are given by:
\be
\begin{split}
F_2=&\tF_{2,0}-\kappa\big(\frac{1}{2}\tF_{1,2}+\frac{1}{2}\tF_{1,1}^2\big)
+\kappa^2\big(\frac{1}{8}\tF_{0,4}+\frac{1}{2}\tF_{1,1}\tF_{0,3}\big)
-\frac{5}{24}\kappa^3\tF_{0,3}^2;\\
F_3=&\tF_{3,0}-\kappa\big(\frac{1}{2}\tF_{2,2}+\tF_{1,1}\tF_{2,1}\big)\\
&+\kappa^2\big(\frac{1}{8}\tF_{1,4}+\frac{1}{4}\tF_{1,2}^2+
\frac{1}{2}\tF_{0,3}\tF_{2,1}+\frac{1}{2}\tF_{1,1}\tF_{1,3}+
\frac{1}{2}\tF_{1,1}^2\tF_{1,2}\big)\\
&-\kappa^3\big(\frac{1}{48}\tF_{0,6}+\frac{1}{4}\tF_{0,4}\tF_{1,2}+
\frac{5}{12}\tF_{0,3}\tF_{1,3}+\frac{1}{8}\tF_{0,5}\tF_{1,1}
+\tF_{0,3}\tF_{1,1}\tF_{1,2}\\
&\qquad +\frac{1}{4}\tF_{0,4}\tF_{1,1}^2+\frac{1}{6}\tF_{0,3}\tF_{1,1}^3\big)\\
&+\kappa^4\big(\frac{1}{12}\tF_{0,4}^2+\frac{7}{48}\tF_{0,3}\tF_{0,5}+
\frac{5}{8}\tF_{0,3}^2 \tF_{1,2}+\frac{2}{3}\tF_{0,3}\tF_{0,4}\tF_{1,1}
+\frac{1}{2}\tF_{0,3}^2 \tF_{1,1}^2\big)\\
&-\kappa^5\big(\frac{25}{48}\tF_{0,3}^2\tF_{0,4}+\frac{5}{8}\tF_{0,3}^3\tF_{1,1}\big)
+\frac{5}{16}\kappa^6 \tF_{0,3}^4.
\end{split}
\ee

\end{Example}

\subsection{Representation by formal Gaussian integrals}
\label{subsec-transf}

In this subsection let us represent the above constructions
as formal Gaussian integrals.

Recall the formal Gaussian integral in Theorem \ref{thm-Gaussian},
in the case of $N=1$ this theorem gives us:
\ben
\exp\biggl(\sum_{g\geq 2}\lambda^{2g-2}\wF_{g}\biggr)
=\frac{1}{(2\pi\lambda^2\kappa)^{\frac{1}{2}}}
\int \exp\biggl\{\sum_{2g-2+n>0}\lambda^{2g-2}F_{g,n}\cdot
\frac{\eta^n}{n!}-\frac{\lambda^{-2}}{2\kappa}\eta^2\biggr\}d\eta.
\een
By Theorem \ref{thm-dual-realization},
we also have:
\ben
\exp\biggl(\sum_{g\geq 2}\lambda^{2g-2}F_{g,0}\biggr)
=\frac{1}{(-2\pi\lambda^2\kappa)^{\frac{1}{2}}}
\int \exp\biggl\{\sum_{2g-2+n>0}\lambda^{2g-2}\tF_{g,n}\cdot
\frac{\eta^n}{n!}+\frac{\lambda^{-2}}{2\kappa}\eta^2\biggr\}d\eta,
\een
where $\tF_{g,n}=n!\cdot \wF_{g,n}$.

Now let us generalize the formal Gaussian integral
described in Theorem \ref{thm-Gaussian}
in the following way:

\begin{Theorem}\label{thm-transf-1}
Let $z$ be a formal variable.
Then we have:
\begin{equation}\label{eq-transf-1}
\begin{split}
&\exp\bigg(\sum_{2g-2+n>0}
\lambda^{2g-2}z^n\cdot \wF_{g,n}\bigg)\\
=&
\frac{1}{(2\pi\lambda^{2}\kappa)^{\frac{1}{2}}}
\int \exp\biggl\{
\sum_{2g-2+n>0}\lambda^{2g-2}F_{g,n}\cdot
\frac{\eta^n}{n!}-\frac{\lambda^{-2}}{2\kappa}(\eta-z)^2
\biggr\}d\eta.
\end{split}
\end{equation}
\end{Theorem}

\begin{proof}
First let us consider the formal formal Gaussian integral:
\ben
\widehat Z^*:=\frac{1}{(2\pi\lambda^{2}\kappa)^{\frac{1}{2}}}
\int \exp\biggl\{\lambda^{-2}(\kappa^{-1}z)\eta+
\sum_{2g-2+n>0}\lambda^{2g-2}F_{g,n}\cdot
\frac{\eta^n}{n!}-\frac{\lambda^{-2}}{2\kappa}\eta^2\biggr\}d\eta.
\een
It is clear that
\be\label{eq-Z-star}
\log(\widehat Z^*)=\sum\limits_{g\geq 0}\lambda^{2g-2}\widehat F^*_g,
\ee
where for every $g\geq 0$,
\ben
\widehat F_g^*=\sum_{\Gamma\in\cG_g^*}\frac{\omega_\Gamma}{|\Aut(\Gamma)|}
:=\sum_{\Gamma\in\cG_g^*}\frac{1}{|\Aut(\Gamma)|}
\prod_{v\in V(\Gamma)}F_{g_v,n_v}
\cdot\prod_{e\in E(\Gamma)}\kappa,
\een
where $\cG_g^*$ is the set of connected graphs of genus $g$ without external edges,
whose vertices are either stable,
or of genus $0$ and valence $1$;
and the weight of a vertex of genus $0$ and valence $1$
is defined to be
\ben
F_{0,1}:=\kappa^{-1}\cdot z.
\een

Notice that for every $g\geq 1$,
there is a one-to-one correspondence
between $\cG_g^*$ and the set of all connected stable graphs of genus $g$:
\ben
A_g:\quad
\cG_g^*\quad\to\quad
\bigsqcup_{n\geq 0}\cG_{g,n}^c,
\een
given by deleting the vertices of genus $0$ and valence $1$
(to obtain the same number of external edges).
Let $\Gamma\in \cG_g^*$ ($g\geq 1$) be a graph of genus $g$,
which contains $n$ vertices of genus $0$ and valence $1$.
Denote by $V_0(\Gamma)\subset V(\Gamma)$ the subset of
vertices of genus $0$ and valence $1$,
and by $E_0(\Gamma)$ the subset of internal edges
attached to these vertices,
then clearly we have:
\ben
\omega_\Gamma&=&\prod_{v\in V_0(\Gamma)}\omega_v
\cdot\prod_{v\in V(\Gamma)\backslash V_0(\Gamma)}\omega_v
\cdot\prod_{e\in E_0(\Gamma)}\omega_e
\cdot\prod_{e\in E(\Gamma)\backslash E_0(\Gamma)}\omega_e
\\
&=&
F_{0,1}^n\cdot \kappa^n\cdot \biggl(
\prod_{v\in V(\Gamma)\backslash V_0(\Gamma)}\omega_v
\cdot\prod_{e\in E(\Gamma)\backslash E_0(\Gamma)}\omega_e\biggr)\\
&=&z^n\cdot\omega_{A_g(\Gamma)}.
\een
Therefore,
we have
\ben
\sum_{\Gamma\in \cG_g^*}\frac{\omega_\Gamma}{|\Aut(\Gamma)|}&=&
\sum_{n=0}^\infty\biggl( z^n\cdot\sum_{\Gamma\in \cG_{g,n}^c}
\frac{\omega_\Gamma}{|\Aut(\Gamma)|}
\biggr),
\qquad \forall g\geq 1,
\een
since it is clear that $\Aut(\Gamma)\cong\Aut(A_g(\Gamma))$
for every $\Gamma\in \cG_g^*$.

In the case of genus $0$,
we also have a similar one-to-one correspondence
\ben
A_0:\quad
\cG_g^*\quad\to\quad
\{\Gamma_0\}\sqcup\biggl(\bigsqcup_{n\geq 3}\cG_{0,n}^c\biggr),
\een
where $\Gamma_0$ is the graph consisting of two vertices of genus $0$ and valence $1$,
together with one internal edge connecting these two vertices.
It is clear that in this case we also have
\ben
\omega_\Gamma=z^n\cdot\omega_{A_0(\Gamma)}
\een
for every $\Gamma\in \cG_0^*\backslash\{\Gamma_0\}$.
Thus in the case of genus $0$,
we have
\ben
\sum_{\Gamma\in \cG_0^*}\frac{\omega_\Gamma}{|\Aut(\Gamma)|}&=&
\frac{1}{|\Aut(\Gamma_0)|}\omega_{\Gamma_0}+
\sum_{n=3}^\infty\biggl( z^n\cdot\sum_{\Gamma\in \cG_{0,n}^c}
\frac{\omega_\Gamma}{|\Aut(\Gamma)|}
\biggr)\\
&=&\half \lambda^{-2}\kappa^{-1}z^2+
\sum_{n=3}^\infty\biggl( z^n\cdot\sum_{\Gamma\in \cG_{0,n}^c}
\frac{\omega_\Gamma}{|\Aut(\Gamma)|}
\biggr),
\een
Thus the free energy \eqref{eq-Z-star} equals to:
\ben
\log(\widehat Z^*)&=&
\sum_{\Gamma\in \cG_0^*}\frac{\omega_\Gamma}{|\Aut(\Gamma)|}
+\sum_{g\geq 1}\sum_{\Gamma\in \cG_g^*}\frac{\omega_\Gamma}{|\Aut(\Gamma)|}
\\
&=&\frac{\lambda^{-2}}{2\kappa}z^2+
\sum_{2g-2+n>0}
\lambda^{2g-2}z^n\cdot \wF_{g,n},
\een
therefore we have:
\begin{equation*}
\begin{split}
&\exp\bigg(\sum_{2g-2+n>0}
\lambda^{2g-2}z^n\cdot \wF_{g,n}\bigg)\\
=&
\frac{1}{(2\pi\lambda^{2}\kappa)^{\frac{1}{2}}}
\int \exp\biggl\{\lambda^{-2}(\kappa^{-1}z)\eta+
\sum_{2g-2+n>0}\lambda^{2g-2}F_{g,n}\cdot
\frac{\eta^n}{n!}-\frac{\lambda^{-2}}{2\kappa}\eta^2
-\frac{\lambda^{-2}}{2\kappa}z^2\biggr\}d\eta\\
=&
\frac{1}{(2\pi\lambda^{2}\kappa)^{\frac{1}{2}}}
\int \exp\biggl\{
\sum_{2g-2+n>0}\lambda^{2g-2}F_{g,n}\cdot
\frac{\eta^n}{n!}-\frac{\lambda^{-2}}{2\kappa}(\eta-z)^2
\biggr\}d\eta.
\end{split}
\end{equation*}
\end{proof}

Similarly,
applying Theorem \ref{thm-dual-realization},
we get the inversed version:

\begin{Theorem}\label{thm-transf-2}
Let $\eta$ be a formal variable.
Then we have:
\begin{equation}\label{eq-transf-2}
\begin{split}
&\exp\bigg(\sum_{2g-2+n>0}
\lambda^{2g-2}\eta^n\cdot \wF^\vee_{g,n}\bigg)\\
=&
\frac{1}{(-2\pi\lambda^{2}\kappa)^{\frac{1}{2}}}
\int \exp\biggl\{
\sum_{2g-2+n>0}\lambda^{2g-2}\tF_{g,n}\cdot
\frac{z^n}{n!}+\frac{\lambda^{-2}}{2\kappa}(z-\eta)^2
\biggr\}dz.
\end{split}
\end{equation}
\end{Theorem}

\subsection{Transformations on the space of theories}
\label{sec:Space}

Notice that $F_{g,n}=n!\cdot \wF^\vee_{g,n}$
and $\tF_{g,n}=n!\cdot \wF_{g,n}$,
the two formal Gaussian integrals \eqref{eq-transf-1} and \eqref{eq-transf-1}
can be rewritten as:
\begin{equation}\label{eq-transf-3}
\begin{split}
&\exp\bigg(\sum_{2g-2+n>0}\frac{1}{n!}
\lambda^{2g-2}z^n\cdot \tF_{g,n}\bigg)\\
=&
\frac{1}{(2\pi\lambda^{2}\kappa)^{\frac{1}{2}}}
\int \exp\biggl\{
\sum_{2g-2+n>0}\lambda^{2g-2}F_{g,n}\cdot
\frac{\eta^n}{n!}-\frac{\lambda^{-2}}{2\kappa}(\eta-z)^2
\biggr\}d\eta,
\end{split}
\end{equation}
and
\begin{equation}\label{eq-transf-4}
\begin{split}
&\exp\bigg(\sum_{2g-2+n>0}\frac{1}{n!}
\lambda^{2g-2}\eta^n\cdot F_{g,n}\bigg)\\
=&
\frac{1}{(-2\pi\lambda^{2}\kappa)^{\frac{1}{2}}}
\int \exp\biggl\{
\sum_{2g-2+n>0}\lambda^{2g-2}\tF_{g,n}\cdot
\frac{z^n}{n!}+\frac{\lambda^{-2}}{2\kappa}(z-\eta)^2
\biggr\}dz.
\end{split}
\end{equation}

Now let us rewrite the conclusions of Theorem \ref{thm-transf-1} and Theorem \ref{thm-transf-2}.

\begin{Definition}
Let $\mathcal W$ be the following infinite-dimensional vector space
of formal power series of the form:
\be
w(\eta)=
\sum_{2g-2+n>0}\frac{1}{n!}F_{g,n}\cdot \lambda^{2g-2}\eta^n,
\ee
where $F_{g,n}$ is a smooth function or formal power series,
depending on the concrete problems.
We will also use $\{F_{g,n}\}$ to represent such a series.
We will call $\cW$ the {\em space of field theories}.
\end{Definition}

Define the following two transformations on $\cW$ using
the two formal Gaussian integrals \eqref{eq-transf-3} and \eqref{eq-transf-4}:
\begin{equation*}
\begin{split}
&{\mathcal S}_\kappa:\quad \mathcal W\to \cW,\quad
w(\eta) \mapsto
\log\biggl\{
\frac{1}{(2\pi\lambda^{2}\kappa)^{\frac{1}{2}}}
\int \exp\biggl(
w(z)-\frac{\lambda^{-2}}{2\kappa}(\eta-z)^2
\biggr)dz\biggr\}
\end{split}
\end{equation*}
for $\kappa \neq 0$;
and for $\kappa = 0$,
we understand that
\be
\lim_{\kappa \to 0} \frac{1}{(2\pi\lambda^{2}\kappa)^{\frac{1}{2}}} \exp\biggl(
-\frac{\lambda^{-2}}{2\kappa}(\eta-z)^2
\biggr) = \delta(z-\eta)
\ee
and  so:
\be
w(\eta) \mapsto
\log\biggl\{
\int \exp w(z) \cdot \delta(z-\eta) dz\biggr\} = w(\eta).
\ee

Then Theorem \ref{thm-transf-1} and Theorem \ref{thm-transf-2} can be reformulated as follows:

\begin{Theorem}\label{thm-transf-3}
We have
${\mathcal S}_{-\kappa}\circ {\mathcal S}_{\kappa}=id_{\mathcal W}$,
and
${\mathcal S}_{\kappa}\circ {\mathcal S}_{-\kappa}=id_{\mathcal W}$.
\end{Theorem}

By comparing the definitions of ${\mathcal S}_{\kappa}$ and ${\mathcal S}_{-\kappa}$
we find some similarities to the Fourier transformation and
the inverse Fourier transformation.
Indeed,
if we consider the partition functions
\be
Z(x):= \exp\bigg( \sum_{2g-2+n>0}\frac{1}{n!}F_{g,n}\cdot \lambda^{2g-2}x^n\bigg),
\ee
and think of $\cS_\kappa$ as acting on $Z$, then
\be \label{eqn:Fourier}
(\cS_\kappa Z)(y) = \frac{1}{(2\pi\lambda^{2}\kappa)^{\frac{1}{2}}}
\int Z(x) \cdot \exp\biggl(-\frac{\lambda^{-2}}{2\kappa}(x-y)^2
\biggr)dx.
\ee
From this point of view one gets another proof of Theorem \ref{thm-transf-3}.
Therefore,
we have found a surprising connection between the duality on stable graphs and
the duality of Fourier transforms of the type as in \eqref{eqn:Fourier}.

\subsection{The independence assumption and quadratic recursion relations}
\label{sec:Indep}

In the above we have established the duality between
$\{ \wF^\vee_{g,n}(t, \kappa^\vee) \}_{g,n}$ and $\{\wF_{g,n}(t, \kappa)\}_{g,n}$
with $\kappa^\vee = - \kappa$.
In this subsection we will make  the following:

\begin{Definition}
Given a collection $\{ F^\vee_{g,n}(t, \kappa^\vee=-\kappa) \}_{g,n}$,
if  the dual collection  $\{ F_{g,n}(t, \kappa)\}_{g,n}$ is independent of $\kappa$,
then we say  $\{ F^\vee_{g,n}(t, \kappa^\vee= -\kappa) \}_{g,n}$ satisfies {\em the Independence Assumption}.
\end{Definition}

In this subsection we use the following notations:
When we see $\tD \tF_{g,n}$ we understand it as $\tF_{g,n+1}$,
and we understand $\tD \tD\tF_{g-1,n}$ as $\tF_{g-1,n+2}$.
Then the right-hand-side of \eqref{eq-thm2} can be rewritten as:
\be\label{rec-rhs}
\frac{1}{2}\biggl(\frac{1}{n!}\tD \tD\tF_{g-1,n}+
\sum_{\substack{g_1+g_2=g,\\n_1+n_2=n}}
\frac{1}{n_1!\cdot n_2!}\tD\tF_{g_1,n_1}\tD\tF_{g_2,n_2}\biggr).
\ee
We will  present a  realization of $\tD \tF_{g,n}$ as
an operator  acting on the functions $\{\tF_{g,n}(t,\kappa)\}$  in
 \S \ref{sec:KleZas}.

Now let us consider the realization of the left-hand-side
of the quadratic recursion relation \eqref{eq-thm2}.
Notice that the edge-cutting operator $K$ (see \S \ref{sec-pre-def})
acts only on internal edges, not on vertices.
Therefore,
if we want to realize the edge-cutting operator $K$ as
the partial derivative $\pd_\kappa$,
we need to require that
the partial derivative $\pd_\kappa$ can only act
on the weights of internal edges,
not on the weights of ordinary vertices.
I.e., we need the induced realizations $F_{g,n}(t,\kappa)$
of ordinary vertices to be independent of $\kappa$
for every $2g-2+n>0$.
This exactly means that the collection $\{ \wF^\vee_{g,n}(t, \kappa^\vee) \}_{g,n}$ satisfies
the Independence Assumption.
Then the following theorem holds trivially by definition:

\begin{Theorem}\label{thm-realization-HAE}
Let $\{\tF_{g,n}(t,\kappa),-\kappa\}$ be a realization  of
the dual abstract QFT satisfying the Independence Assumption.
Then for $2g-2+n>0$, we have
\be
\pd_\kappa\wF_{g,n}=\frac{1}{2}\biggl(
\tD \tD\wF_{g-1,n}
+\sum_{\substack{g_1+g_2=g,\\n_1+n_2=n}}
\tD\wF_{g_1,n_1}\tD\wF_{g_2,n_2}\biggr).
\ee
where we denote $\wF_{g,n}:=\frac{1}{n!}\tF_{g,n}$.
In particular, by taking $n=0$ we get
\be
\pd_\kappa\tF_g=\frac{1}{2}\biggl(\tD \tD\tF_{g-1}+
\sum_{r=1}^{g-1}\tD\tF_{r}\tD\tF_{g-r}\biggr)
\ee
for $g\geq 2$.
Here we use the following conventions
\be
\begin{split}
&\tpd\tF_1=\tD\tF_{1}:=\tF_{1,1},\\
&\tD\tF_{0,2}:=\tF_{0,3},\\
&\tD \tD\tF_{0,1}:=\tF_{0,3}.
\end{split}
\ee
\end{Theorem}

\section{A Higher-Dimensional Generalization}
\label{sec:higher-dim}

In this section we will generalize the constructions
in \S \ref{sec:dual-diagrammatics} and \S \ref{sec:realization}
to the case of higher-dimensional state spaces.
The Feynman graphs we need here are the stable graphs
with labels on half-edges
(see \S \ref{sec-pre-label}).

\subsection{Dotted stable graphs with labels on half-edges}

Fix an integer $N\geq 1$.
By a dotted $N$-labelled stable graph,
we mean a dotted stable graph with a label in $\{1,2,\cdots,N\}$
assigned to every of its half-edges.
Similar to the case of one-dimensional state space discussed in \S \ref{sec:dual-diagrammatics},
we will define a map that takes
these dotted stable graphs with labels on half-edges
to linear combinations of labelled ordinary graphs in this subsection.

Let us first consider the labelled dotted stable vertices.
Denote by $V_{g;l_1,\cdots,l_N}^\vee$ a dotted stable vertex of genus $g$
and valence $n=l_1+\cdots+l_N$,
where $l_j$ is the number of $j$ appearing in all the $n$ labels
on the half-edges incident at this vertex.
Similar to the one-dimensional case,
we define this dotted vertex to be
the following linear combination of labelled ordinary graphs:
\ben
V_{g;l_1,\cdots,l_N}^\vee:=l_1!\cdots l_N!\cdot
\wcF_{g;l_1,\cdots,l_N}
\een
when $2g-2+\sum\limits_{j=1}^N l_j>0$, i.e.,
\ben
V_{g;l_1,\cdots,l_N}^\vee=l_1!\cdots l_N!\cdot
\sum_{\Gamma\in\cG_{g;l_1,\cdots,l_N}^c}\frac{1}{|\Aut(\Gamma)|}\Gamma.
\een

Lemma \ref{lem-original-N-rec} tells us:
\ben
V_{g;l_1,\cdots,l_j+1,\cdots,l_N}^\vee=\cD_j V_{g;l_1,\cdots,l_N}^\vee,
\een
therefore we have
\ben
V_{g;l_1,\cdots,l_N}^\vee=\cD_1^{l_1}\cdots\cD_N^{l_N} V_{g}^\vee,
\een
where $V_g:=V_{g;0,\cdots,0}=\wcF_g$ for every $g\geq 2$,
and for $g=0$, $1$ we use the following convention:
\ben
&&\cD_j V_1^\vee:=
\wcF_{1;l_1,\cdots,l_N},\quad
\text{where } l_p:=\delta_{j,p};\\
&&\cD_i \cD_j \cD_k V_0^\vee:=l_1!\cdots l_N!\cdot
\wcF_{0;l_1,\cdots,l_N},\quad
\text{where } l_p:=\delta_{i,p}+\delta_{j,p}+\delta_{k,p}.
\een

\begin{Example}
For $N=2$, we have the following examples.
\begin{flalign*}
\begin{tikzpicture}
\draw [dotted,thick](0,0) circle [radius=0.2];
\node [align=center,align=center] at (0,0) {$0$};
\draw [dotted,thick](-0.5,0)--(-0.2,0);
\node [align=center,left] at (-0.5,0) {$1$};
\draw [dotted,thick](0.16,0.1)--(0.5,0.2);
\draw [dotted,thick](0.16,-0.1)--(0.5,-0.2);
\node [align=center,right] at (0.5,0.2) {$2$};
\node [align=center,right] at (0.5,-0.2) {$2$};
\node [align=center,align=center] at (1.2,0) {$=$};
\draw (0+2.4,0) circle [radius=0.2];
\node [align=center,align=center] at (0+2.4,0) {$0$};
\draw (-0.5+2.4,0)--(-0.2+2.4,0);
\node [align=center,left] at (-0.5+2.4,0) {$1$};
\draw (0.16+2.4,0.1)--(0.5+2.4,0.2);
\draw (0.16+2.4,-0.1)--(0.5+2.4,-0.2);
\node [align=center,right] at (0.5+2.4,0.2) {$2$};
\node [align=center,right] at (0.5+2.4,-0.2) {$2$};
\end{tikzpicture},&&
\end{flalign*}

\begin{flalign*}
\begin{split}
&\begin{tikzpicture}
\draw [dotted,thick](0,0) circle [radius=0.2];
\node [align=center,align=center] at (0,0) {$0$};
\draw [dotted,thick](-0.16,0.1)--(-0.5,0.2);
\draw [dotted,thick](-0.16,-0.1)--(-0.5,-0.2);
\node [align=center,left] at (-0.5,0.2) {$1$};
\node [align=center,left] at (-0.5,-0.2) {$1$};
\draw [dotted,thick](0.16,0.1)--(0.5,0.2);
\draw [dotted,thick](0.16,-0.1)--(0.5,-0.2);
\node [align=center,right] at (0.5,0.2) {$2$};
\node [align=center,right] at (0.5,-0.2) {$2$};
\node [align=center,align=center] at (1.2,0) {$=$};
\draw (0+2.4,0) circle [radius=0.2];
\node [align=center,align=center] at (0+2.4,0) {$0$};
\draw (-0.16+2.4,0.1)--(-0.5+2.4,0.2);
\draw (-0.16+2.4,-0.1)--(-0.5+2.4,-0.2);
\node [align=center,left] at (-0.5+2.4,0.2) {$1$};
\node [align=center,left] at (-0.5+2.4,-0.2) {$1$};
\draw (0.16+2.4,0.1)--(0.5+2.4,0.2);
\draw (0.16+2.4,-0.1)--(0.5+2.4,-0.2);
\node [align=center,right] at (0.5+2.4,0.2) {$2$};
\node [align=center,right] at (0.5+2.4,-0.2) {$2$};
\node [align=center,align=center] at (3.6,0) {$+$};
\draw (0+4.8,0) circle [radius=0.2];
\node [align=center,align=center] at (0+4.8,0) {$0$};
\draw (-0.16+4.8,0.1)--(-0.5+4.8,0.2);
\draw (-0.16+4.8,-0.1)--(-0.5+4.8,-0.2);
\node [align=center,left] at (-0.5+4.8,0.2) {$1$};
\node [align=center,left] at (-0.5+4.8,-0.2) {$1$};
\draw (0+5.6,0) circle [radius=0.2];
\node [align=center,align=center] at (0+5.6,0) {$0$};
\draw (0.16+5.6,0.1)--(0.5+5.6,0.2);
\draw (0.16+5.6,-0.1)--(0.5+5.6,-0.2);
\node [align=center,right] at (0.5+5.6,0.2) {$2$};
\node [align=center,right] at (0.5+5.6,-0.2) {$2$};
\draw (5,0)--(5.4,0);
\node [above,align=center] at (5,0) {$1$};
\node [above,align=center] at (5.4,0) {$1$};
\node [align=center,align=center] at (3.6+3.2,0) {$+$};
\draw (0+4.8+3.2,0) circle [radius=0.2];
\node [align=center,align=center] at (0+4.8+3.2,0) {$0$};
\draw (-0.16+4.8+3.2,0.1)--(-0.5+4.8+3.2,0.2);
\draw (-0.16+4.8+3.2,-0.1)--(-0.5+4.8+3.2,-0.2);
\node [align=center,left] at (-0.5+4.8+3.2,0.2) {$1$};
\node [align=center,left] at (-0.5+4.8+3.2,-0.2) {$1$};
\draw (0+5.6+3.2,0) circle [radius=0.2];
\node [align=center,align=center] at (0+5.6+3.2,0) {$0$};
\draw (0.16+5.6+3.2,0.1)--(0.5+5.6+3.2,0.2);
\draw (0.16+5.6+3.2,-0.1)--(0.5+5.6+3.2,-0.2);
\node [align=center,right] at (0.5+5.6+3.2,0.2) {$2$};
\node [align=center,right] at (0.5+5.6+3.2,-0.2) {$2$};
\draw (5+3.2,0)--(5.4+3.2,0);
\node [above,align=center] at (5+3.2,0) {$1$};
\node [above,align=center] at (5.4+3.2,0) {$2$};
\end{tikzpicture}
\\
&\qquad\qquad\qquad
\begin{tikzpicture}
\node [align=center,align=center] at (0.4,0) {$+$};
\draw (0+2.4,0) circle [radius=0.2];
\node [align=center,align=center] at (0+2.4,0) {$0$};
\draw (0+1.6,0) circle [radius=0.2];
\node [align=center,align=center] at (0+1.6,0) {$0$};
\draw (1.8,0)--(2.2,0);
\node [above,align=center] at (1.8,0) {$2$};
\node [above,align=center] at (2.2,0) {$1$};
\draw (-0.16+2.4-0.8,0.1)--(-0.5+2.4-0.8,0.2);
\draw (-0.16+2.4-0.8,-0.1)--(-0.5+2.4-0.8,-0.2);
\node [align=center,left] at (-0.5+2.4-0.8,0.2) {$1$};
\node [align=center,left] at (-0.5+2.4-0.8,-0.2) {$1$};
\draw (0.16+2.4,0.1)--(0.5+2.4,0.2);
\draw (0.16+2.4,-0.1)--(0.5+2.4,-0.2);
\node [align=center,right] at (0.5+2.4,0.2) {$2$};
\node [align=center,right] at (0.5+2.4,-0.2) {$2$};
\node [align=center,align=center] at (3.6,0) {$+$};
\draw (0+4.8,0) circle [radius=0.2];
\node [align=center,align=center] at (0+4.8,0) {$0$};
\draw (-0.16+4.8,0.1)--(-0.5+4.8,0.2);
\draw (-0.16+4.8,-0.1)--(-0.5+4.8,-0.2);
\node [align=center,left] at (-0.5+4.8,0.2) {$1$};
\node [align=center,left] at (-0.5+4.8,-0.2) {$1$};
\draw (0+5.6,0) circle [radius=0.2];
\node [align=center,align=center] at (0+5.6,0) {$0$};
\draw (0.16+5.6,0.1)--(0.5+5.6,0.2);
\draw (0.16+5.6,-0.1)--(0.5+5.6,-0.2);
\node [align=center,right] at (0.5+5.6,0.2) {$2$};
\node [align=center,right] at (0.5+5.6,-0.2) {$2$};
\draw (5,0)--(5.4,0);
\node [above,align=center] at (5,0) {$2$};
\node [above,align=center] at (5.4,0) {$2$};
\node [align=center,align=center] at (3.6+3.2,0) {$+2$};
\draw (0+4.8+3.2,0) circle [radius=0.2];
\node [align=center,align=center] at (0+4.8+3.2,0) {$0$};
\draw (-0.16+4.8+3.2,0.1)--(-0.5+4.8+3.2,0.2);
\draw (-0.16+4.8+3.2,-0.1)--(-0.5+4.8+3.2,-0.2);
\node [align=center,left] at (-0.5+4.8+3.2,0.2) {$1$};
\node [align=center,left] at (-0.5+4.8+3.2,-0.2) {$2$};
\draw (0+5.6+3.2,0) circle [radius=0.2];
\node [align=center,align=center] at (0+5.6+3.2,0) {$0$};
\draw (0.16+5.6+3.2,0.1)--(0.5+5.6+3.2,0.2);
\draw (0.16+5.6+3.2,-0.1)--(0.5+5.6+3.2,-0.2);
\node [align=center,right] at (0.5+5.6+3.2,0.2) {$1$};
\node [align=center,right] at (0.5+5.6+3.2,-0.2) {$2$};
\draw (5+3.2,0)--(5.4+3.2,0);
\node [above,align=center] at (5+3.2,0) {$1$};
\node [above,align=center] at (5.4+3.2,0) {$1$};
\end{tikzpicture}
\\
&\qquad\qquad\qquad
\begin{tikzpicture}
\node [align=center,align=center] at (0.4,0) {$+2$};
\draw (0+2.4,0) circle [radius=0.2];
\node [align=center,align=center] at (0+2.4,0) {$0$};
\draw (0+1.6,0) circle [radius=0.2];
\node [align=center,align=center] at (0+1.6,0) {$0$};
\draw (1.8,0)--(2.2,0);
\node [above,align=center] at (1.8,0) {$2$};
\node [above,align=center] at (2.2,0) {$2$};
\draw (-0.16+2.4-0.8,0.1)--(-0.5+2.4-0.8,0.2);
\draw (-0.16+2.4-0.8,-0.1)--(-0.5+2.4-0.8,-0.2);
\node [align=center,left] at (-0.5+2.4-0.8,0.2) {$1$};
\node [align=center,left] at (-0.5+2.4-0.8,-0.2) {$2$};
\draw (0.16+2.4,0.1)--(0.5+2.4,0.2);
\draw (0.16+2.4,-0.1)--(0.5+2.4,-0.2);
\node [align=center,right] at (0.5+2.4,0.2) {$1$};
\node [align=center,right] at (0.5+2.4,-0.2) {$2$};
\node [align=center,align=center] at (3.6,0) {$+4$};
\draw (0+4.8,0) circle [radius=0.2];
\node [align=center,align=center] at (0+4.8,0) {$0$};
\draw (-0.16+4.8,0.1)--(-0.5+4.8,0.2);
\draw (-0.16+4.8,-0.1)--(-0.5+4.8,-0.2);
\node [align=center,left] at (-0.5+4.8,0.2) {$1$};
\node [align=center,left] at (-0.5+4.8,-0.2) {$2$};
\draw (0+5.6,0) circle [radius=0.2];
\node [align=center,align=center] at (0+5.6,0) {$0$};
\draw (0.16+5.6,0.1)--(0.5+5.6,0.2);
\draw (0.16+5.6,-0.1)--(0.5+5.6,-0.2);
\node [align=center,right] at (0.5+5.6,0.2) {$1$};
\node [align=center,right] at (0.5+5.6,-0.2) {$2$};
\draw (5,0)--(5.4,0);
\node [above,align=center] at (5,0) {$1$};
\node [above,align=center] at (5.4,0) {$2$};
\end{tikzpicture}.
\end{split}&&
\end{flalign*}
\end{Example}

We have assigned to a labelled dotted stable vertex
a linear combination of labelled stable graphs.
Now we can do the same for labelled dotted stable graphs by gluing these labelled
dotted stable vertices together, such that
an internal edge labelled by $i,j$
is obtained by gluing two external edges
labelled by $i$ and $j$ together.
The gluing procedure is the same as the unlabelled case (see \S \ref{subsec3.2}).
We first give different names to half-edges of a dotted stable graph
and cut off all internal edges,
then write every dotted vertices in terms of labelled ordinary graphs
with names on external edges.
We glue two external edges together if their ancestors are joined together.
Finally we forget the names and multiply by a factor $(-1)^{|E^\vee|}$.

\begin{Example}
For the case of $N=2$,
let us take the following dotted stable graph as an example:
\be\label{eg-dotted-label}
\begin{tikzpicture}
\draw [dotted,thick](0,0) circle [radius=0.2];
\node [align=center,align=center] at (0,0) {$0$};
\draw [dotted,thick](0.16,0.1) .. controls (0.5,0.2) and (0.5,-0.2) ..  (0.16,-0.1);
\draw [dotted,thick](-0.16,0.1) .. controls (-0.5,0.2) and (-0.5,-0.2) ..  (-0.16,-0.1);
\node [above,left] at (-0.2,0.3) {$1$};
\node [below,left] at (-0.2,-0.3) {$1$};
\node [above,right] at (0.2,0.3) {$2$};
\node [below,right] at (0.2,-0.3) {$2$};
\end{tikzpicture}.
\ee
We first give four names $\{a,b,c,d\}$ to the half-edges of this graph,
and then cut the internal edges off:
\ben
\begin{tikzpicture}
\draw [dotted,thick](0,0) circle [radius=0.2];
\node [align=center,align=center] at (0,0) {$0$};
\draw [dotted,thick](0.16,0.1) .. controls (0.5,0.2) and (0.5,-0.2) ..  (0.16,-0.1);
\draw [dotted,thick](-0.16,0.1) .. controls (-0.5,0.2) and (-0.5,-0.2) ..  (-0.16,-0.1);
\node [above,left] at (-0.2,0.3) {$1$};
\node [below,left] at (-0.2,-0.3) {$1$};
\node [above,right] at (0.2,0.3) {$2$};
\node [below,right] at (0.2,-0.3) {$2$};
\node [above,left] at (-0.5,0.3) {$a$};
\node [below,left] at (-0.5,-0.3) {$b$};
\node [above,right] at (0.5,0.3) {$c$};
\node [below,right] at (0.5,-0.3) {$d$};
\node [align=center,align=center] at (2,0) {$\to$};
\draw [dotted,thick](4,0) circle [radius=0.2];
\node [align=center,align=center] at (4,0) {$0$};
\draw [dotted,thick](4.16,0.1)--(4.5,0.2);
\draw [dotted,thick](4.16,-0.1)--(4.5,-0.2);
\draw [dotted,thick](3.84,0.1)--(3.5,0.2);
\draw [dotted,thick](3.84,-0.1)--(3.5,-0.2);
\node [above,left] at (-0.4+4,0.3) {$1$};
\node [below,left] at (-0.4+4,-0.3) {$1$};
\node [above,right] at (0.4+4,0.3) {$2$};
\node [below,right] at (0.4+4,-0.3) {$2$};
\node [above,left] at (-0.6+4,0.3) {$a$};
\node [below,left] at (-0.6+4,-0.3) {$b$};
\node [above,right] at (0.6+4,0.3) {$c$};
\node [below,right] at (0.6+4,-0.3) {$d$};
\end{tikzpicture}.
\een
The labelled dotted stable vertex with names on the right-hand-side
can be written as the following linear combination:

\begin{flalign*}
\begin{split}
&
\begin{tikzpicture}
\draw (0+2.4,0) circle [radius=0.2];
\node [align=center,align=center] at (0+2.4,0) {$0$};
\draw (-0.16+2.4,0.1)--(-0.5+2.4,0.2);
\draw (-0.16+2.4,-0.1)--(-0.5+2.4,-0.2);
\node [align=center,left] at (-0.5+2.4,0.2) {$a1$};
\node [align=center,left] at (-0.5+2.4,-0.2) {$b1$};
\draw (0.16+2.4,0.1)--(0.5+2.4,0.2);
\draw (0.16+2.4,-0.1)--(0.5+2.4,-0.2);
\node [align=center,right] at (0.5+2.4,0.2) {$2c$};
\node [align=center,right] at (0.5+2.4,-0.2) {$2d$};
\node [align=center,align=center] at (3.6,0) {$+$};
\draw (0+4.8,0) circle [radius=0.2];
\node [align=center,align=center] at (0+4.8,0) {$0$};
\draw (-0.16+4.8,0.1)--(-0.5+4.8,0.2);
\draw (-0.16+4.8,-0.1)--(-0.5+4.8,-0.2);
\node [align=center,left] at (-0.5+4.8,0.2) {$a1$};
\node [align=center,left] at (-0.5+4.8,-0.2) {$b1$};
\draw (0+5.6,0) circle [radius=0.2];
\node [align=center,align=center] at (0+5.6,0) {$0$};
\draw (0.16+5.6,0.1)--(0.5+5.6,0.2);
\draw (0.16+5.6,-0.1)--(0.5+5.6,-0.2);
\node [align=center,right] at (0.5+5.6,0.2) {$2c$};
\node [align=center,right] at (0.5+5.6,-0.2) {$2d$};
\draw (5,0)--(5.4,0);
\node [above,align=center] at (5,0) {$1$};
\node [above,align=center] at (5.4,0) {$1$};
\node [align=center,align=center] at (3.6+3.2,0) {$+$};
\draw (0+4.8+3.2,0) circle [radius=0.2];
\node [align=center,align=center] at (0+4.8+3.2,0) {$0$};
\draw (-0.16+4.8+3.2,0.1)--(-0.5+4.8+3.2,0.2);
\draw (-0.16+4.8+3.2,-0.1)--(-0.5+4.8+3.2,-0.2);
\node [align=center,left] at (-0.5+4.8+3.2,0.2) {$a1$};
\node [align=center,left] at (-0.5+4.8+3.2,-0.2) {$b1$};
\draw (0+5.6+3.2,0) circle [radius=0.2];
\node [align=center,align=center] at (0+5.6+3.2,0) {$0$};
\draw (0.16+5.6+3.2,0.1)--(0.5+5.6+3.2,0.2);
\draw (0.16+5.6+3.2,-0.1)--(0.5+5.6+3.2,-0.2);
\node [align=center,right] at (0.5+5.6+3.2,0.2) {$2c$};
\node [align=center,right] at (0.5+5.6+3.2,-0.2) {$2d$};
\draw (5+3.2,0)--(5.4+3.2,0);
\node [above,align=center] at (5+3.2,0) {$1$};
\node [above,align=center] at (5.4+3.2,0) {$2$};
\node [align=center,align=center] at (3.6+3.2+3.2,0) {$+$};
\draw (0+4.8+3.2+3.2,0) circle [radius=0.2];
\node [align=center,align=center] at (0+4.8+3.2+3.2,0) {$0$};
\draw (-0.16+4.8+3.2+3.2,0.1)--(-0.5+4.8+3.2+3.2,0.2);
\draw (-0.16+4.8+3.2+3.2,-0.1)--(-0.5+4.8+3.2+3.2,-0.2);
\node [align=center,left] at (-0.5+4.8+3.2+3.2,0.2) {$a1$};
\node [align=center,left] at (-0.5+4.8+3.2+3.2,-0.2) {$b1$};
\draw (0+5.6+3.2+3.2,0) circle [radius=0.2];
\node [align=center,align=center] at (0+5.6+3.2+3.2,0) {$0$};
\draw (0.16+5.6+3.2+3.2,0.1)--(0.5+5.6+3.2+3.2,0.2);
\draw (0.16+5.6+3.2+3.2,-0.1)--(0.5+5.6+3.2+3.2,-0.2);
\node [align=center,right] at (0.5+5.6+3.2+3.2,0.2) {$2c$};
\node [align=center,right] at (0.5+5.6+3.2+3.2,-0.2) {$2d$};
\draw (5+3.2+3.2,0)--(5.4+3.2+3.2,0);
\node [above,align=center] at (5+3.2+3.2,0) {$2$};
\node [above,align=center] at (5.4+3.2+3.2,0) {$1$};
\end{tikzpicture}
\\
&
\begin{tikzpicture}
\node [align=center,align=center] at (3.6,0) {$+$};
\draw (0+4.8,0) circle [radius=0.2];
\node [align=center,align=center] at (0+4.8,0) {$0$};
\draw (-0.16+4.8,0.1)--(-0.5+4.8,0.2);
\draw (-0.16+4.8,-0.1)--(-0.5+4.8,-0.2);
\node [align=center,left] at (-0.5+4.8,0.2) {$a1$};
\node [align=center,left] at (-0.5+4.8,-0.2) {$b1$};
\draw (0+5.6,0) circle [radius=0.2];
\node [align=center,align=center] at (0+5.6,0) {$0$};
\draw (0.16+5.6,0.1)--(0.5+5.6,0.2);
\draw (0.16+5.6,-0.1)--(0.5+5.6,-0.2);
\node [align=center,right] at (0.5+5.6,0.2) {$2c$};
\node [align=center,right] at (0.5+5.6,-0.2) {$2d$};
\draw (5,0)--(5.4,0);
\node [above,align=center] at (5,0) {$2$};
\node [above,align=center] at (5.4,0) {$2$};
\node [align=center,align=center] at (3.6+3.2,0) {$+$};
\draw (0+4.8+3.2,0) circle [radius=0.2];
\node [align=center,align=center] at (0+4.8+3.2,0) {$0$};
\draw (-0.16+4.8+3.2,0.1)--(-0.5+4.8+3.2,0.2);
\draw (-0.16+4.8+3.2,-0.1)--(-0.5+4.8+3.2,-0.2);
\node [align=center,left] at (-0.5+4.8+3.2,0.2) {$a1$};
\node [align=center,left] at (-0.5+4.8+3.2,-0.2) {$c2$};
\draw (0+5.6+3.2,0) circle [radius=0.2];
\node [align=center,align=center] at (0+5.6+3.2,0) {$0$};
\draw (0.16+5.6+3.2,0.1)--(0.5+5.6+3.2,0.2);
\draw (0.16+5.6+3.2,-0.1)--(0.5+5.6+3.2,-0.2);
\node [align=center,right] at (0.5+5.6+3.2,0.2) {$1b$};
\node [align=center,right] at (0.5+5.6+3.2,-0.2) {$2d$};
\draw (5+3.2,0)--(5.4+3.2,0);
\node [above,align=center] at (5+3.2,0) {$1$};
\node [above,align=center] at (5.4+3.2,0) {$1$};
\node [align=center,align=center] at (3.6+3.2+3.2,0) {$+$};
\draw (0+4.8+3.2+3.2,0) circle [radius=0.2];
\node [align=center,align=center] at (0+4.8+3.2+3.2,0) {$0$};
\draw (-0.16+4.8+3.2+3.2,0.1)--(-0.5+4.8+3.2+3.2,0.2);
\draw (-0.16+4.8+3.2+3.2,-0.1)--(-0.5+4.8+3.2+3.2,-0.2);
\node [align=center,left] at (-0.5+4.8+3.2+3.2,0.2) {$a1$};
\node [align=center,left] at (-0.5+4.8+3.2+3.2,-0.2) {$d2$};
\draw (0+5.6+3.2+3.2,0) circle [radius=0.2];
\node [align=center,align=center] at (0+5.6+3.2+3.2,0) {$0$};
\draw (0.16+5.6+3.2+3.2,0.1)--(0.5+5.6+3.2+3.2,0.2);
\draw (0.16+5.6+3.2+3.2,-0.1)--(0.5+5.6+3.2+3.2,-0.2);
\node [align=center,right] at (0.5+5.6+3.2+3.2,0.2) {$1b$};
\node [align=center,right] at (0.5+5.6+3.2+3.2,-0.2) {$2c$};
\draw (5+3.2+3.2,0)--(5.4+3.2+3.2,0);
\node [above,align=center] at (5+3.2+3.2,0) {$1$};
\node [above,align=center] at (5.4+3.2+3.2,0) {$1$};
\end{tikzpicture}
\\
&
\begin{tikzpicture}
\node [align=center,align=center] at (3.6,0) {$+$};
\draw (0+4.8,0) circle [radius=0.2];
\node [align=center,align=center] at (0+4.8,0) {$0$};
\draw (-0.16+4.8,0.1)--(-0.5+4.8,0.2);
\draw (-0.16+4.8,-0.1)--(-0.5+4.8,-0.2);
\node [align=center,left] at (-0.5+4.8,0.2) {$a1$};
\node [align=center,left] at (-0.5+4.8,-0.2) {$c2$};
\draw (0+5.6,0) circle [radius=0.2];
\node [align=center,align=center] at (0+5.6,0) {$0$};
\draw (0.16+5.6,0.1)--(0.5+5.6,0.2);
\draw (0.16+5.6,-0.1)--(0.5+5.6,-0.2);
\node [align=center,right] at (0.5+5.6,0.2) {$1b$};
\node [align=center,right] at (0.5+5.6,-0.2) {$2d$};
\draw (5,0)--(5.4,0);
\node [above,align=center] at (5,0) {$2$};
\node [above,align=center] at (5.4,0) {$2$};
\node [align=center,align=center] at (3.6+3.2,0) {$+$};
\draw (0+4.8+3.2,0) circle [radius=0.2];
\node [align=center,align=center] at (0+4.8+3.2,0) {$0$};
\draw (-0.16+4.8+3.2,0.1)--(-0.5+4.8+3.2,0.2);
\draw (-0.16+4.8+3.2,-0.1)--(-0.5+4.8+3.2,-0.2);
\node [align=center,left] at (-0.5+4.8+3.2,0.2) {$a1$};
\node [align=center,left] at (-0.5+4.8+3.2,-0.2) {$d2$};
\draw (0+5.6+3.2,0) circle [radius=0.2];
\node [align=center,align=center] at (0+5.6+3.2,0) {$0$};
\draw (0.16+5.6+3.2,0.1)--(0.5+5.6+3.2,0.2);
\draw (0.16+5.6+3.2,-0.1)--(0.5+5.6+3.2,-0.2);
\node [align=center,right] at (0.5+5.6+3.2,0.2) {$1b$};
\node [align=center,right] at (0.5+5.6+3.2,-0.2) {$2c$};
\draw (5+3.2,0)--(5.4+3.2,0);
\node [above,align=center] at (5+3.2,0) {$2$};
\node [above,align=center] at (5.4+3.2,0) {$2$};
\node [align=center,align=center] at (3.6+3.2+3.2,0) {$+$};
\draw (0+4.8+3.2+3.2,0) circle [radius=0.2];
\node [align=center,align=center] at (0+4.8+3.2+3.2,0) {$0$};
\draw (-0.16+4.8+3.2+3.2,0.1)--(-0.5+4.8+3.2+3.2,0.2);
\draw (-0.16+4.8+3.2+3.2,-0.1)--(-0.5+4.8+3.2+3.2,-0.2);
\node [align=center,left] at (-0.5+4.8+3.2+3.2,0.2) {$a1$};
\node [align=center,left] at (-0.5+4.8+3.2+3.2,-0.2) {$c2$};
\draw (0+5.6+3.2+3.2,0) circle [radius=0.2];
\node [align=center,align=center] at (0+5.6+3.2+3.2,0) {$0$};
\draw (0.16+5.6+3.2+3.2,0.1)--(0.5+5.6+3.2+3.2,0.2);
\draw (0.16+5.6+3.2+3.2,-0.1)--(0.5+5.6+3.2+3.2,-0.2);
\node [align=center,right] at (0.5+5.6+3.2+3.2,0.2) {$1b$};
\node [align=center,right] at (0.5+5.6+3.2+3.2,-0.2) {$2d$};
\draw (5+3.2+3.2,0)--(5.4+3.2+3.2,0);
\node [above,align=center] at (5+3.2+3.2,0) {$1$};
\node [above,align=center] at (5.4+3.2+3.2,0) {$2$};
\end{tikzpicture}
\\
&
\begin{tikzpicture}
\node [align=center,align=center] at (3.6,0) {$+$};
\draw (0+4.8,0) circle [radius=0.2];
\node [align=center,align=center] at (0+4.8,0) {$0$};
\draw (-0.16+4.8,0.1)--(-0.5+4.8,0.2);
\draw (-0.16+4.8,-0.1)--(-0.5+4.8,-0.2);
\node [align=center,left] at (-0.5+4.8,0.2) {$a1$};
\node [align=center,left] at (-0.5+4.8,-0.2) {$d2$};
\draw (0+5.6,0) circle [radius=0.2];
\node [align=center,align=center] at (0+5.6,0) {$0$};
\draw (0.16+5.6,0.1)--(0.5+5.6,0.2);
\draw (0.16+5.6,-0.1)--(0.5+5.6,-0.2);
\node [align=center,right] at (0.5+5.6,0.2) {$1b$};
\node [align=center,right] at (0.5+5.6,-0.2) {$2c$};
\draw (5,0)--(5.4,0);
\node [above,align=center] at (5,0) {$1$};
\node [above,align=center] at (5.4,0) {$2$};
\node [align=center,align=center] at (3.6+3.2,0) {$+$};
\draw (0+4.8+3.2,0) circle [radius=0.2];
\node [align=center,align=center] at (0+4.8+3.2,0) {$0$};
\draw (-0.16+4.8+3.2,0.1)--(-0.5+4.8+3.2,0.2);
\draw (-0.16+4.8+3.2,-0.1)--(-0.5+4.8+3.2,-0.2);
\node [align=center,left] at (-0.5+4.8+3.2,0.2) {$b1$};
\node [align=center,left] at (-0.5+4.8+3.2,-0.2) {$c2$};
\draw (0+5.6+3.2,0) circle [radius=0.2];
\node [align=center,align=center] at (0+5.6+3.2,0) {$0$};
\draw (0.16+5.6+3.2,0.1)--(0.5+5.6+3.2,0.2);
\draw (0.16+5.6+3.2,-0.1)--(0.5+5.6+3.2,-0.2);
\node [align=center,right] at (0.5+5.6+3.2,0.2) {$1a$};
\node [align=center,right] at (0.5+5.6+3.2,-0.2) {$2d$};
\draw (5+3.2,0)--(5.4+3.2,0);
\node [above,align=center] at (5+3.2,0) {$1$};
\node [above,align=center] at (5.4+3.2,0) {$2$};
\node [align=center,align=center] at (3.6+3.2+3.2,0) {$+$};
\draw (0+4.8+3.2+3.2,0) circle [radius=0.2];
\node [align=center,align=center] at (0+4.8+3.2+3.2,0) {$0$};
\draw (-0.16+4.8+3.2+3.2,0.1)--(-0.5+4.8+3.2+3.2,0.2);
\draw (-0.16+4.8+3.2+3.2,-0.1)--(-0.5+4.8+3.2+3.2,-0.2);
\node [align=center,left] at (-0.5+4.8+3.2+3.2,0.2) {$b1$};
\node [align=center,left] at (-0.5+4.8+3.2+3.2,-0.2) {$d2$};
\draw (0+5.6+3.2+3.2,0) circle [radius=0.2];
\node [align=center,align=center] at (0+5.6+3.2+3.2,0) {$0$};
\draw (0.16+5.6+3.2+3.2,0.1)--(0.5+5.6+3.2+3.2,0.2);
\draw (0.16+5.6+3.2+3.2,-0.1)--(0.5+5.6+3.2+3.2,-0.2);
\node [align=center,right] at (0.5+5.6+3.2+3.2,0.2) {$1a$};
\node [align=center,right] at (0.5+5.6+3.2+3.2,-0.2) {$2c$};
\draw (5+3.2+3.2,0)--(5.4+3.2+3.2,0);
\node [above,align=center] at (5+3.2+3.2,0) {$1$};
\node [above,align=center] at (5.4+3.2+3.2,0) {$2$};
\end{tikzpicture}.
\end{split}&&
\end{flalign*}
Finally, we glue the half-edge named by $a$ with the half-edge named by $b$,
and glue the half-edge named by $c$ with the half-edge named by $d$.
then we forget the names, and multiply by $(-1)^2$.
In this way
we can obtain the following expression for the graph \eqref{eg-dotted-label}:
\begin{flalign*}
\begin{split}
&
\begin{tikzpicture}
\draw (0,0) circle [radius=0.2];
\node [align=center,align=center] at (0,0) {$0$};
\draw (0.16,0.1) .. controls (0.5,0.2) and (0.5,-0.2) ..  (0.16,-0.1);
\draw (-0.16,0.1) .. controls (-0.5,0.2) and (-0.5,-0.2) ..  (-0.16,-0.1);
\node [above,left] at (-0.2,0.3) {$1$};
\node [below,left] at (-0.2,-0.3) {$1$};
\node [above,right] at (0.2,0.3) {$2$};
\node [below,right] at (0.2,-0.3) {$2$};
\node [align=center,align=center] at (0.8,0) {$+$};
\draw (1.6,0) circle [radius=0.2];
\node [align=center,align=center] at (1.6,0) {$0$};
\draw (2.4,0) circle [radius=0.2];
\node [align=center,align=center] at (2.4,0) {$0$};
\draw (1.8,0) -- (2.2,0);
\draw (0.16+2.4,0.1) .. controls (0.5+2.4,0.2) and (0.5+2.4,-0.2) ..  (0.16+2.4,-0.1);
\draw (-0.16+1.6,0.1) .. controls (-0.5+1.6,0.2) and (-0.5+1.6,-0.2) ..  (-0.16+1.6,-0.1);
\node [above,left] at (-0.2+1.6,0.3) {$1$};
\node [below,left] at (-0.2+1.6,-0.3) {$1$};
\node [above,right] at (0.2+2.4,0.3) {$2$};
\node [below,right] at (0.2+2.4,-0.3) {$2$};
\node [above,align=center] at (1.8,0) {$1$};
\node [above,align=center] at (2.2,0) {$1$};
\node [align=center,align=center] at (0.8+2.4,0) {$+$};
\draw (1.6+2.4,0) circle [radius=0.2];
\node [align=center,align=center] at (1.6+2.4,0) {$0$};
\draw (2.4+2.4,0) circle [radius=0.2];
\node [align=center,align=center] at (2.4+2.4,0) {$0$};
\draw (1.8+2.4,0) -- (2.2+2.4,0);
\draw (0.16+2.4+2.4,0.1) .. controls (0.5+2.4+2.4,0.2) and (0.5+2.4+2.4,-0.2) ..  (0.16+2.4+2.4,-0.1);
\draw (-0.16+1.6+2.4,0.1) .. controls (-0.5+1.6+2.4,0.2) and (-0.5+1.6+2.4,-0.2) ..  (-0.16+1.6+2.4,-0.1);
\node [above,left] at (-0.2+1.6+2.4,0.3) {$1$};
\node [below,left] at (-0.2+1.6+2.4,-0.3) {$1$};
\node [above,right] at (0.2+2.4+2.4,0.3) {$2$};
\node [below,right] at (0.2+2.4+2.4,-0.3) {$2$};
\node [above,align=center] at (1.8+2.4,0) {$1$};
\node [above,align=center] at (2.2+2.4,0) {$2$};
\node [align=center,align=center] at (0.8+2.4+2.4,0) {$+$};
\draw (1.6+2.4+2.4,0) circle [radius=0.2];
\node [align=center,align=center] at (1.6+2.4+2.4,0) {$0$};
\draw (2.4+2.4+2.4,0) circle [radius=0.2];
\node [align=center,align=center] at (2.4+2.4+2.4,0) {$0$};
\draw (1.8+2.4+2.4,0) -- (2.2+2.4+2.4,0);
\draw (0.16+2.4+2.4+2.4,0.1) .. controls (0.5+2.4+2.4+2.4,0.2) and (0.5+2.4+2.4+2.4,-0.2) ..  (0.16+2.4+2.4+2.4,-0.1);
\draw (-0.16+1.6+2.4+2.4,0.1) .. controls (-0.5+1.6+2.4+2.4,0.2) and (-0.5+1.6+2.4+2.4,-0.2) ..  (-0.16+2.4+1.6+2.4,-0.1);
\node [above,left] at (-0.2+1.6+2.4+2.4,0.3) {$1$};
\node [below,left] at (-0.2+1.6+2.4+2.4,-0.3) {$1$};
\node [above,right] at (0.2+2.4+2.4+2.4,0.3) {$2$};
\node [below,right] at (0.2+2.4+2.4+2.4,-0.3) {$2$};
\node [above,align=center] at (1.8+2.4+2.4,0) {$2$};
\node [above,align=center] at (2.2+2.4+2.4,0) {$1$};
\node [align=center,align=center] at (0.8+2.4+2.4+2.4,0) {$+$};
\draw (1.6+2.4+2.4+2.4,0) circle [radius=0.2];
\node [align=center,align=center] at (1.6+2.4+2.4+2.4,0) {$0$};
\draw (2.4+2.4+2.4+2.4,0) circle [radius=0.2];
\node [align=center,align=center] at (2.4+2.4+2.4+2.4,0) {$0$};
\draw (1.8+2.4+2.4+2.4,0) -- (2.2+2.4+2.4+2.4,0);
\draw (0.16+2.4+2.4+2.4+2.4,0.1) .. controls (0.5+2.4+2.4+2.4+2.4,0.2) and (0.5+2.4+2.4+2.4+2.4,-0.2) ..  (0.16+2.4+2.4+2.4+2.4,-0.1);
\draw (-0.16+1.6+2.4+2.4+2.4,0.1) .. controls (-0.5+1.6+2.4+2.4+2.4,0.2) and (-0.5+2.4+1.6+2.4+2.4,-0.2) ..  (-0.16+2.4+1.6+2.4+2.4,-0.1);
\node [above,left] at (-0.2+1.6+2.4+2.4+2.4,0.3) {$1$};
\node [below,left] at (-0.2+1.6+2.4+2.4+2.4,-0.3) {$1$};
\node [above,right] at (0.2+2.4+2.4+2.4+2.4,0.3) {$2$};
\node [below,right] at (0.2+2.4+2.4+2.4+2.4,-0.3) {$2$};
\node [above,align=center] at (1.8+2.4+2.4+2.4,0) {$2$};
\node [above,align=center] at (2.2+2.4+2.4+2.4,0) {$2$};
\end{tikzpicture}\\
&
\begin{tikzpicture}
\node [align=center,align=center] at (-0.6,0) {$+2$};
\draw (0,0) circle [radius=0.2];
\node [align=center,align=center] at (0,0) {$0$};
\draw (1,0) circle [radius=0.2];
\node [align=center,align=center] at (1,0) {$0$};
\draw (0.1,0.16) .. controls (0.3,0.4) and (0.7,0.4) ..  (0.9,0.16);
\node [align=center,align=center] at (0.35,0) {$1$};
\node [align=center,align=center] at (0.65,0) {$1$};
\draw (0.2,0) -- (0.8,0);
\node [above,align=center] at (0.1,0.16) {$1$};
\node [above,align=center] at (0.9,0.16) {$1$};
\draw (0.1,-0.16) .. controls (0.3,-0.4) and (0.7,-0.4) ..  (0.9,-0.16);
\node [below,align=center] at (0.1,-0.16) {$2$};
\node [below,align=center] at (0.9,-0.16) {$2$};
\node [align=center,align=center] at (-0.6+2.2,0) {$+2$};
\draw (0+2.2,0) circle [radius=0.2];
\node [align=center,align=center] at (0+2.2,0) {$0$};
\draw (1+2.2,0) circle [radius=0.2];
\node [align=center,align=center] at (1+2.2,0) {$0$};
\draw (0.1+2.2,0.16) .. controls (0.3+2.2,0.4) and (0.7+2.2,0.4) ..  (0.9+2.2,0.16);
\node [align=center,align=center] at (0.35+2.2,0) {$2$};
\node [align=center,align=center] at (0.65+2.2,0) {$2$};
\draw (0.2+2.2,0) -- (0.8+2.2,0);
\node [above,align=center] at (0.1+2.2,0.16) {$1$};
\node [above,align=center] at (0.9+2.2,0.16) {$1$};
\draw (0.1+2.2,-0.16) .. controls (0.3+2.2,-0.4) and (0.7+2.2,-0.4) ..  (0.9+2.2,-0.16);
\node [below,align=center] at (0.1+2.2,-0.16) {$2$};
\node [below,align=center] at (0.9+2.2,-0.16) {$2$};
\node [align=center,align=center] at (-0.6+2.2+2.2,0) {$+4$};
\draw (0+2.2+2.2,0) circle [radius=0.2];
\node [align=center,align=center] at (0+2.2+2.2,0) {$0$};
\draw (1+2.2+2.2,0) circle [radius=0.2];
\node [align=center,align=center] at (1+2.2+2.2,0) {$0$};
\draw (0.1+2.2+2.2,0.16) .. controls (0.3+2.2+2.2,0.4) and (0.7+2.2+2.2,0.4) ..  (0.9+2.2+2.2,0.16);
\node [align=center,align=center] at (0.35+2.2+2.2,0) {$1$};
\node [align=center,align=center] at (0.65+2.2+2.2,0) {$2$};
\draw (0.2+2.2+2.2,0) -- (0.8+2.2+2.2,0);
\node [above,align=center] at (0.1+2.2+2.2,0.16) {$1$};
\node [above,align=center] at (0.9+2.2+2.2,0.16) {$1$};
\draw (0.1+2.2+2.2,-0.16) .. controls (0.3+2.2+2.2,-0.4) and (0.7+2.2+2.2,-0.4) ..  (0.9+2.2+2.2,-0.16);
\node [below,align=center] at (0.1+2.2+2.2,-0.16) {$2$};
\node [below,align=center] at (0.9+2.2+2.2,-0.16) {$2$};
\end{tikzpicture}.
\end{split}
\end{flalign*}
\end{Example}

\subsection{Dual abstract quantum field theory}

Now we can generalize the dual abstract QFT to the case
of labelled dotted stable graphs.
The dual abstract free energy of genus $g$ ($g\geq 2$) is defined to be:
\ben
\wcF_g^\vee:=\sum_{\Gamma^\vee\in\cG_g^{\vee,c}(N)}
\frac{1}{|\Aut(\Gamma^\vee)|}\Gamma^\vee,
\een
where $\cG_g^{\vee,c}(N)$ is the set of all connected
dotted stable graphs of genus $g$ without external edges,
whose half-edges are labelled by indices in $\{1,2,\cdots,N\}$.
Also, we define the dual abstract $n$-point function of genus $g$ to be:
\ben
\wcF^\vee_{g;l_1,\cdots,l_N}:=
\sum_{\Gamma^\vee\in \cG_{g;l_1,\cdots,l_N}^{\vee,c}(N)}
\frac{1}{|\Aut(\Gamma^\vee)|}\Gamma^\vee,
\een
where $\cG_{g;l_1,\cdots,l_N}^{\vee,c}(N)$ is the set of all connected labelled
dotted stable graphs of genus $g$,
with $l_j$ external edges
labelled by $j$ for every $j\in\{1,\cdots,N\}$
(we have $l_1+\cdots+l_N=n$).

Similar to the one-dimensional case,
the edge-cutting operators $K_{ij}$
and the edge-adding operators $\cD_i=\pd_i+\gamma_i$
discussed in \S \ref{sec-pre-label} can be generalized to
the case of dual abstract QFT.
Let $K_{ij}^\vee$ be the edge-cutting operator which cuts an
dotted internal edge labelled by $i$ and $j$,
and $\pd_i^\vee$, $\cD_i^\vee$ be the straightforward generalizations
of $\pd^\vee$, $\cD^\vee$ respectively
to the labelled case using the same way in \S \ref{sec-pre-label}.
Then we have the following recursion relations:

\begin{Lemma}
For every $2g-2+\sum\limits_{j=1}^N l_j>0$,
we have
\ben
\cD_j^\vee\wcF^\vee_{g;l_1,\cdots,l_N}=
\wcF^\vee_{g;l_1,\cdots,l_j+1,\cdots,l_N}.
\een
\end{Lemma}

\begin{Theorem}
For $2g-2+\sum\limits_{j=1}^N l_j>0$, we have
\begin{equation*}
\begin{split}
&K_{ij}^\vee\widehat{\cF}^\vee_{g;l_1,\cdots,l_N}=
\cD_i^\vee\cD_j^\vee\widehat{\cF}^\vee_{g-1;l_1,\cdots,l_N}+
\sum_{\substack{g_1+g_2=g\\p_k+q_k=l_k}}
\cD_i^\vee\widehat{\cF}^\vee_{g_1;p_1,\cdots,p_N}
\cD_j^\vee\widehat{\cF}^\vee_{g_2;q_1,\cdots,q_N},
i\not=j;\\
&K_{ii}^\vee\widehat{\cF}^\vee_{g;l_1,\cdots,l_N}=\half\biggl(
\cD_i^\vee\cD_i^\vee\widehat{\cF}^\vee_{g-1;l_1,\cdots,l_N}+
\sum_{\substack{g_1+g_2=g\\p_k+q_k=l_k}}
\cD_i^\vee\widehat{\cF}^\vee_{g_1;p_1,\cdots,p_N}
\cD_i^\vee\widehat{\cF}^\vee_{g_2;q_1,\cdots,q_N}\biggr).
\end{split}
\end{equation*}
In particular,
for $l_1=\cdots=l_N=0$,
we have the following recursions for the dual free energy
of genus $g$ ($g\geq 2$):
\be
\begin{split}
&K_{ij}^\vee\widehat{\cF}^\vee_g=\cD_i^\vee\partial_j^\vee\widehat{\cF}^\vee_{g-1}+
\sum_{r=1}^{g-1}\partial_i^\vee\widehat{\cF}^\vee_r\partial_j^\vee\widehat{\cF}^\vee_{g-r},
\quad i\not=j;\\
&K_{ii}^\vee\widehat{\cF}^\vee_g=\frac{1}{2}\biggl(\cD_i^\vee\partial_i^\vee\widehat{\cF}^\vee_{g-1}
+\sum_{r=1}^{g-1}\partial_i^\vee\widehat{\cF}^\vee_r\partial_i^\vee\widehat{\cF}^\vee_{g-r}\biggr).
\end{split}
\ee
Here we use the convention
\ben
&&\pd_j^\vee \wcF_{1;l_1,\cdots,l_N}^\vee:=\wcF_{1;l_1,\cdots,l_j+1,\cdots,l_N}^\vee,
\quad \text{for } l_1=\cdots=l_N=0;\\
&&\cD_j^\vee \wcF_{0;l_1,\cdots,l_N}^\vee:= (l_j+1)\wcF_{0;l_1,\cdots,l_j+1,\cdots,l_N}^\vee,
\quad \text{for } l_1+\cdots+l_N=2;\\
&&\cD_i^\vee\cD_j^\vee \wcF_{0;l_1,\cdots,l_N}^\vee:
=(l_j+1)\cD_i^\vee \wcF_{0;l_1,\cdots,l_j+1,\cdots,l_N}^\vee,
\quad \text{for } l_1+\cdots+l_N=1.\\
\een
\end{Theorem}

\subsection{Duality theorem in higher dimension}

In last subsection we have generalized the dual abstract QFT,
dual edge-cutting operator, and dual edge-adding operators
to higher-dimensional case,
i.e., the case of dotted stable graphs with labels on half-edges.
Now let us generalize the duality theorem.

First, similar to Lemma \ref{lem-gamma} and Theorem \ref{thm1},
we have the following relations between ordinary edge-adding operators
$\cD_i=\pd_i+\gamma_i$
and dual edge-adding operators $\cD_i^\vee=\pd_i^\vee+\gamma_i^\vee$:

\begin{Lemma}
We have $\gamma_j=-\gamma^\vee_j$ for every $j\in\{1,2,\cdots,N\}$.
\end{Lemma}

\begin{Theorem}
We have
\be\label{thm1-eq}
\cD^\vee_j=\pd_j,\qquad\qquad \pd_j^\vee=\cD_j
\ee
for every $j\in\{1,2,\cdots,N\}$.
\end{Theorem}

Furthermore,
the duality theorem for the abstract free energies and
abstract $n$-point functions (see Theorem \ref{thm-dual-FE})
can also be generalized to the higher-dimensional case as follows:

\begin{Theorem}\label{N-duality}
For $2g-2+\sum\limits_{j=1}^N l_j>0$, we have
\be
\wcF^\vee_{g;l_1,\cdots,l_N}=\frac{1}{l_1!\cdots l_N!}
\cdot V_{g;l_1,\cdots,l_N},
\ee
where $V_{g;l_1,\cdots,l_N}$ is a labelled ordinary stable vertex of genus $g$,
with $l_j$ external edges labelled by $j$ for every $j\in\{1,2,\cdots,N\}$.
In particular, for $l_1=l_2=\cdots=l_N=0$,
we have
\ben
\begin{tikzpicture}
\node [align=center,align=center] at (-1.8,0) {$\wcF_g^\vee=\wcF_{g;0,\cdots,0}^\vee=$};
\node [align=center,align=center] at (0,0) {$g$};
\draw (0,0) circle [radius=0.2];
\end{tikzpicture}
\een
for every $g\geq 2$.
\end{Theorem}

We omit the proof of this theorem
since it is similar to the one-dimensional case.

\subsection{Generalization of the induced realizations}

In this subsection we describe the higher-dimensional generalization
of the induced realizations in \S \ref{sec:realization}.

Fix a sequence of smooth functions
$\{F_{g;l_1,\cdots,l_N}(t,\kappa)\}_{2g-2+\sum l_j>0}$
and a symmetric matrix $\kappa=(\kappa_{ij})_{1\leq i,j\leq N}$
whose entries $\kappa_{ij}$ are some formal variables.
Then we obtain a realization of the abstract QFT following the construction
in \S\ref{sec-pre-realization}.
This construction induces a natural realization of the dual abstract QFT by assigning
the weights of dotted vertices to be
\ben
\tF_{g;l_1,\cdots,l_N}=l_1!\cdots l_N!\cdot\wF_{g;l_1,\cdots,l_N},
\een
and the weights of dotted internal edges
\ben
\omega_{e^\vee}=-\kappa_{ij}.
\een

Conversely,
if we are given a realization
$\{F^\vee_{g;l_1,\cdots,l_N}(t,\kappa)\}_{2g-2+\sum l_j>0}$
and $\kappa^\vee=(\kappa_{ij}^\vee)_{1\leq i,j\leq N}$
of the dual abstract QFT,
we can construct a realization of the abstract QFT by assigning
the weights of ordinary vertices to be
\ben
\tF^\vee_{g;l_1,\cdots,l_N}=l_1!\cdots l_N!\cdot\wF^\vee_{g;l_1,\cdots,l_N},
\een
and the weights of dotted internal edges
\ben
\omega_e=-\kappa_{ij}^\vee,
\een
where
\be
\wF^\vee_{g;l_1,\cdots,l_N}=
\sum_{\Gamma^\vee\in \cG_{g;l_1,\cdots,l_N}^{\vee,c}(N)}
\frac{1}{|\Aut(\Gamma^\vee)|}\omega_{\Gamma^\vee}
\ee
is the realization of the dual abstract $n$-point functions.

Then the duality theorem relates the two induced realizations above
in the following way:

\begin{Theorem} \label{thm-duality-Nrealization}
Let $\{F_{g;l_1,\cdots,l_N}(t,\kappa)\}_{2g-2+\sum l_j>0}$ and
$\kappa=(\kappa_{ij})_{1\leq i,j\leq N}$
be a realization of the abstract QFT,
and let $\{F^\vee_{g;l_1,\cdots,l_N}(t,\kappa)\}_{2g-2+\sum l_j>0}$
and $\kappa^\vee=(\kappa_{ij}^\vee)_{1\leq i,j\leq N}$
be the induced realization of the dual abstract QFT,
i.e.,
\ben
F^\vee_{g;l_1,\cdots,l_N}(t,\kappa):=\tF_{g;l_1,\cdots,l_N}(t,\kappa),
\qquad
\kappa_{ij}^\vee=-\kappa_{ij}.
\een
Let $\{\tF^\vee_{g;l_1,\cdots,l_N}(t,\kappa)\}_{2g-2+\sum l_j>0}$
and $-\kappa^\vee=\kappa$
be the realization of the abstract QFT induced by
$\{F^\vee_{g;l_1,\cdots,l_N}(t,\kappa)\}_{2g-2+\sum l_j>0}$
and $\kappa^\vee=(\kappa_{ij}^\vee)$,
then we have:
\ben
\tF^\vee_{g;l_1,\cdots,l_N}=
F_{g,l_1,\cdots,l_N}
\een
for every $2g-2+\sum\limits_{j=1}^N l_j>0$.
In particular,
for $l_1=\cdots=l_N=0$, we have
\ben
\tF^\vee_{g}=F_{g;0,\cdots,0}
\een
for every $g\geq 2$.
\end{Theorem}

\subsection{Representations by formal Gaussian integrals}

In this subsection we generalize the construction in \S \ref{subsec-transf}.

Recall that we have
\begin{equation*}
\begin{split}
&(2\pi\lambda^2)^{\frac{N}{2}}\cdot\biggl(
\sum_{g=2}^\infty \lambda^{2g-2}\wF_g
\biggr)
\\
=&\frac{1}{\sqrt{\det (\kappa)}}
\int \exp\biggl(\sum_{2g-2+\sum\limits_{j=1}^N l_j>0}
\frac{\lambda^{2g-2}}{l_1!\cdots l_N!}F_{g;l_1,\cdots,l_N}
\cdot\prod_{j=1}^{N}\eta_{j}^{l_j}
-\frac{\lambda^{-2}}{2}
\eta^T\kappa^{-1}\eta \biggr)d\eta,
\end{split}
\end{equation*}
and its inverse:
\begin{equation*}
\begin{split}
&(2\pi\lambda^2)^{\frac{N}{2}}\cdot\biggl(
\sum_{g=2}^\infty \lambda^{2g-2} F_{g;0,\cdots,0}
\biggr)
\\
=&\frac{1}{\sqrt{\det (-\kappa)}}
\int \exp\biggl(\sum_{2g-2+\sum\limits_{j=1}^N l_j>0}
\frac{\lambda^{2g-2}}{l_1!\cdots l_N!}\tF_{g;l_1,\cdots,l_N}
\cdot\prod_{j=1}^{N}\eta_{j}^{l_j}
+\frac{\lambda^{-2}}{2}
\eta^T\kappa^{-1}\eta \biggr)d\eta,
\end{split}
\end{equation*}
where $\tF_{g;l_1,\cdots,l_N}:=l_1!\cdots l_N!\cdot \wF_{g;l_1,\cdots,l_N}$.

Let us first generalize the two formal Gaussian integrals above such that
the $n$-point functions $\wF_{g;l_1,\cdots,l_N}$
and the dual $n$-point functions $\wF^\vee_{g;l_1,\cdots,l_N}$
are included in  the left-hand-sides.
Similar to Theorem\ref{thm-transf-1} and Theorem\ref{thm-transf-2},
we have the following:

\begin{Theorem}
Let $z_1,\cdots,z_N$ be a sequence of formal variables.
Then we have:
\begin{equation*}
\begin{split}
&(2\pi\lambda^2\det(\kappa))^{\frac{N}{2}}\cdot\biggl(
\sum_{2g-2+\sum\limits_{j=1}^N l_j>0}\
\frac{1}{l_1!\cdots l_N!} \lambda^{2g-2}z_1^{l_1}\cdots z_N^{l_N}\cdot
\tF_{g;l_1,\cdots,l_N}
\biggr)
\\
=&\int \exp\biggl(\sum_{2g-2+\sum\limits_{j=1}^N l_j>0}
\frac{\lambda^{2g-2}}{l_1!\cdots l_N!}F_{g;l_1,\cdots,l_N}
\cdot\prod_{j=1}^{N}\eta_{j}^{l_j}
-\frac{\lambda^{-2}}{2}
(\eta-z)^T\kappa^{-1}(\eta-z)\biggr)d\eta,
\end{split}
\end{equation*}
and its dual version:
\begin{equation*}
\begin{split}
&(2\pi\lambda^2\det(-\kappa))^{\frac{N}{2}}\cdot\biggl(
\sum_{2g-2+\sum\limits_{j=1}^N l_j>0}
\frac{1}{l_1!\cdots l_N!} \lambda^{2g-2} \eta_1^{l_1}\cdots \eta_N^{l_N}\cdot
F_{g;l_1,\cdots,l_N}
\biggr)
\\
=&\int \exp\biggl(\sum_{2g-2+\sum\limits_{j=1}^N l_j>0}
\frac{\lambda^{2g-2}}{l_1!\cdots l_N!}\tF_{g;l_1,\cdots,l_N}
\cdot\prod_{j=1}^{N}z_{j}^{l_j}
+\frac{\lambda^{-2}}{2}
(\eta-z)^T\kappa^{-1}(\eta-z)\biggr)dz.
\end{split}
\end{equation*}
\end{Theorem}

Let us rewrite the above theorem in a way similar to Theorem \ref{thm-transf-3}.
Let $\mathcal W_N$ be the following infinite-dimensional vector space
of formal power series of the form:
\be
w(\eta)=
 \sum_{2g-2+\sum\limits_{j=1}^N l_j>0}\frac{1}{l_1!\cdots l_N!}
\lambda^{2g-2}\eta_1^{l_1}\cdots \eta_N^{l_N}\cdot F_{g;l_1,\cdots,l_N}.
\ee
Define a transformation $\cS$ on the linear space $\cW$ using
the two formal Gaussian integrals \eqref{eq-transf-3} and \eqref{eq-transf-4}:
\begin{equation*}
\begin{split}
\mathcal S^{(N)}_{\kappa}:  &\mathcal W_N\to {\mathcal W}_N,\\
& w(\eta)\mapsto
\log\biggl\{
\frac{1}{(2\pi\lambda^{2}\det(\kappa))^{\frac{N}{2}}}
\int \exp\biggl(
w(z)
-\frac{\lambda^{-2}}{2}
(\eta-z)^T\kappa^{-1}(\eta-z)\biggr)d\eta\biggr\}
\end{split}
\end{equation*}
for $\det(\kappa)\not=0$.
For the special case $\det(\kappa)=0$,
we take \eqref{realization-npt-real-N} as the definition
of the transformation $\cS^{(N)}$.

Then similar to Theorem \ref{thm-transf-3}, we have:
\begin{Theorem}
We have
${\mathcal S}_{-\kappa}^{(N)}\circ {\mathcal S}^{(N)}_{\kappa}=id_{\mathcal W_N}$,
and
${\mathcal S}_{\kappa}^{(N)} \circ {\mathcal S}_{-\kappa}^{(N)}=id_{{\mathcal W}_N}$.
\end{Theorem}

\subsection{The independence assumption}

In this subsection we generalize the Independence Assumption
introduced in \S \ref{sec:Indep} to the higher-dimensional case.

Given a realization $\{F^\vee_{g;l_1,\cdots,l_N}(t,\kappa),-\kappa_{ij}(t)\}$
of the dual abstract QFT,
we can construct a realization $\{\tF^\vee_{g;l_1,\cdots,l_N}(t,\kappa),\kappa_{ij}(t)\}$
of the abstract QFT by applying Theorem \ref{thm-duality-Nrealization}.
Similar to the case of dimension one,
a priori we do not know if the weights of the ordinary stable vertices
$\tF^\vee_{g;l_1,\cdots,l_N}(t,\kappa)$ are independent of $\kappa_{ij}$.
If $\tF^\vee_{g;l_1,\cdots,l_N}(t,\kappa)$ is independent of $\kappa$
for every $2g-2+\sum\limits_{j=1}^N l_j>0$,
then the edge-cutting operator $K_{ij}$ can be realized by
the partial derivative $\frac{\pd}{\pd\kappa_{ij}}$.

\begin{Definition}
We say that a realization $\{F^\vee_{g;l_1,\cdots,l_N}(t,\kappa),-\kappa_{ij}(t)\}$ of
the dual abstract QFT satisfies the {\em Independence Assumption},
if the functions
$\{\tF^\vee_{g,l_1,\cdots,l_N}\}$
are independent of $\kappa$
for every $2g-2+\sum\limits_{j=1}^N l_j>0$.
\end{Definition}

We have the following quadratic recursion generalizing Theorem \ref{thm-realization-HAE}:

\begin{Theorem}
Let $\{\tF_{g;l_1,\cdots,l_N}(t,\kappa),-\kappa_{ij}(t)\}$ be a realization  of
the dual abstract QFT which satisfies the Independence Assumption,
and let
\ben
&& \tD_j\tF_{g;l_1,\cdots,l_N}=\tF_{g;l_1,\cdots,l_j+1,\cdots,l_N}, \\
&& \tD_i\tD_j\tF_{g-1;l_1,\cdots,l_N}
= \tF_{g;l_1,\cdots, l_i+1, \cdots,l_j+1,\cdots,l_N}.
\een
Then for $2g-2+\sum\limits_{j=1}^N l_j>0$, we have
\begin{equation*}
\begin{split}
&\pd_{\kappa_{ij}}\widehat{F}_{g;l_1,\cdots,l_N}=\tD_i\tD_j\widehat{F}_{g-1;l_1,\cdots,l_N}+
\sum_{\substack{g_1+g_2=g\\p_k+q_k=l_k}}
\tD_i\widehat{F}_{g_1;p_1,\cdots,p_N}\tD_j\widehat{F}_{g_2;q_1,\cdots,q_N},
\quad i\not=j;\\
&\pd_{\kappa_{ii}}\widehat{F}_{g;l_1,\cdots,l_N}=\half\biggl(
\tD_i\tD_i\widehat{F}_{g-1;l_1,\cdots,l_N}+
\sum_{\substack{g_1+g_2=g\\p_k+q_k=l_k}}
\tD_i\widehat{F}_{g_1;p_1,\cdots,p_N}\tD_i\widehat{F}_{g_2;q_1,\cdots,q_N}\biggr).
\end{split}
\end{equation*}
where we denote $\wF_{g,n}:=\frac{1}{l_1!\cdots l_N!}\cdot\tF_{g,n}$.
In particular, by taking $l_1=\cdots=l_N=0$, we get
\ben
&&\pd_{\kappa_{ij}}\widehat{F}_g=\tD_i\tD_j\widehat{F}_{g-1}+
\sum_{r=1}^{g-1}\tD_i\widehat{F}_r\cdot\tD_j\widehat{F}_{g-r},
\quad i\not=j;\\
&&\pd_{\kappa_{ii}}\widehat{F}_g=\frac{1}{2}\biggl(\tD_i\tD_i\widehat{F}_{g-1}
+\sum_{r=1}^{g-1}\tD_i\widehat{F}_r\cdot\tD_i\widehat{F}_{g-r}\biggr).
\een
\end{Theorem}

\section{Fourier-Like Transforms of Stable Graphs and
Transformations on Space of Field Theories}

\label{sec:Fourier}

In this section,
we generalize the duality theory of stable graphs to
construct a family of Fourier-like transforms of stable graphs.
When we consider the realizations of the abstract QFT,
these Fourier-like transforms of stable graphs give us
a family of transformations on the space of field theories,
which can be represented by a family of Fourier-like transformations
on formal Gaussian integrals.
We also discuss the group structure on the space of
all such transformations.

\subsection{Type-$\epsilon$ stable graphs and Fourier-like transforms}
\label{sec-epsilon}

Similar to the case of dotted stable vertices and
dotted stable graphs in the dual abstract QFT,
in this subsection we define stable graphs of Type-$\epsilon$.

Let $\cA$ be an arbitrary commutative algebra over $\bQ$.
For example,
we may simply take $\cA=\bC$,
or the ring of smooth functions or formal power series of some formal variable $t$.
In what follows we will work on the linear space $\cA \otimes\cG_{g,n}$,
i.e., the space spanned by stable graphs with coefficients in $\cA$.

First let us define stable vertices of Type-$\epsilon$.
Fix an element $\epsilon\in \cA$.
First,
we formulate a new type of linear combination of stable graphs.
Since they will also be graphically represented by stable vertices
(similar to the dotted stable vertices),
but with an indicator $\epsilon$ to distinguish them from the ordinary stable vertices,
they will be called `Type-$\epsilon$ stable vertices'.
Denote $V_{g,n}^{\epsilon}$ be a Type-$\epsilon$ stable vertex
of genus $g$ and valence $n$,
then we assign the following linear combination of
ordinary stable graphs to $V_{g,n}^\epsilon$:
\be\label{eq-type-vertex}
V_{g,n}^\epsilon:=n!\cdot\sum_{\Gamma\in \cG_{g,n}^{c}}
\frac{\epsilon^{|E(\Gamma)|}}{|\Aut(\Gamma)|}\Gamma.
\ee

\begin{Remark}
It is clear that the dotted stable vertices defined in \S \ref{sec:Dotted}
are stable vertices of Type-$1$.
\end{Remark}

\begin{Example}
Let us present some examples of Type-$\epsilon$ stable vertices:

\begin{flalign*}
\begin{tikzpicture}
\node [above,align=center] at (0,0.1) {\LARGE{$\epsilon$}};
\draw (0,0) circle [radius=0.2];
\node [align=center,align=center] at (0,0) {$0$};
\draw (-0.5,0)--(-0.2,0);
\draw (0.16,0.1)--(0.5,0.15);
\draw (0.16,-0.1)--(0.5,-0.15);
\node [align=center,align=center] at (0.8,0) {$=$};
\draw (0+1.6,0) circle [radius=0.2];
\node [align=center,align=center] at (0+1.6,0) {$0$};
\draw (-0.5+1.6,0)--(-0.2+1.6,0);
\draw (0.16+1.6,0.1)--(0.5+1.6,0.15);
\draw (0.16+1.6,-0.1)--(0.5+1.6,-0.15);
\end{tikzpicture},&&
\end{flalign*}

\begin{flalign*}
\begin{tikzpicture}
\node [above,align=center] at (0,0.1) {\LARGE{$\epsilon$}};
\draw (1.6-1.6,0) circle [radius=0.2];
\draw (1.1-1.6,0.15)--(1.44-1.6,0.1);
\draw (1.1-1.6,-0.15)--(1.44-1.6,-0.1);
\draw (1.76-1.6,0.1)--(2.1-1.6,0.15);
\draw (1.76-1.6,-0.1)--(2.1-1.6,-0.15);
\node [align=center,align=center] at (1.6-1.6,0) {$0$};
\node [align=center,align=center] at (0.8,0) {$=$};
\draw (1.6+0.1,0) circle [radius=0.2];
\draw (1.1+0.1,0.15)--(1.44+0.1,0.1);
\draw (1.1+0.1,-0.15)--(1.44+0.1,-0.1);
\draw (1.76+0.1,0.1)--(2.1+0.1,0.15);
\draw (1.76+0.1,-0.1)--(2.1+0.1,-0.15);
\node [align=center,align=center] at (1.6+0.1,0) {$0$};
\node [align=center,align=center] at (2.7,0) {$+3\epsilon$};
\draw (1.6+1.9+0.2,0) circle [radius=0.2];
\draw (2.2+1.9+0.2,0) circle [radius=0.2];
\draw (1.1+1.9+0.2,0.15)--(1.44+1.9+0.2,0.1);
\draw (1.1+1.9+0.2,-0.15)--(1.44+1.9+0.2,-0.1);
\draw (1.1+1.9+0.2,0.15)--(1.44+1.9+0.2,0.1);
\draw (2.36+1.9+0.2,0.1)--(2.7+1.9+0.2,0.15);
\draw (2.36+1.9+0.2,-0.1)--(2.7+1.9+0.2,-0.15);
\draw (1.8+1.9+0.2,0)--(2+1.9+0.2,0);
\node [align=center,align=center] at (1.6+1.9+0.2,0) {$0$};
\node [align=center,align=center] at (2.2+1.9+0.2,0) {$0$};
\end{tikzpicture},&&
\end{flalign*}

\begin{flalign*}
\begin{tikzpicture}
\node [above,align=center] at (-0.6,0.1) {\LARGE{$\epsilon$}};
\draw (1-1.6,0) circle [radius=0.2];
\draw (1.2-1.6,0)--(1.5-1.6,0);
\node [align=center,align=center] at (1-1.6,0) {$1$};
\node [align=center,align=center] at (0.3,0) {$=$};
\draw (1,0) circle [radius=0.2];
\draw (1.2,0)--(1.5,0);
\node [align=center,align=center] at (1,0) {$1$};
\node [align=center,align=center] at (2.1,0) {$+\frac{1}{2}\epsilon$};
\draw (1+1.8+0.2,0) circle [radius=0.2];
\draw (1.2+1.8+0.2,0)--(1.5+1.8+0.2,0);
\draw (0.84+1.8+0.2,0.1) .. controls (0.5+1.8+0.2,0.2) and (0.5+1.8+0.2,-0.2) ..  (0.84+1.8+0.2,-0.1);
\node [align=center,align=center] at (1+1.8+0.2,0) {$0$};
\end{tikzpicture}.&&
\end{flalign*}
\end{Example}

Now let us define stable graphs of Type-$\epsilon$.
By a `Type-$\epsilon$ stable graph',
we mean a linear combination of ordinary stable graphs,
whose expression can be obtained by gluing Type-$\epsilon$ stable vertices together.
Their definition is similar to the definition of dotted stable graphs
introduced in \S\ref{sec:Dotted},
the only difference is that in the final step (i.e., `gluing'),
we do NOT multiply the factor $(-1)^{|E^\vee(\Gamma^\vee)|}$.

\begin{Remark}
\label{rmk-dotted-type-1}
Let $\Gamma$ be a stable graph,
and $\Gamma^\vee$
(resp. $\Gamma^1$)
be the dotted stable graph (resp. Type-$1$ stable graph)
obtained by directly changing the vertices and edges of $\Gamma$ into
dotted (resp. Type-$1$) vertices and edges.
Then we have:
\ben
\Gamma^\vee=(-1)^{|E(\Gamma)|}\cdot \Gamma^{1}.
\een
\end{Remark}

\begin{Definition}
Let $\Gamma$ be a stable graph,
and denote by $\Gamma^\epsilon$ the stable graph of Type-$\epsilon$
obtained by directly changing the vertices and edges of $\Gamma$
into vertices and edges of Type-$\epsilon$.
Then we call the transformation
\be
\Phi_\epsilon :\quad \cV\to \cV,\quad
\Gamma\mapsto\Gamma^\epsilon
\ee
a Fourier-like transform of stable graphs.
In particular,
\be\label{eq-vertex-fourier}
\Phi_\epsilon(V_{g,n})=V_{g,n}^\epsilon,
\ee
where $V_{g,n}$ a stable vertex of genus $g$ and valence $n$,
and $V_{g,n}^\epsilon$ the stable vertex of Type-$\epsilon$, of genus $g$ and valence $n$.
\end{Definition}

\subsection{Type-$(\epsilon_1,\epsilon_2)$ stable graphs}

Now fix two elements $\epsilon_1,\epsilon_2\in \cA$,
in this subsection
let us define stable vertices of Type-$(\epsilon_1,\epsilon_2)$.

Recall that Type-$\epsilon$ stable vertices are obtained by
gluing ordinary stable vertices together using Type-$\epsilon$ edges
(see \eqref{eq-type-vertex}).
Similarly,
roughly speaking, stable vertices of Type-$(\epsilon_1,\epsilon_2)$
is obtained by gluing Type-$\epsilon_1$ stable vertices together
using Type-$\epsilon_2$ edges.
Or more precisely,
\be\label{eq-type-vertex-2}
V_{g,n}^{(\epsilon_1,\epsilon_2)}:=
\sum_{\Gamma^{\epsilon_1}\in\cG_{g,n}^{\epsilon_1,c}}
\frac{\epsilon_2^{|E^{\epsilon_1}(\Gamma^{\epsilon_1})|}}{|\Aut(\Gamma^{\epsilon_1})|}
\Gamma^{\epsilon_1},
\ee
where $V_{g,n}^{(\epsilon_1,\epsilon_2)}$ is a stable vertex
of Type-$(\epsilon_1,\epsilon_2)$, of genus $g$ and valence $n$;
and $\cG_{g,n}^{\epsilon_1,c}$ is the set of all connected Type-$\epsilon_1$ stable graphs
of genus $g$ with $n$ external edges,
and $E^{\epsilon_1}(\Gamma^{\epsilon_1})$ is the set of
Type-$\epsilon_1$ internal edges of $\Gamma^{\epsilon_1}$.

Similarly,
we can glue stable vertices of Type-$(\epsilon_1,\epsilon_2)$
to obtain stable graphs of Type-$(\epsilon_1,\epsilon_2)$.
The transformation
\ben
\Gamma\quad\mapsto\quad \Gamma^{(\epsilon_1,\epsilon_2)},
\qquad\qquad
V_{g,n}\quad\mapsto\quad
V_{g,n}^{(\epsilon_1,\epsilon_2)},
\een
can be understood as the composition of Fourier-like transformations
$\Phi_{\epsilon_1}$ and $\Phi_{\epsilon_2}$.

\begin{Example}
We give some examples:

\begin{flalign*}
\begin{tikzpicture}
\draw (0,0) circle [radius=0.2];
\node [align=center,align=center] at (0,0) {$0$};
\draw (-0.5,0)--(-0.2,0);
\draw (0.16,0.1)--(0.5,0.15);
\draw (0.16,-0.1)--(0.5,-0.15);
\node [above,align=center] at (0,0.1) {\large{$(\epsilon_1,\epsilon_2)$}};
\node [align=center,align=center] at (0.8,0) {$=$};
\draw (0+1.6,0) circle [radius=0.2];
\node [above,align=center] at (1.6,0.1) {\Large{$\epsilon_1$}};
\node [align=center,align=center] at (0+1.6,0) {$0$};
\draw (-0.5+1.6,0)--(-0.2+1.6,0);
\draw (0.16+1.6,0.1)--(0.5+1.6,0.15);
\draw (0.16+1.6,-0.1)--(0.5+1.6,-0.15);
\node [align=center,align=center] at (0.8+1.6,0) {$=$};
\draw (0+1.6+1.6,0) circle [radius=0.2];
\node [align=center,align=center] at (0+1.6+1.6,0) {$0$};
\draw (-0.5+1.6+1.6,0)--(-0.2+1.6+1.6,0);
\draw (0.16+1.6+1.6,0.1)--(0.5+1.6+1.6,0.15);
\draw (0.16+1.6+1.6,-0.1)--(0.5+1.6+1.6,-0.15);
\end{tikzpicture},&&
\end{flalign*}

\begin{flalign*}
\begin{tikzpicture}
\draw (1-1.6,0) circle [radius=0.2];
\draw (1.2-1.6,0)--(1.5-1.6,0);
\node [align=center,align=center] at (1-1.6,0) {$1$};
\node [align=center,align=center] at (0.3,0) {$=$};
\node [above,align=center] at (1-1.6,0.1) {\large{$(\epsilon_1,\epsilon_2)$}};
\draw (1,0) circle [radius=0.2];
\node [above,align=center] at (1,0.1) {\Large{$\epsilon_1$}};
\draw (1.2,0)--(1.5,0);
\node [align=center,align=center] at (1,0) {$1$};
\node [align=center,align=center] at (2.1,0) {$+\frac{1}{2}\epsilon_2$};
\draw (1+2.1,0) circle [radius=0.2];
\node [above,align=center] at (3.1,0.1) {\Large{$\epsilon_1$}};
\draw (1.2+2.1,0)--(1.5+2.1,0);
\draw (0.84+2.1,0.1) .. controls (0.5+2.1,0.2) and (0.5+2.1,-0.2) ..  (0.84+2.1,-0.1);
\node [align=center,align=center] at (1+2.1,0) {$0$};
\node [align=center,align=center] at (4.1,0) {$=$};
\draw (1+3.9,0) circle [radius=0.2];
\draw (1.2+3.9,0)--(1.5+3.9,0);
\node [align=center,align=center] at (1+3.9,0) {$1$};
\node [align=center,align=center] at (2.1+4.4,0) {$+\frac{1}{2}(\epsilon_1+\epsilon_2)$};
\draw (1+2.1+3.9+1,0) circle [radius=0.2];
\draw (1.2+2.1+3.9+1,0)--(1.5+2.1+3.9+1,0);
\draw (0.84+2.1+3.9+1,0.1) .. controls (0.5+2.1+3.9+1,0.2) and (0.5+2.1+3.9+1,-0.2) ..  (0.84+2.1+3.9+1,-0.1);
\node [align=center,align=center] at (1+2.1+3.9+1,0) {$0$};
\end{tikzpicture}.&&
\end{flalign*}

\end{Example}

\begin{Remark}
Compare Remark \ref{rmk-dotted-type-1}
with the definition of dual abstract free energies and abstract $n$-point functions
(see Definition \ref{def-dual-absfe}),
we know that the dual abstract $n$-point function
equals the stable vertex of Type-$(1,-1)$:
\ben
\wcF_{g,n}^\vee=V_{g,n}^{(1,-1)}.
\een
\end{Remark}

Similar to the case of Type-$(1,-1)$
(see Theorem\ref{thm-dual-FE}),
we want to represent a stable vertex $V_{g,n}^{(\epsilon_1,\epsilon_2)}$
of Type-${(\epsilon_1,\epsilon_2)}$
directly in terms of a linear combination of ordinary stable graphs.
Our main result in this subsection is the following:

\begin{Theorem}
\label{thm-epsilon-abstract}
For $2g-2+n>0$,
we have:
\be\label{eq-epsilon-abstract}
V_{g,n}^{(\epsilon_1,\epsilon_2)}
=\sum_{\Gamma\in\cG_{g,n}^{c}}
\frac{(\epsilon_1+\epsilon_2)^{|E(\Gamma)|}}{|\Aut(\Gamma)|}\Gamma.
\ee
\end{Theorem}
\begin{proof}
All we need to do is to mimic the proof of Theorem \ref{thm-dual-FE}.
In fact,
given a stable graph $\Gamma\in\cG_{g,n}^c$,
let us compute the coefficient of $\Gamma$ in $V_{g,n}^{(\epsilon_1,\epsilon_2)}$
using exactly the same method in the proof of Theorem\ref{thm-dual-FE},
and we will get the following result
(which generalizes the coefficient \eqref{eq-coeff-final}):
\be\label{eq-coeff-add}
\frac{1}{|\Aut(\Gamma)|}\cdot\sum_{E\subset E(\Gamma)}\epsilon_1^{|E|}
\cdot \epsilon_2^{|E(\Gamma)\backslash E|}.
\ee
Therefore we have:
\ben
V_{g,n}^{(\epsilon_1,\epsilon_2)}
&=&\sum_{\Gamma\in\cG_{g,n}^{c}}
\frac{1}{|\Aut(\Gamma)|}\cdot\biggl(\sum_{E\subset E(\Gamma)}\epsilon_1^{|E|}
\cdot \epsilon_2^{|E(\Gamma)\backslash E|}\biggr)
\Gamma\\
&=&\sum_{\Gamma\in\cG_{g,n}^{c}}
\frac{(\epsilon_1+\epsilon_2)^{|E(\Gamma)|}}{|\Aut(\Gamma)|}\Gamma.
\een
\end{proof}

The above theorem gives us the following:
\begin{Corollary}
For every two elements $\epsilon_1,\epsilon_2\in\cA$,
we have:
\be
V_{g,n}^{(\epsilon_1,\epsilon_2)}=
V_{g,n}^{(\epsilon_2,\epsilon_1)}.
\ee
In other words, we have:
\be
\Phi_{\epsilon_1}\circ\Phi_{\epsilon_2}
=\Phi_{\epsilon_2}\circ\Phi_{\epsilon_1}.
\ee
\end{Corollary}

\subsection{Realization and formal Gaussian integrals}

In this subsection we consider the realization of
the constructions in the previous subsections.
We also represent them in terms of formal Gaussian integrals.

First,
we need to fix a realization $(\{F_{g,n}\}_{2g-2+n>0},\kappa)$
of the abstract quantum field theory ($\kappa\not=0$),
where $F_{g,n}$ and $\kappa$ are some smooth functions
(or formal power series) in a formal variable $t$.

Let $\cA$ be the commutative algebra consisting of some smooth functions
(or formal power series) in $t$,
and take two elements $\epsilon_1,\epsilon_2$ in $\cA$.
Then from \eqref{eq-type-vertex} and \eqref{eq-type-vertex-2} we know that
a Type-$\epsilon_1$ stable vertex $V_{g,n}^{\epsilon_1}$ of genus $g$ and valence $n$
is realizaed by:
\ben
U_{g,n}=n!\cdot\sum_{\Gamma\in\cG_{g,n}^c}\bigg(
\frac{(\epsilon_1\kappa)^{|E(\Gamma)|}}{|\Aut(\Gamma)|}\cdot
\prod_{v\in V(\Gamma)}F_{g_v,\val_v}\bigg),
\een
and a Type-$(\epsilon_1,\epsilon_2)$ stable vertex
$V_{g,n}^{(\epsilon_1,\epsilon_2)}$ of genus $g$ and valence $n$
is realizaed by:
\ben
W_{g,n}=n!\cdot\sum_{\Gamma\in\cG_{g,n}^{c}}\bigg(
\frac{(\epsilon_2\kappa)^{|E(\Gamma)|}}{|\Aut(\Gamma)|}\cdot
\prod_{v\in V(\Gamma)}U_{g_v,\val_v}\bigg).
\een

Then Theorem \ref{thm-epsilon-abstract} tells us:

\begin{Theorem}
\label{thm-epsilon-real}
Let $2g-2+n>0$. Then we have:
\be
W_{g,n}=n!\cdot\sum_{\Gamma\in\cG_{g,n}^{c}}\bigg(
\frac{\big((\epsilon_1+\epsilon_2)\kappa\big)^{|E(\Gamma)|}}{|\Aut(\Gamma)|}\cdot
\prod_{v\in V(\Gamma)}F_{g_v,\val_v}\bigg).
\ee
\end{Theorem}

Similar to the case in \S \ref{subsec-transf},
the above constructions can be represented by formal Gaussian integrals.
We already know that
\begin{equation*}
\begin{split}
&\exp\bigg(\sum_{2g-2+n>0}
\lambda^{2g-2}\eta^n\cdot U_{g,n}\bigg)\\
=&
\frac{1}{(2\pi\lambda^{2}\epsilon_1\kappa)^{\frac{1}{2}}}
\int \exp\biggl\{
\sum_{2g-2+n>0}\lambda^{2g-2}F_{g,n}\cdot
\frac{z^n}{n!}-\frac{\lambda^{-2}}{2\epsilon_1\kappa}(z-\eta)^2
\biggr\}dz
\end{split}
\end{equation*}
and
\begin{equation*}
\begin{split}
&\exp\bigg(\sum_{2g-2+n>0}
\lambda^{2g-2}\eta^n\cdot W_{g,n}\bigg)\\
=&
\frac{1}{(2\pi\lambda^{2}\epsilon_2\kappa)^{\frac{1}{2}}}
\int \exp\biggl\{
\sum_{2g-2+n>0}\lambda^{2g-2}U_{g,n}\cdot
\frac{z^n}{n!}-\frac{\lambda^{-2}}{2\epsilon_2\kappa}(z-\eta)^2
\biggr\}dz,
\end{split}
\end{equation*}
and Theorem \ref{thm-epsilon-real} gives us:

\begin{Theorem}
We have:
\begin{equation*}
\begin{split}
&\exp\bigg(\sum_{2g-2+n>0}
\lambda^{2g-2}\eta^n\cdot W_{g,n}\bigg)\\
=&
\frac{1}{(2\pi\lambda^{2}(\epsilon_1+\epsilon_2)\kappa)^{\frac{1}{2}}}
\int \exp\biggl\{
\sum_{2g-2+n>0}\lambda^{2g-2}F_{g,n}\cdot
\frac{z^n}{n!}-\frac{\lambda^{-2}}{2(\epsilon_1+\epsilon_2)\kappa}(z-\eta)^2
\biggr\}dz.
\end{split}
\end{equation*}
\end{Theorem}

\subsection{A group action on the space of field theories}

Recall that in \S \ref{sec:Space} we have defined
the space of field theories $\cW$ together with
a family of Fourier-like transformations
\ben
\cS_\kappa:\quad\cW\to\cW
\een
on this space, and showed that
\ben
\cS_\kappa\circ\cS_{-\kappa}=\cS_{-\kappa}\circ\cS_{\kappa}=\cS_0= id_{\mathcal{W}},
\een
where $\kappa\in\cA$ is a smooth function (or a formal power series) in $t$.
In this subsection we will discuss the composition of two such transformations.

Given two propagators $\kappa_1,\kappa_2\in \cA$,
a priori we do not know that the composition of $\cS_{\kappa_1}$ and $\cS_{\kappa_2}$
is still of the same form,
i.e., $\cS_{\kappa_1}\circ\cS_{\kappa_2}=\cS_{\kappa'}$ for some $\kappa'\in \cA$.
But if this is true for every $\kappa_1,\kappa_2\in\cA$,
then this structure gives us a natural group structure on
the space of all such transformations $\{\cS_\kappa\}_{\kappa\in\cA}$.

Our main theorem in this subsection is the following:

\begin{Theorem}\label{thm-grouplaw}
For arbitrary two propagators $\kappa_1$ and $\kappa_2$,
we have:
\ben
\cS_{\kappa_1}\circ\cS_{\kappa_2}=\cS_{\kappa_2}\circ\cS_{\kappa_1}=\cS_{\kappa_1+\kappa_2}.
\een
I.e., we have $\{\cS_\kappa\}_{\kappa\in \cA}\cong \cA$ as abelian groups.
\end{Theorem}

\begin{proof}
This theorem is a direct corollary of Theorem \ref{thm-epsilon-real}.
All we need to do is to take $\epsilon_1=\frac{\kappa_1}{\kappa}$
and $\epsilon_2=\frac{\kappa_2}{\kappa}$.
\end{proof}

\begin{Example}
Let $\{F_{g,n}\}\in \cW$,
and $\kappa_1,\kappa_2$ be two propagators.
Denote by $\{\tF_{g,n}\}=\cS_{\kappa_1}\big(\{F_{g,n}\}\big)$,
and $\{G_{g,n}\}=\cS_{\kappa_2}\big(\{\tF_{g,n}\}\big)$.
Then the following examples can be checked by direct computations:
\ben
G_{0,3}&=&\tF_{0,3}=F_{0,3};\\
G_{0,4}&=&\tF_{0,4}+3\kappa_2 \tF_{0,3}^2\\
&=&F_{0,4}+3\kappa_1F_{0,3}^2+3\kappa_2 F_{0,3}^2;\\
G_{1,1}&=&\tF_{1,1}+\half\kappa_2\tF_{0,3}\\
&=&F_{1,1}+\half(\kappa_1+\kappa_2)F_{0,3};\\
G_{1,2}&=&\tF_{1,2}+\kappa_2\tF_{1,1}\tF_{0,3}+\half\kappa_2\tF_{0,4}+\kappa_2^2\tF_{0,3}^2\\
&=&F_{1,2}+(\kappa_1+\kappa_2)F_{1,1}F_{0,3}+\half(\kappa_1+\kappa_2)F_{0,4}
+(\kappa_1+\kappa_2)^2F_{0,3}^2,\\
G_2&=&
\tF_{2,0}+\half\kappa_2\tF_{1,1}^2+\half \kappa_2\tF_{1,2}^2
+\half\kappa_2^2\tF_{1,1}\tF_{0,3}+\frac{1}{8}\kappa_2^2\tF_{0,4}+\frac{5}{24}\kappa_2^3\tF_{0,3}^2\\
&=&F_{2,0}+\half(\kappa_1+\kappa_2)F_{1,1}^2+\half (\kappa_1+\kappa_2)F_{1,2}^2
+\half(\kappa_1+\kappa_2)^2F_{1,1}F_{0,3}\\
&&+\frac{1}{8}(\kappa_1+\kappa_2)^2F_{0,4}+\frac{5}{24}(\kappa_1+\kappa_2)^3F_{0,3}^2.
\een
\end{Example}

\begin{Remark}
As in \S \ref{sec:Space},
we can also regard $\cS_{\kappa_1}$ and $\cS_{\kappa_2}$ as acting on the partition functions.
Then they are given by Fourier transforms of the type as in \eqref{eqn:Fourier}:
\ben
&& (\cS_{\kappa_i} Z)(y) = \frac{1}{(2\pi\lambda^{2}\kappa_i)^{\frac{1}{2}}}
\int Z(x) \cdot \exp\biggl(-\frac{\lambda^{-2}}{2\kappa_i}(x-y)^2
\biggr)dx,
\een
for $i=1,2$.
Therefore,
\ben
&& (\cS_{\kappa_2}(\cS_{\kappa_1} Z))(z) \\
& = & \frac{1}{(2\pi\lambda^{2}\kappa_2)^{\frac{1}{2}}}
\int (\cS_{\kappa_2}Z)(y) \cdot \exp\biggl(-\frac{\lambda^{-2}}{2\kappa_2}(y-z)^2  \biggr)dy\\
& = & \frac{1}{2\pi\lambda^{2}(\kappa_1\kappa_2)^{1/2}}
\int\int Z(x) \cdot \exp\biggl(-\frac{\lambda^{-2}}{2\kappa_1}(x-y)^2-\frac{\lambda^{-2}}{2\kappa_2}(y-z)^2  \biggr)dxdy\\
& = & \frac{1}{2\pi\lambda^{2}(\kappa_1\kappa_2)^{1/2}}
\int\int Z(x) \cdot \exp\biggl(-\frac{\lambda^{-2}(x-z)^2}{2(\kappa_1+\kappa_2)} \\
&& -\frac{\lambda^{-2}(\kappa_1+\kappa_2)(y-(\kappa_2 x+\kappa_1z)/(\kappa_1+\kappa_2))^2}{2\kappa_1\kappa_2 }  \biggr)dxdy\\
& = & \frac{1}{(2\pi\lambda^{2}(\kappa_1+\kappa_2))^{1/2} }
\int Z(x) \cdot \exp\biggl(-\frac{\lambda^{-2}}{2(\kappa_1+\kappa_2)}(x-z)^2   \biggr)dx \\
& = & (\cS_{\kappa_1+\kappa_2}Z)(z).
\een
This gives us another proof of Theorem \ref{thm-grouplaw}.
\end{Remark}

\section{Application: the Holomorphic Anomaly Equations}
\label{sec:HAE}

In this section,
we derive the holomorphic anomaly equations
as an application of the duality theory.
We will use the version of these equations in \cite{gkmw}.
In \cite[\S 5]{wz},
we have already interpreted the special case
$\kappa=-\frac{1}{2\sqrt{-1}}\big(Im (\tau)\big)^{-1}$,
where $\tau=\frac{\pd^2 F_0}{\pd t_0^2}$
and $F_0$ is the genus zero free energy.
We will see that our formalism can be applied to understand the general case
$\kappa=-\frac{1}{2\sqrt{-1}}\big(Im (\tau)\big)^{-1}+\cE$
where $\cE$ is a holomorphic function in $t$.

\subsection{Backgrounds about the holomorphic anomaly equations}

The holomorphic anomaly equations were introduced by
Bershadsky, Cecotti, Ooguri, and Vafa \cite{bcov1, bcov2}
in order to compute the the Gromov-Witten invariants of the
quintic Calabi-Yau threefolds.
These authors introduced a sequence of non-holomorphic free energy
$\cF_g(t,\bar{t})$,
and solved them recursively
using the following holomorphic anomaly equations
(cf. \cite[(3.6)]{bcov2}):
\be
\bar\partial_{\bar{i}}\cF_g=
\frac{1}{2}\bar{C}_{\bar{i}\bar{j}\bar{k}}e^{2K}G^{j\bar{j}}G^{k\bar{k}}
(D_j D_k\cF_{g-1}+\sum_{r=1}^{g-1}D_j\cF_r D_k\cF_{g-r}), \quad g\geq 2.
\ee
Then the free energy of genus $g$
can be obtained by
formally taking the following `holomorphic limit':
\ben
F_g(t)=\lim_{\bar{t}\to\infty}\cF_g(t,\bar{t}).
\een
The non-holomorphic free energies $\cF_g(t,\bar{t})$ are supposed
to have modularity,
while the holomorphic free energies $F_g(t)$ are not.
The ring of modular forms are often polynomial algebras with finitely many generators,
and some polynomiality properties of $\cF_g(t,\bar{t})$ are studied by
Yamaguchi-Yau \cite{yy}.

The non-holomorphic free energy are supposed to be closely related to some
geometric structures on the relevant moduli space according to \cite{bcov2}.
The genus zero free energy is known to be related to the special K\"ahler geometry
of the moduli space \cite{cdgp}, and the genus one free energy is
related to the $tt^*$-geometry \cite{cv, dub} and
the theory of analytic torsions \cite{bcov1, cfiv, cv2}.
The higher genus non-holomorphic free energies are supposed to be some
suitable non-holomorphic sections of a holomorphic line bundle
on the moduli space.

In the work \cite{wit},
Witten interpreted the holomorphic anomaly equations from the point of view
of geometric quantization of symplectic vector spaces and
background independence. Inspired by this work,
Aganagic, Bouchard and Klemm \cite{abk} obtained
the following general form for $\wF_g$:
\be
\wF_g(t, \bar{t}) = F_g(t) + \Gamma_g \biggl(\Delta^{IJ},
\pd_{I_1} \cdots \pd_{I_n}F_{r<g}(t)\biggr).
\ee
In  some special cases, the propagator
$\Delta^{IJ}$ is of the form $-(\tau -\bar{\tau})^{IJ}$
where $\tau = (\tau_{ij}) = \big(\frac{\pd^2F_0}{\pd t_i\pd t_j})$.
And these cases were interpreted for matrix models
using Eynard-Orantin topological recursion
by Eynard, Mari\~no and Orantin \cite{emo}.
They also represented the partition function as a
formal Gaussian integral \cite[(4.27)]{emo}
and presented the Feynman graphs and Feynman rules
for the terms that contribute to $\Gamma_g$.
Grimm, Klemm, Mari\~no, and Weiss used the method of direct integration
to solve the holomorphic anomaly equations in \cite{gkmw}.
They reformulated the holomorphic anomaly equation as a
quadratic recursion relation for the derivative of $\wF_g$
with respective to the propagators
$\Delta^{IJ}$ (cf. \cite[(7.50)]{gkmw}:
\be \label{eqn:HA}
\frac{\pd \wF_g}{\pd \Delta^{IJ}} =
\frac{1}{2}
D_I\pd_J\wF_{g-1} + \frac{1}{2}
\sum^{g-1}_{r=1}
\pd_I\wF_r\pd_J\wF_{g-r}.
\ee
In this work, the propagator can be taken to be of the form
\ben
\Delta^{IJ}=-\frac{1}{2\sqrt{-1}}\big(Im(\tau)^{-1}\big)^{-1}+\cE^{IJ},
\een
where $\cE$ is a holomorphic function.

Inspired by the physics literatures \cite{abk, ey, emo, gkmw},
the authors interpreted the special case $\cE=0$ of the holomorphic anomaly equation
as a particular realization of the abstract quantum field theory in \cite[\S 5]{wz}.
In order to obtain the holomorphic anomaly equation,
We need to realize the operators $K$, $\pd$, $\cD$
by the partial derivatives $\pd_\kappa$, $\pd_t$,
and the covariant derivative $D_t=\pd_t+n\kappa F_{0,3}$ respectively.
In this special case $\cE=0$,
the realization of the abstract QFT is given by
the Feynman rule in Example \ref{eg-HAE-special}.
In this case the propagator is $\kappa=\frac{1}{\bar{\tau}-\tau}$,
and we can easily check:
\ben
\frac{\pd}{\pd t}\kappa=\kappa^2\cdot F_{0,3}.
\een
Therefore in this special case the partial derivative $\pd_t$ is indeed a realization of the
edge-adding operator $\pd$.
For details, see \cite[\S 4, \S 5]{wz}.

In the rest of this section,
let us consider the general case
$\kappa=\frac{1}{\bar{\tau}-\tau}+\cE$,
where $\cE$ is a holomorphic function in $t$.
We will see that in this general case the Feynman rule we need
is not the one in Example \ref{eg-HAE-special} anymore.
We need to start from the Feynman rule for the dual abstract QFT.

\subsection{The realization of the dual abstract QFT}
\label{sec:KleZas}

In this subsection we formulate the Feynman rule of the dual abstract QFT
for the holomorphic anomaly equations,
inspired by physics literatures \cite{gkmw, kz}.

Let us first recall some results in these literatures.
In \cite{kz}, Klemm and Zaslow solved the holomorphic anomaly equation
at genus $2$ and $3$ directly. In their notations,
the solutions are given by the following recursive formulas
(cf.\cite[\S 3.5]{kz}):
\begin{flalign} \label{eqn:F2}
\begin{split}
F^{(2)}=&-\frac{1}{8}S_2^2F_{,4}^{(0)}
+\frac{1}{2}S_2F_{,2}^{(1)}
+\frac{5}{24}S_2^3(F_{,3}^{(0)})^2
- \frac{1}{2}S_2^2F_{,1}^{(1)}F_{,3}^{(0)}\\
+&\frac{1}{2}S_2(F_{,1}^{(1)})^2
+f^{(2)},
\end{split}&&
\end{flalign}
\begin{flalign}  \label{eqn:F3}
\begin{split}
F^{(3)}=&S_2F_{,1}^{(2)}F_{,1}^{(1)}
-\frac{1}{2}S_2^2F_{,1}^{(2)}F_{,3}^{(0)}
+\frac{1}{2}S_2F_{,2}^{(2)}
+\frac{1}{6}S_2^3(F_{,1}^{(1)})^3F_{,3}^{(0)} \\
- & \frac{1}{2}S_2^2F_{,2}^{(1)}(F_{,1}^{(1)})^2
-\frac{1}{2}S_2^4(F_{,1}^{(1)})^2(F_{,3}^{(0)})^2
+\frac{1}{4}S_2^3(F_{,1}^{(1)})^2F_{,4}^{(0)} \\
+& S_2^3F_{,2}^{(1)}F_{,1}^{(1)}F_{,3}^{(0)}
-\frac{1}{2}S_2^2F_{,3}^{(1)}F_{,1}^{(1)}
-\frac{1}{4}S_2^2(F_{,2}^{(1)})^2
+\frac{5}{8}S_2^5F_{,1}^{(1)}(F_{,3}^{(0)})^3 \\
-& \frac{2}{3}S_2^4F_{,1}^{(1)}F_{,4}^{(0)}F_{,3}^{(0)}
-\frac{5}{8}S_2^4F_{,2}^{(1)}(F_{,3}^{(0)})^2
+\frac{1}{4}S_2^3F_{,2}^{(1)}F_{,4}^{(0)}\\
+&\frac{5}{12}S_2^3F_{,3}^{(1)}F_{,3}^{(0)}
+  \frac{1}{8}S_2^3F_{,5}^{(0)}F_{,1}^{(1)}
-\frac{1}{8}S_2^2F_{,4}^{(1)}
-\frac{7}{48}S_2^4F_{,5}^{(0)}F_{,3}^{(0)}\\
+&\frac{25}{48}S_2^5F_{,4}^{(0)}(F_{,3}^{(0)})^2
- \frac{5}{16}S_2^6(F_{,3}^{(0)})^4
-\frac{1}{12}S_2^4(F_{,4}^{(0)})^2
+\frac{1}{48}S_2^3F_{,6}^{(0)}
+f^{(3)},
\end{split}&&
\end{flalign}
where $F^{(g)}$ is the non-holomorphic free energy of genus $g$;
$F_{,n}^{(g)}:=(D_z)^n F^{(g)}$ are the $n$-th covariant derivatives
of $F^{(g)}$; $f^{(g)}$ are the holomorphic ambiguity of genus $g$;
and $S_2$ is the propagator.
They also interpreted these two recursions as summations over
stable graphs of genus $2$ and $3$ respectively:
\be\label{eq-kz-fr}
\begin{split}
&F^{(2)}=-\sum_{\Gamma\in\cG_{2,0}^c\backslash\{V_2\}}\frac{1}{|\Aut(\Gamma)|}
\prod_{v\in V(\Gamma)}F_{,\val(v)}^{(g_v)}
\prod_{e\in E(\Gamma)}(-S_2)+f^{(2)},\\
&F^{(3)}=-\sum_{\Gamma\in\cG_{3,0}^c\backslash\{V_3\}}\frac{1}{|\Aut(\Gamma)|}
\prod_{v\in V(\Gamma)}F_{,\val(v)}^{(g_v)}
\prod_{e\in E(\Gamma)}(-S_2)+f^{(3)},\\
\end{split}
\ee
where $V_g$ ($g\geq 2$) is the stable graph consisting of one single vertex
of genus $g$.

Now let us rewrite the above equations in the following way:
we move $f^{(g)}$ to the left-hand-side,
and move $F^{(g)}$ to the right-hand-side;
and we also write these formulas as summations over dotted stable graphs
instead of ordinary stable graphs.
In this way, we obtain:

\ben
&&f^{(2)}=\sum_{\Gamma\in\cG_{2,0}^{\vee,c}}\frac{1}{|\Aut(\Gamma)|}
\prod_{v^\vee\in V^\vee(\Gamma)}F_{,\val(v^\vee)}^{(g_v^\vee)}
\prod_{e^\vee\in E^\vee(\Gamma)}(-S_2),\\
&&f^{(3)}=\sum_{\Gamma\in\cG_{3,0}^{\vee,c}}\frac{1}{|\Aut(\Gamma)|}
\prod_{v^\vee\in V^\vee(\Gamma)}F_{,\val(v^\vee)}^{(g_v^\vee)}
\prod_{e^\vee\in E^\vee(\Gamma)}(-S_2).\\
\een
Now These equations suggest us to interpret $\{F_{,n}^{(g)},-S_2\}$
as a realization of the dual abstract QFT.
This realization induces a realization of the abstract QFT,
and by our duality theorem (see Theorem \ref{thm-dual-realization}),
the weight of an ordinary stable vertex of genus $g$
is the holomorphic ambiguity $f^{(g)}$,
and the propagator is $S_2$.

Now we are able to formulate a mathematical construction for the holomorphic anomaly equations.
First, suppose we have a holomorphic function $\tau=\tau(t)$.
Denote $\tF_{0,3}(t)=F_{0,3}:=\frac{\pd}{\pd t}\tau$,
then $\tF_{0,3}(t)$ is also a holomorphic function in $t$.
In physics literatures,
$\tau$ and $\tF_{0,3}$ play the roles of the periods and Yukawa couplings respectively,
and they should be the second and third derivatives of a prepotential respectively,
see \cite{cdgp, gkmw}.
The imaginary part of $\tau$ gives us the K\"ahler potential of a special K\"ahler structure
of the moduli space \cite{str, fr}.

Define $\kappa$ to be the following function in $t$:
\be \label{eqn:Kappa-E}
\kappa:=-\frac{1}{2\sqrt{-1}}\big(Im(\tau)\big)^{-1}+\cE(t),
\ee
where $\cE(t)$ is a holomorphic function in $t$.
The function $\kappa(t)$ will play the role of propagator
(see \cite[(7.24)]{gkmw}).

Moreover,
we need the following data:
a function $\tF_{1,1}(t,\kappa)$
and a sequence of holomorphic functions $\{F_{g}\}_{g\geq 2}$.
Here $\tF_{1,1}(t,\kappa)$ is of the form
\ben
\tF_{1,1}=\half \tF_{0,3}\cdot\kappa+\pd_t f^{(1)}
\een
where $f^{(1)}$ is a holomorphic function
(see \cite[(7.27)]{gkmw}).
The functions $F_{g}$ ($g\geq 2$) plays the role
of the holomorphic ambiguity $f^{(g)}$.

Now let us construct a realization of the dual abstract QFT as follows.
\begin{itemize}
\item[(1)]
First,
we require the weight of a dotted trivalent stable vertex
of genus $0$ to be $\tF_{0,3}$,
and the weight of a dotted stable vertex
of genus $1$ and valence $1$ to be $\tF_{1,1}$.
\item[(2)]
The weight of a dotted internal edge is $-\kappa$.
\item[(3)]
If the weight of a dotted stable vertex of genus $g$ and valence $n$ is $\tF_{g,n}$,
then we define the weight of a dotted stable vertex of genus $g$ and valence $n+1$
to be:
\ben
\tF_{g,n+1}:=
D_t(\tF_{g,n})=\big(\pd_t-\frac{n}{2\sqrt{-1}}\cdot
\tF_{0,3}\cdot(Im(\tau))^{-1}\big)\tF_{g,n},
\een
where $D_t$ is the covariant derivative associated to the special K\"ahler structure.
\item[(4)]
For $g\geq 2$,
we define the weight of a dotted stable vertex of genus $g$ and valence $0$ to be
(eg. see the recursive definitions \eqref{eqn:F2} and \eqref{eqn:F3}):
\ben
\qquad
\tF_{g}=\tF_{g,0}:=-\sum_{\Gamma^\vee\in\cG_{g,0}^{^\vee,c}\backslash\{V_g^\vee\}}
\frac{1}{|\Aut(\Gamma^\vee)|}
\prod_{v^\vee\in V(\Gamma^\vee)}\tF_{g_{v^\vee},\val_{v^\vee}}
\prod_{e^\vee\in E(\Gamma^\vee)}(-\kappa)+F_g.
\een

\end{itemize}

The above procedure gives us a recursive definition of
a sequence of functions $\{\tF_{g,n}(t,\kappa)\}_{2g-2+n>0}$,
thus we indeed obtain a realization of the dual abstract QFT.
Following \cite[\S 7.4]{gkmw},
let us impose some other conditions on this realization.
First,
the connection $D_t$ can be split as
(see \cite[(7.41),(7.42)]{gkmw}):
\be\label{eq-holoconn}
D_t =\check{D}_t+n\kappa\cdot\tF_{0,3},
\ee
where $\check{D}_t =\pd_t-n\cE\tF_{0,3}$ is the holomorphic connection,
which maps a holomorphic section to a holomorphic section.
We also require
(see \cite[7.47]{gkmw}):
\be\label{eq-dt-propagator}
D_t(\kappa)=-\kappa^2\cdot \tF_{0,3}+\cE_4\cdot\tF_{0,3},
\ee
where $\cE_4$ holomorphic,
and
\be\label{eq-Dt-Leibniz}
\begin{split}
D_t(\omega_{\Gamma^\vee})=&
\sum_{v^\vee\in V^\vee(\Gamma^\vee)}D_t(\omega_v)\cdot
\prod_{v_1^\vee\not=v^\vee}\omega_{v_1}
\prod_{e^\vee\in E^\vee(\Gamma^\vee)}\omega_e\\
&+\sum_{e^\vee\in E^\vee(\Gamma^\vee)}D_t(\omega_e)\cdot
\prod_{e_1^\vee\not=e^\vee}\omega_{e_1}
\prod_{v^\vee\in V^\vee(\Gamma^\vee)}\omega_v
\end{split}
\ee
for every dotted stable graph $\Gamma^\vee$.

Now following \S\ref{sec:Indep} we know that,
if this realization satisfies the Independence Assumption,
then we automatically obtain the quadratic recursion relation
described in Theorem \ref{thm-realization-HAE}.
We will check the Independence Assumption in this case
in the next subsection.

\subsection{Checking the Independence Assumption}

Let us first briefly summarize what we have done in this section.

In physics literatures (eg. BCOV \cite{bcov2}),
the mathematical definition of the non-holomorphic free energies $F^{(g)}$ is not clear.
In \cite{kz},  Klemm and Zaslow solve the holomorphic anomaly equations
to obtain expressions and Feynman rules \eqref{eq-kz-fr}
for the solutions $F^{(g)}$  in the case of $g=2,3$.
For the convenience of the reader,
we have recalled their results in \eqref{eqn:F2} and in \eqref{eqn:F3} in \S \ref{sec:KleZas}.
In the previous subsection,
we have constructed our formulation of the holomorphic
anomaly equation based on the duality theory developed in this paper:
we have started with functions $\tau$, $\tF_{0,3}$, $\kappa$, $\tF_{1,1}$
together with a sequence of holomorphic functions $\{F_g\}_{g\geq 2}$,
and recursively defined a realization $\{\tF_{g,n}(t,\kappa)\}$
of the dual abstract QFT.
In order to obtain the holomorphic anomaly equations for $\tF_g=\tF_{g,0}$,
we only need to check the Independence Assumption.
We will do this in the present subsection.

Let $\{F_{g,n}\}_{2g-2+n>0}$ be the induced realization  of the abstract QFT.
We need to show that $F_{g,n}$ is independent of $\kappa$ for every $2g-2+n>0$.
Since $\kappa$ is non-holomorphic,
we only need to show that $F_{g,n}$ is holomorphic for every $2g-2+n>0$.

\begin{Example}
Let us check some examples of $F_{g,n}$.
First, we have $F_{0,3}=\tF_{0,3}$,
which is clearly holomorphic.

Now consider $F_{0,4}$.
We have:
\ben
F_{0,4}&=&
\tF_{0,4}-3\kappa \tF_{0,3}^2\\
&=&D_t\tF_{0,3}-3\kappa \tF_{0,3}\cdot\tF_{0,3}\\
&=&(D_t-3\kappa \tF_{0,3})\tF_{0,3}\\
&=&\check{D}_t \tF_{0,3},
\een
thus it is also holomorphic.

Now let us compute $F_{0,5}$.
First, we have:
\be\label{eq-f05}
F_{0,5}=\tF_{0,5}-10\kappa\tF_{0,3}\tF_{0,4}+15\kappa^2\tF_{0,3}^3,
\ee
where $\tF_{0,5}$ can be rewritten as:
\ben
\tF_{0,5}&=&D_t \tF_{0,4}\\
&=&D_t(F_{0,4}+3\kappa\cdot\tF_{0,3}^2)\\
&=&\check{D}_t F_{0,4}+4\kappa\tF_{0,3}\cdot F_{0,4}
+6\kappa \tF_{0,3}\tF_{0,4}+3\tF_{0,3}^2(-\kappa^2\tF_{0,3}+\cE_4\cdot \tF_{0,3})\\
&=&\check{D}_t F_{0,4}+10\kappa\tF_{0,3}\tF_{0,4}
-15\kappa^2\tF_{0,3}^3+3\cE_4\cdot\tF_{0,3}^3.
\een
Plug this into \eqref{eq-f05},
then we have:
\ben
F_{0,5}=\check{D}_t F_{0,4}+3\cE_4\cdot\tF_{0,3}^3.
\een
Since $F_{0,4}$ is holomorphic,
we know that $F_{0,5}$ is also holomorphic.

Similarly,
we can check:
\ben
F_{0,6}&=&\check{D}_t F_{0,5}+10\cE_4\tF_{0,3}^2\tF_{0,4}-30\kappa\cE_4\cdot \tF_{0,3}^4\\
&=&\check{D}_t F_{0,5}+10\cE_4\tF_{0,3}^2\cdot(\tF_{0,4}-3\kappa\tF_{0,3}^2).
\een
Since $F_{0,5}$ and $\tF_{0,4}-3\kappa\tF_{0,3}^2=F_{0,4}$ are both holomorphic,
we know that $F_{0,6}$ is also holomorphic.

Now let us see some examples of genus one.
By definition,
the function
\ben
F_{1,1}=\tF_{1,1}-\half\kappa\tF_{0,3}=\pd_t f^{(1)}
\een
is holomorphic.

Now we compute $F_{1,2}$.
First, we have
\be\label{eq-f12}
F_{1,2}=\tF_{1,2}-\half\kappa \tF_{0,4}-\kappa\tF_{0,3}\tF_{1,1}
+\kappa^2\tF_{0,3}^2,
\ee
where $\tF_{1,2}$ can be rewritten as:
\ben
\tF_{1,2}&=&D_t\tF_{1,1}\\
&=&D_t(F_{1,1}+\half\kappa\tF_{0,3})\\
&=&\check{D}_t F_{1,1}+\kappa\tF_{0,3}\cdot F_{1,1}
+\half\kappa\tF_{0,4}+\half\tF_{0,3}(-\kappa^2\tF_{0,3}+\cE_4\cdot \tF_{0,3})\\
&=&\check{D}_t F_{1,1}+\kappa\cdot \tF_{0,3}\tF_{1,1}
+\half\kappa\tF_{0,4}-\kappa^2\tF_{0,3}^2+\half\cE_4\cdot \tF_{0,3}^2,
\een
(see also \cite[(7.48)]{gkmw}).
Plug this into \eqref{eq-f12}, then we have:
\ben
F_{1,2}=\check{D}_t F_{1,1}+\half\cE_4\cdot \tF_{0,3}^2.
\een
Since $F_{1,1}$ is holomorphic,
we know that $F_{1,2}$ is also holomorphic.
\end{Example}

From the above computations,
one might see that:
\begin{Lemma}
\label{lem-holo}
Assume that $2g-2+n>0$,
then we have:
\ben
F_{g,n+1}=\check{D}_t F_{g,n}+
\cE_4\tF_{0,3}\cdot A_{g,n}.
\een
where $A_{g,n}$ is given by:
\ben
A_{g,n}=
\half F_{g-1,n+2}+\half
\sum_{\substack{g_1+g_2=g,n_1+n_2=n,\\n_1\geq 1, n_2\geq 2}}
\frac{n!}{(n_1-1)!\cdot(n_2-1)!}\cdot F_{g_1,n_1}\cdot F_{g_2,n_2}.
\een
\end{Lemma}
\begin{proof}
We have $F_{g,n+1}=\tpd F_{g,n}$,
where $\tpd$ is the realization of the operator $\pd$.
Recall that by Theorem \ref{thm1} we have:
\ben
\pd=\cD^\vee=\pd^\vee+\gamma^\vee,
\een
where $\pd^\vee$ is realized by $D_t$ when acting on
the weights of dotted stable vertices (i.e., on $\tF_{g,n}$);
and by $\pd^\vee(\kappa^\vee)=(\kappa^\vee)^2\cdot\tF_{0,3}$
when acting on the weight of a dotted internal edge.
The operator $\gamma^\vee$ is realized by
multiplying by $n\kappa\cdot\tF_{0,3}$ when acting on $\omega_{\Gamma^\vee}$
for a dotted stable graph $\Gamma^\vee$ with $n$ external edges,
therefore by \eqref{eq-holoconn} and \eqref{eq-dt-propagator}
we know that the difference between $\check{D}_t$
and the realization of $\pd^\vee+\gamma^\vee$ equals:
\ben
\cE_4\cdot\tF_{0,3}\cdot \widetilde{K}^\vee (F_{g,n}),
\een
where $\widetilde{K}^\vee$ is the realization of the
dual edge-cutting operator $K^\vee$.
Recall that $F_{g,n}$ is the following summation over dotted stable graphs:
\ben
F_{g,n}=n!\cdot\sum_{\Gamma^\vee\in\cG_{g,n}^{\vee,c}}\frac{1}{|\Aut(\Gamma^\vee)|}
\prod_{v^\vee\in V^\vee(\Gamma^\vee)}\tF_{g_{v^\vee},\val_{v^\vee}}\cdot
\prod_{e^\vee\in E^\vee(\Gamma^\vee)}(-\kappa),
\een
thus by the dual version of Theorem \ref{original-rec-2}
we know that
\ben
\widetilde{K}^\vee\big(\frac{F_{g,n}}{n!}\big)=
\binom{n+2}{2}\cdot\frac{F_{g-1,n+2}}{(n+2)!}
+\frac{1}{2}\sum_{\substack{g_1+g_2=g\\n_1+n_2=n+2,\\n_1\geq 1,n_2\geq 1}}
\big(n_1\cdot \frac{F_{g_1,n_1}}{n_1!}\big)\big(n_2\cdot\frac{F_{g_2,n_2}}{n_2!}\big),
\een
which completes the proof.
\end{proof}

Notice that $\check{D}_t$ is a holomorphic connection,
and $A_{g,n}$ is a polynomial in
$\{F_{r,h}\}_{r<g}$ and $\{F_{g,m}\}_{m<n}$.
Since we already know that $F_{0,3}$, $F_{1,1}$ and $\{F_{g}\}_{g\geq 2}$
are all holomorphic,
by Lemma \ref{lem-holo} we easily obtain the following theorem
by induction:

\begin{Theorem}
The realization $\{\tF_{g,n}\}$ introduced in \S \ref{sec:KleZas}
satisfies the Independence Assumption.
\end{Theorem}

This theorem automatically gives us:
\begin{Corollary}
Let $\{\tF_{g,n}\}$ be the realization introduced in \S\ref{sec:KleZas},
then we have the holomorphic anomaly equation:
\be
\pd_\kappa\wF_{g,n}=\frac{1}{2}\biggl(
D_t D_t\wF_{g-1,n}
+\sum_{\substack{g_1+g_2=g,\\n_1+n_2=n}}
D_t \wF_{g_1,n_1} \cdot D_t \wF_{g_2,n_2}\biggr),
\quad 2g-2+n>0,
\ee
where we denote $\wF_{g,n}:=\frac{1}{n!}\tF_{g,n}$.
In particular, by taking $n=0$ we have:
\be
\pd_\kappa\tF_g=\frac{1}{2}\biggl(D_t D_t\tF_{g-1}+
\sum_{r=1}^{g-1}D_t\tF_{r}\cdot D_t\tF_{g-r}\biggr),
\quad g\geq 2,
\ee
where we use the convention $D_t\tF_1:=\tF_{1,1}$.
\end{Corollary}

\vspace{.3in}
{\em Acknowledgements}.
The second author is partly supported by NSFC grants 11661131005 and 11890662.

\end{document}